\documentclass[12pt]{article}
\usepackage[utf8]{inputenc}
\usepackage[LGR, T1]{fontenc}
\usepackage{graphicx}
\usepackage{subcaption}
\usepackage{cancel}
\usepackage{float}
\usepackage[inline, shortlabels]{enumitem}
\usepackage[reqno]{amsmath} 
\usepackage{amsthm}
\usepackage{mathtools}
\usepackage{amssymb}
\usepackage{amsfonts}
\usepackage{eurosym}
\usepackage{geometry}
\usepackage{dsfont}
\usepackage{setspace}
\usepackage{natbib}
\usepackage{sectsty} 
\usepackage[colorlinks]{hyperref}
\hypersetup{citecolor=blue,linkcolor=blue, urlcolor=blue}
\usepackage{pgfplots}
\usepackage{tikz}
\usetikzlibrary{datavisualization}
\usetikzlibrary{datavisualization.formats.functions}
\usepackage{changepage}
\usepackage{adjustbox}

\usepackage{array}
\newcolumntype{C}[1]{>{\centering\arraybackslash}p{#1}}

\theoremstyle{plain} 
\newtheorem{theorem}{Theorem}
\newtheorem{assumption}{Assumption}

\newtheorem*{normalization}{Normalization}
\newtheorem{proposition}{Proposition}
\newenvironment{assumptionbis}[1]
  {\renewcommand{\theassumption}{$D$\ref{#1}}%
   \addtocounter{assumption}{-1}%
   \begin{assumption}}
  {\end{assumption}}

\newtheorem{Lemma}{Lemma}

\newcommand*{\myfont}{\fontfamily{qpl}\selectfont}
\DeclareTextFontCommand{\textmyfont}{\myfont}
\makeatletter
\newcommand{\leqnomode}{\tagsleft@true}
\newcommand{\reqnomode}{\tagsleft@false}
\makeatother

\def\mysingleq#1{`#1'}

\DeclareSymbolFont{upgreek}{LGR}{cmr}{m}{n}
\SetSymbolFont{upgreek}{bold}{LGR}{cmr}{bx}{n}
\DeclareMathSymbol{\Epsilon}{\mathalpha}{letters}{"0F}
\DeclareMathSymbol{\Eta}{\mathalpha}{letters}{"11}
\DeclareMathSymbol{\epsilon}{\mathalpha}{letters}{`e}
\DeclareMathSymbol{\eta}{\mathalpha}{letters}{`h}

\usepackage{titlesec}
\titlespacing*{\section}{0pt}{19pt}{7pt}
\titlespacing*{\subsection}{0pt}{14pt}{5pt}
\titlespacing*{\subsubsection}{0pt}{12pt}{5pt}
\titlespacing*{\paragraph}{0pt}{9pt}{9pt}
\titleformat{\section}{\normalfont\fontsize{16}{15}\bfseries}{\thesection}{1em}{}
\titleformat{\subsection}{\normalfont\fontsize{14}{15}\bfseries}{\thesubsection}{1em}{}
\titleformat{\subsubsection}{\normalfont\fontsize{12}{15}\bfseries}{\thesubsubsection}{1em}{}

\allowdisplaybreaks

\begin{document}

\title{{\fontsize{14}{20} \selectfont 
\textbf{\textmyfont{ 
Dynamic Discrete-Continuous Choice Models: \\
Identification and Conditional Choice Probability Estimation}}}}
\author{{\fontsize{12}{20} \textbf{Christophe Bruneel-Zupanc}}\footnote{E-mail address: \href{mailto:christophe.bruneel@gmail.com}{christophe.bruneel@gmail.com}.  
I am especially grateful to Thierry Magnac and Olivier De Groote for their supervision and guidance. I am also grateful to Geert Dhaene for his detailed remarks. I thank Peter Arcidiacono, Jad Beyhum, Matteo Bobba, Stéphane Bonhomme, Catherine Casamatta, Edoardo Ciscato, Fabrice Collard, Giovanni Compiani, Paul Diegert, Pierre Dubois, Patrick Fève, Chao Fu, Alexandre Gaillard, Ana Gazmuri, Cristina Gualdani, Jasmine Hao, Vishal Kamat, Jihyun Kim, Dennis Kristensen, Pascal Lavergne, Arnaud Maurel, Nour Meddahi, Robert Miller, Pedro Mira, François Poinas, Shruti Sinha and Frank Verboven for helpful comments and discussions. I also thank the participants in the internal workshop at the Toulouse School of Economics. This research has received funding from ANR under grant ANR-17-EURE-0010 (Investissements d'Avenir program). I also acknowledge financial support from KU Leuven grants STG/21/040 and C14/24/013.} \\
{\fontsize{12}{20} \textit{Department of Economics, KU Leuven}}} 
\date{{\fontsize{12}{20} \textit{\today} }}
\maketitle

\vspace{-0.2in}

\setstretch{1}
\setlength{\abovedisplayskip}{0pt} 
\setlength{\belowdisplayskip}{0pt}

\begin{abstract}
\noindent This paper develops a general framework for dynamic models in which individuals simultaneously make both discrete and continuous choices. 
The framework incorporates a wide range of unobserved heterogeneity. 
I show that such models are nonparametrically identified. 
Based on constructive identification arguments, I build a novel two-step estimation method in the lineage of \cite{hm1993} and \cite{am2011} but extended to simultaneous discrete-continuous choice. In the first step, I recover the (type-dependent) optimal choices with an expectation-maximization algorithm and instrumental variable quantile regression.  In the second step, I estimate the primitives of the model taking the estimated optimal choices as given. The method is especially attractive for complex dynamic models because it significantly reduces the computational burden associated with their estimation compared to alternative full solution methods. \\ 

\end{abstract}

\setstretch{1.30} 
\setlength{\abovedisplayskip}{6pt} 
\setlength{\belowdisplayskip}{6pt}

\pagebreak

\section{Introduction}

%
%

Many economic problems involve joint discrete and continuous choices. 
For example, a firm can decide what to produce and the corresponding sale price \citep{crawford2019}. Firms also decide whether to register their business and how many workers to hire \citep{ulyssea2018}. 
Students select their majors and decide how much effort to exert in their study \citep{arcidiaconoetal2019}.
Consumers decide what to buy and how much to consume \citep[e.g., appliance choice and demand for energy, ][]{dubin1984}. 
In housing, buyers decide on their house size and housing tenure \citep{hanemann1984, bckm2013}. 
The buyer of a car selects a model and the mileage of the car \citep{bento2009}. 
Individuals decide whether to retire or not and how much they plan to consume accordingly \citep{iskhakov2017}. 
Similarly, labor force participation and consumption/savings are joint choices for potential workers \citep{altugmiller1998, bcms2016, arellano2017earnings}. \\ 
\indent In all these examples, a rational individual makes both decisions simultaneously. 
As a result, the discrete choice is endogenous with respect to the continuous choice and vice versa. Taking the labor and consumption problem as the leading example throughout the paper, if an individual works, she consumes differently than if she does not work: she has two different conditional consumption choices. Moreover, her decision to work or not is dependent on these two conditional continuous choices. 
%
%
%
%
Unfortunately, the identification of models with simultaneous choices is difficult \citep{matzkin2007}. Indeed, there is a core observability problem because we only observe the continuous choice made in the selected discrete alternative, and we do not know the counterfactual choices the individual would have made in the other alternatives. 
Ideally, we would like to recover counterfactual continuous choices using the choices of individuals with similar characteristics but who chose another alternative. However, doing so is not possible if individuals also differ on factors which are unobserved by the econometrician and affect both continuous and discrete choices. 
In this case, two identical individuals as measured by their observed covariates might still differ along the unobserved dimension. There is likely a problem of \textit{selection on unobservables}, which prevents the identification of counterfactual continuous choices. 
To further pursue the example, if a researcher observes that working individuals consume more than unemployed individuals, she cannot identify whether this is because the consumption choice conditional on working is truly higher or because individuals with an unobserved higher taste for consumption select themselves more into working. \\
%
%
%
%
%
%
\indent This paper develops a general framework of \textit{dynamic simultaneous discrete-continuous choice} models suited for dynamic problems including a wide range of unobserved heterogeneity with transitory period-specific shocks and unobserved permanent types. I show how nonparametric identification of these models can be obtained by combining and extending insights from both, the dynamic discrete choice model literature \citep{hm1993, kasaharashimotsu2009, am2011} and the reduced-form literature on quantile treatment effects identification \citep{chernozhukovhansen2005, vuongxu2017}. 
Then, building upon the identification, I provide a two-step estimation method for these models. The method is attractive because it yields significant computational gains regarding the estimation of dynamic discrete-continuous choice models, in the lineage of \cite{hm1993} for dynamic discrete choice models. \\ 
%
\indent The first contribution of this paper is that I provide a constructive proof of the nonparametric identification of a general class of structural dynamic models in which individuals simultaneously make a discrete and a continuous choice. 
First, I identify the optimal discrete and continuous choice policies directly from the data, and then, taking these policies as given, I identify the primitives of the structural model. \\
\indent The identification of the optimal choices proceeds in two sub-stages, each handling one of the two unobserved endogenous shocks present in the framework. The framework includes both (i) permanent unobserved types, capturing intrinsic latent differences between individuals and (ii) transitory shocks which only affect the individuals in a given time-period. To pursue the previous example, the types capture intrinsic differences in individuals' preferences for consumption and labor, in addition to individual-specific transitory taste shocks to consumption and labor every periods. 
\noindent In the first sub-stage, provided that the panel is long enough (more than $6$ observations for each individuals), I show how to identify the unobserved types from the complete panel of individual choices. To do so, I extend the dynamic discrete choice identification proof of \cite{kasaharashimotsu2009} to joint discrete and continuous choices with time-dependent joint densities and lagged dependent variables. 
\noindent Then, given the identified types, 
I show that the endogeneity of the discrete choice with respect to the transitory shocks can be handled using the lagged discrete choice as an instrumental variable (IV) to nonparametrically identify the optimal continuous and discrete choices. 
Indeed, provided that there are some switching costs in the discrete choice (e.g., switching costs of changing labor decision), the previous discrete choice affects the current discrete decision. However, in most models, conditional on the current discrete choice (and on the types and other current covariates), the lagged discrete choice has no direct effect on the current continuous choice. Thus, the lagged discrete choice is often a valid instrument, \textit{relevant} for the current discrete choice (treatment), \textit{excluded} from the current continuous choice (outcome), and exogenous with respect to the transitory shock. 
In this way, observable differences in the distribution of the choices due to variations in the instrument can be attributed to unobserved differences in selection, and not to differences in the continuous choices. I show that, paired with restrictions on the effect of unobserved heterogeneity on the continuous choice (monotonicity, rank invariance), the instrument allows us to establish nonparametric identification of the optimal discrete and continuous choices. The proof relates to and extends reduced form results on the nonparametric identification of quantile treatment effects \citep{chernozhukovhansen2005, vuongxu2017} with IVs. Indeed, I show that identifying the optimal choices in each period can be framed as identifying the effect of the discrete choice (endogenous treatment) on the continuous choice (outcome). This link is appealing as it grounds the identification of dynamic structural models in the treatment effect literature, making it less reliant on sometimes arbitrary structural assumptions (e.g., timing of the choices, discretization of the continuous choice, or implicit exogeneity assumptions between the choices). \\
%
%
%
%
%
\indent Once the optimal choices are identified, I show how to use them to identify the primitives of the structural model. The key lies in linking these choices to the first-order conditions of the true structural model (e.g., the Euler equation determines the optimal continuous choices). Thus, one can reverse engineer the identified choices to identify the true primitives that generated them \citep{hm1993, bmm1997, escancianoetal2021}. \\
\indent The second contribution of the paper is in terms of estimation. I build a two-step estimation method, similar to \cite{hm1993} and \cite{am2011}, 
but for discrete and continuous choices. In the first step, one estimates the policies, which I name after Hotz-Miller's CCPs: \textit{conditional continuous choices} (\textit{CCCs}) and \textit{conditional choice probabilities} (\textit{CCPs}). This step builds on the identification arguments. The policies are estimated directly from the data without solving the structural model. To account for unobserved types, I use an expectation-maximization (EM) algorithm, in the spirit of \cite{am2011}. Then, I estimate the CCCs and CCPs building on IV quantile regression (IVQR) literature \citep{chernozhukovhansen2006, kaido2021decentralization}, using the lagged discrete choice as an instrument which is valid conditional on the estimated types (and covariates). 
In the second step, one uses the estimated CCCs and CCPs to estimate the structure of the model. More specifically, I exploit the fact that within my framework, the primitives of the model are related to optimal choices through the first-order conditions. Given the estimated optimal choices, one can estimate the primitives of the model that generated these choices by satisfying these optimality conditions. 
The two-step estimation method is attractive because it yields sizeable computational gains. Typical dynamic discrete \textit{or} continuous choice models are difficult to estimate because they involve solving the theoretical model (either by backward recursion or fixed point algorithms). Dynamic discrete-continuous choice models are even more difficult to estimate because the mixed choices can introduce kinks and non-concavities in the value function \citep{iskhakov2017}.  
Given that I can recover the CCCs and CCPs in the first step, I can exploit them to estimate the rest of the model without having to compute the value function or solve the model.\footnote{Since I do not solve for the CCCs and CCPs using an optimization algorithm, there is also no concerns about kinks and non-concavities in the value function that would make the estimation of these optimal choices more complicated. } This yields computational gains comparable to those obtained by \cite{hm1993} in the dynamic discrete choice literature, achieving estimation times already hundreds of times faster than the best available alternative \citep{iskhakov2017} in a simple toy model (see Section \ref{section_comparison_estimation}), and even greater improvements in more complex settings.    
The gains are so important that they not only reduce the time required to estimate the models, but also make it possible to estimate models that have thus far been deemed computationally intractable. 
In this respect, my method may facilitate the use of simultaneous discrete-continuous choice models, in particular the estimation of single-agent partial equilibrium life-cycle dynamic models. \\ 
\indent Overall, the method builds a bridge between more reduced-form policy estimation and dynamic structural models. By enabling the estimation of structural models directly using reduced-form estimates, the method unlocks the possibility of doing counterfactual policy analysis on the basis of reduced-form results. 
In the leading example of the consumption and labor choice problem, the method described in this paper shows how to estimate standard life-cycle structural models of consumption and labor choices \citep[e.g.,][]{bcms2016} based on reduced-form/semi-structural estimates of optimal labor-specific consumption rules \citep[e.g.,][]{arellano2017earnings}. With the structural model deep parameters (e.g., risk aversion), one can run many counterfactual policy analysis, varying tax/subsidies on labor or consumption for example. \\






\noindent \textbf{Related literature.} \\
\noindent There is a vast empirical literature that uses dynamic discrete choice models, for example, in studies of labor market transition and career choice \citep{keanewolpin1996}, fertility choice \citep{ecksteinwolpin1989} and education choice \citep{arcidiacono2004}. Starting from the bus replacement problem of \cite{rust1987}, developments have been made regarding the estimation and identification of these models, including \cite{hm1993}, \cite{hmss1994}, \cite{rust1994}, \cite{mt2002}, \cite{aguirregabiriamira2002}, \cite{aguirregabiriamira2007}, \cite{kasaharashimotsu2009}, \cite{am2011}, \cite{hushum2012}, \cite{am2019}, \cite{am2020}, \cite{abbringdaljord2020}, and \cite{berry2023instrumental} among others.  For a survey, see \cite{aguirregabiriamira2010} or \cite{arcidiaconoellickson2011}. \\
\indent Similarly, the literature on dynamic continuous choice models is also voluminous, especially concerning consumption/saving \citep{carroll2006} or investment choices \citep{hs2010}. 
There are also methods such as \cite{bbl2007} that can be applied to either dynamic discrete choice models or dynamic continuous choice models (but not both).\footnote{More precisely, \cite{bbl2007} describe problems with discrete or continuous policy functions separately. 
Extending their estimation techniques to more general Framework with discrete and continuous choices and with several unobservables yielding endogeneity of both choices would require identifying the first stage optimal policies following the approach described in this paper first. }  \\ 
%
%
\indent However, many economic problems involve multiple joint decisions, not only one discrete choice or only one continuous choice. For example, labor force participation is very much related to saving decisions. By focusing only on one of these two dimensions and ignoring the other (endogenous) choice, one might be missing something important. 
%
%
%
%
%
%
%
%
%
Unfortunately, empirical applications of the dynamic discrete-continuous choice framework are less common, as there was no general identification result available. For example, \cite{bmm1997} provide identification of such models once the optimal choices are identified but do not directly address the identification of these choices. 
The existing literature employs several tricks to overcome the problem of selection on unobservables. 
The most extreme is to assume away the problem by assuming selection on observables only, i.e., conditional on the observed covariates, assume that there is no other unobservable affecting the optimal choices. This is fairly strong, especially in dynamic models where the number of covariates is typically limited. Without ruling out the existence of these unobservables, another common approach is to have implicit or explicit assumptions about the selection process, through assumptions about the relation between the error terms affecting both choices, e.g., independence, measurement errors or known joint distribution \citep{dubin1984, hanemann1984, bento2009}. 
Another common technique is to discretize the continuous choice so that the discrete-continuous model can be rewritten as a discrete choice model \citep{degrooteverboven2019}. This is appealing, as it allows the application of known techniques in the dynamic discrete choice literature. However, discretizing the continuous choice is implicitly equivalent to making an assumption about the selection process via an assumption on the distribution of the additive discrete error terms. 
Another approach is to resort to timing assumptions which implicitly break the endogeneity of the choices. \cite{blevins2014} shows nonparametric identification of dynamic discrete-continuous choice models assuming a specific timing in which the discrete choice takes place before the realization of the nonseparable shocks affecting the continuous choice: hence the selection (discrete choice) does not depend on the nonseparable shock. \cite{iskhakov2017} and \cite{murphy2018} use similar timing assumptions, which are effectively equivalent to imposing that the discrete choice is exogenous. 
A more convincing alternative is to allow for endogeneity but reduce the level of unobserved heterogeneity, for example, by including only a finite number of unobserved types \citep{bcms2016}.
My approach is more general, as I allow for a more flexible distribution of unobserved heterogeneity with both period-specific transitory shocks and permanent unobserved types. I handle the complex endogeneity of the discrete and continuous choices by extending techniques from both the dynamic discrete choice literature to identify the types \citep{kasaharashimotsu2009, hushum2012}, and from the reduced form quantile treatment effect literature to handle the intra-period endogeneity \citep{chernozhukovhansen2005}. Linking the identification of structural models with nonparametric treatment effect identification results is appealing as these results rely less on sometimes arbitrary structural assumptions (timing, discretization, exogeneity, distribution of the errors, ...). Furthermore, my identification allows to test these assumptions. \\
\indent Most closely related to this paper, contemporaneous work by \cite{levy2024identification} also addresses the identification of simultaneous discrete-continuous dynamic choice models with rich unobserved heterogeneity in two steps: first they identify the optimal choices, then the model primitives. 
The main difference between our papers is the manner in which they handle the endogeneity to identify the optimal policies in the first stage. They focus on infinite horizon setups with a stationary environment and address the selection by requiring the existence of a (sequence of) variable(s) such that the probability of selecting some specific discrete alternatives becomes arbitrarily high (tends to one). In practice, however, the existence of such a variable is hard to satisfy in most applications. To provide an analogy with the treatment effect identification literature, their identification arguments are similar to identification-at-infinity arguments, requiring the existence of a "infinitely relevant" instruments, such that the selection probability tends to one. Instead, I only need a weaker standard relevant instrument to address endogeneity and identify the optimal choices. While even standard IVs may be hard to find in the context of dynamic models, I show that in many structural models, the past discrete choice will be a valid instrument as soon as there are nonzero switching costs (conditional on the covariates and the types). This is a relatively mild condition in many applications, and will be testable with my framework. 
In the special case of \cite{levy2024identification}'s application to retirement and consumption decisions, and more generally in the presence of absorbing states in the discrete choice, our identification arguments coincide. Indeed, retirement is an absorbing state, so the probability of being retired today conditional on being previously retired is one. Consequently, the previous retirement status satisfy their identification-at-infinity condition, and is also an (infinitely) relevant instrument in my case (scenario equivalent to infinitely high switching cost). In fact, I show that when the discrete choice has an absorbing state, my identification arguments are considerably simplified and focussing on individuals who are already in the absorbing state has additional identification power (see Section \ref{subsection_dynamic_identification}). 
In addition to these, the difference with my paper is that they use pairwise differencing for the estimation of their primitives and require separability of the unobservables in the marginal utilities to do so, while I do not need it. I also take into account auto-correlated shocks through permanent types, making a link with the dynamic discrete choice literature.  \\
\indent As already mentioned, this paper builds a general framework that connects the identification of structural models with the reduced form nonparametric identification literature \citep{neweypowell2003, chesher2003, newey2007, matzkin2007, matzkin2008, imbensnewey2009, torgovitsky2015, dhaultfoeuillefevrier2015}. 
By casting the optimal choices in the form of a triangular simultaneous system of equations, I show how their identification can be framed as the identification of quantile treatment effects with an endogenous treatment, i.e., the IV quantile regression (IVQR) Framework \citep{chernozhukovhansen2005, vuongxu2017, feng2024matching}, where the discrete and continuous choices can be understood as the treatment and the outcome, respectively. Moreover, I improve on the existing results of \cite{chernozhukovhansen2005} by weakening their relevance condition: instead of their global full rank condition, I show that the identification can be obtained under weaker, testable, and easier to interpret relevance condition. 
More precisely, I need that the instrument is relevant almost everywhere, except possibly at a finite set of isolated values of the unobservable shocks. Allowing for some isolated points of irrelevance is important, especially in dynamic models where the discrete choice has many alternatives, as these locally irrelevant points may often occur, even in simple models.  
The reason why I can relax the full rank condition of \cite{chernozhukovhansen2005} is that they do not exploit a key property of their quantile model: the fact the continuous choices (outcomes) are strictly increasing in their unobservable shocks (ranks). This monotonicity has extra power in terms of identification. To the best of my knowledge, \cite{vuongxu2017} are the only others who also exploit the power of monotonicity to relax the full rank condition of \cite{chernozhukovhansen2005} and still identify quantile treatment effects, but only in the context of a binary treatment. I further show that this weaker relevance condition can be expressed as an easy-to-interpret conditions on the conditional choice probabilities (depending on the unobservable shock affecting the continuous choices), which is testable. \\
\indent Similarly, I contribute to the literature on the identification of models with unobserved types \citep{kasaharashimotsu2009, hushum2012, higgins2023identification}. In particular, I extend the identification of \cite{kasaharashimotsu2009} to type-dependent joint discrete-continuous choice densities, where the densities are time-dependent and depend on lagged choices. I also show how to account for covariates  for which the transition is deterministic given the choices (e.g., assets), which violates standard assumptions in this literature. \\
\indent For the identification of the primitives of the model given the identified optimal choices, I build upon \cite{escancianoetal2021} and \cite{bmm1997}. I adapt \cite{escancianoetal2021} to my framework to identify the marginal utilities and the discount factor from the Euler equations. Then I adapt \cite{bmm1997} to identify the remaining primitives (value functions) using these marginal utilities.  \\
\indent I also contribute to the literature on fast estimation methods, avoiding the computation of the value function \citep{rust1987, hm1993, hmss1994, carroll2006, am2011, iskhakov2017}. I provide a faster alternative to indirect inference and the most recent developments of endogenous grid methods \citep{iskhakov2017}. A timing comparison of the different estimation methods is given in Section \ref{section_comparison_estimation}. \\
%
%
%
%
%
%
%
%
%
%
%
%

\noindent \textbf{Outline.} 
The Framework contains several building blocks, that we will develop backwards. First, Section \ref{section_framework} describes the \textit{intra-period} simultaneous discrete-continuous choice problem, for any given period $t$, and assuming the types are already identified. It also discusses nonparametric identification of the optimal choices within any period. Then, Section \ref{section_dynamic} shows the general dynamic models that yields these intra-period problems. Section \ref{section_types} shows how to identify the permanent types beforehand. \\
Building on the complete Framework and identification arguments, Section \ref{section_estimation} describes the estimation method and Section \ref{section_comparison_estimation} shows the estimator performances, in terms of precision and computational time, using Monte-Carlo simulations of a dynamic life-cycle model of consumption and labor force participation choices. Section \ref{section_conclusion} concludes.

\section{The intra-period problem}\label{allsection_framework}

This section describes the intra-period problem of a dynamic model and its nonparametric identification for any specific period $t$. This serves as a building block and the identification of the dynamic model will then be described in Section \ref{section_dynamic}. I also proceed conditional on the type, $m$, which should have been identified beforehand (see Section \ref{section_types}). I abstract from the period $t$ and type $m$ to simplify the notation.  The main text describes the framework with a binary discrete choice, extension and identification with more than $2$ discrete alternatives is in Appendix \ref{appendix_discrete}. 


\subsection{Intra-period Framework}\label{section_framework}
Consider an individual's decision problem with the following timing within a period: 
{\setlength{\baselineskip}{15pt}
\vspace{-10pt}
\begin{center}
\begin{tikzpicture}[x=2.5cm]
\draw[->,thick,>=latex]
  (0,0) -- (3,0) node[below right] {};

\draw[thick] (0,-.2) -- (0, .2) node[above] {};
\draw[thick] (1,-.2) -- (1, .2) node[above] {};
\draw[thick] (2.2,-.2) -- (2.2, .2) node[above] {};

\draw (0,0) node[align=center, below=6pt] {State $z$ \\};
\draw (1,0) node[align=center, below=6pt] {Shocks $(\epsilon, \eta)$ \\ occur};
\draw (2.2,0) node[align=center, below=6pt] {Individual \\ picks $(d, c_d)$};
\end{tikzpicture}
\end{center}
}
\vspace{-0.5em}

\noindent The individual simultaneously selects a discrete action $d$ $\in \mathcal{D} = \{0, 1\}$ and accordingly makes one continuous choice $c_d \in \mathcal{C}_d$, where $\mathcal{C}_d$ is a compact subset of $\mathbb{R}$, to maximize his payoff.\footnote{For a more general discrete choice with $\mathcal{D} = \{0, ..., J\}$,  see Appendix \ref{appendix_discrete}.} 
The decision is made given some state $z \in \mathcal{Z}$ observed by the researcher, as well as two transitory period $t-$specific preference shocks, $\epsilon = (\epsilon_0, \epsilon_1) \in \mathcal{E} \subset \mathbb{R}^2$ and $\eta \in \mathcal{H} \subset \mathbb{R}$. The shocks $\epsilon$ and $\eta$ are realizations of the random variables $\Epsilon = (\Epsilon_0, \Epsilon_1)$ and $\Eta$ and are unobserved by the researcher. The shock $\epsilon$ only affects the discrete choice $d$, while $\eta$ impacts the continuous choice $c$ and the discrete choice. The \textit{same} $\eta$ impacts the continuous choice decision in both discrete-choice states ($c_0$ and $c_1$), that is, there is \textit{rank invariance} \citep{hsc1997, chernozhukovhansen2005}.\footnote{The continuous choices could even represent different variables depending on the discrete option selected: for example, if $d$ represents the choice between working and studying, $c$ might represent the amount of time worked and the effort of the student respectively, hence with possibly different supports. The main restriction is that even if they represent two different choices, these two continuous choices are impacted by the \textit{same} unobserved shock $\Eta$. } 

The payoffs of the individual are given by the function $\mathcal{V}_d(c_d, z, \eta, \epsilon_d)$. 
The individual simultaneously selects $d$ and $c_d$ to solve:
\begin{equation}
\underset{d, c_d}{\textrm{max}} \quad \mathcal{V}_d(c_d, z, \eta, \epsilon_d).
\end{equation}


\noindent I require additional assumptions for tractability and identification of the model. 
\begin{assumption}[Additive Separability]\label{additive}
The shock $\epsilon$ enters the payoff additively:
\begin{align*}
\mathcal{V}_d(c_d, z, \eta, \epsilon_d) = \tilde{v}_d(c_d, z, \eta) + \epsilon_d, \quad \text{for } d = 0, 1.
\end{align*}
\end{assumption}

\indent The additive separability assumption is common in the discrete choice model literature \citep{rust1987, am2011}. It applies to $\epsilon$, while $\eta$ can still enter the payoff in a nonseparable manner. 
A consequence of Assumption \ref{additive} is that the optimal conditional policy functions given $d$, $c_d^*(\cdot)$, will not depend on $\epsilon$: 
\begin{align*}
c_d^* = \underset{c_d}{\textrm{argmax}} \  \mathcal{V}_d(c_d, z, \eta, \epsilon_d)  = \underset{c_d}{\textrm{argmax}} \ \tilde{v}_d(c_d, z, \eta), \quad \text{ for } d = 0, 1.
\end{align*}



\begin{assumption}[Instrument]\label{instrument} The state vector contains two kinds of variables, $z = (x, w)$, where $x \in \mathcal{X}$ and $w \in \mathcal{W} = \mathcal{D}$,\footnote{Given the IV I use (past discrete choice), the support of $W$ is the support of $D$. In general the support of $W$ must be larger than the support of $D$. For discrete or continuous $W$, the identification proof follows along the same lines.} and 
\begin{center} 
$\tilde{v}_d(c_d, z, \eta)= v_d(c_d, x, \eta) + m_d(x, w, \eta)$, \quad $\text{ for } d=0,1.$ 
\end{center}

\end{assumption}

Here, $x$ represents general state variables and $w$ is an `instrument' to recover the optimal conditional policies $c_d^*$. On the one hand, $w$ is \textit{excluded} from the optimal policies $c_d^*$ since 
\begin{align*}
c_d^*\  = \  \underset{c_d}{\textrm{argmax}} \  \tilde{v}_d(c_d, z, \eta) \  = \ \underset{c_d}{\textrm{argmax}} \ v_d(c_d, x, \eta), \quad \text{ for } d =0,1.
\end{align*}
On the other hand, $w$ might still be \textit{relevant} and impact the discrete choice $d$. 

\begin{assumption}[Monotonicity]\label{monotone}
The payoff functions $v_d$ are twice continuously differentiable and 
\begin{equation*}
\frac{\partial^2 v_d(c_d, x, \eta)}{\partial c_d \partial \eta} > 0, \quad \text{ for } d=0,1. \\
\end{equation*}
\end{assumption}

\indent Assumption \ref{monotone} implies that, given $D=d$ and $X=x$, the conditional optimal policy function $c_d^*(x, \eta)$ is continuously differentiable and strictly increasing in $\eta$. 
Hence $\eta$ and $c_d^*$ are one-to-one for every $d$ and $x$. 
This kind of monotonicity condition has been widely used for identification \citep{chernozhukovhansen2005, bbl2007, hs2010}. In a sense, it means that I only identify monotone effects of the unobserved nonseparable source of heterogeneity, $\eta$.  
An important limitation of Assumption \ref{monotone} is that it requires a nontrivial continuous choice $c_d$ for each discrete alternative $d$. 
For example, Assumption \ref{monotone} is not satisfied in the case where an investor decides whether to invest ($d=1$) or not ($d=0$) and the corresponding investment conditional on investing ($d=1$) \citep{hs2010}. Indeed, in this case, $c^*_0(x, \eta) = 0$ for all $\eta$ (and $x$), and $c_0^*$ is not strictly increasing in $\eta$. In contrast, Assumption \ref{monotone} holds in the case of a discrete choice between portfolios and the corresponding conditional level of investment. \\
%
%
%
%
%
\indent Under Assumptions \ref{additive}, \ref{instrument} and \ref{monotone} we obtain the following triangular structure for the reduced-form optimal choices: 
\begin{align*}
\left\{
    \begin{array}{l}
        C_d = c_d^*(X, \Eta), \\
        D = d^*(X, W, \Eta, \Epsilon).
    \end{array}
\right. 
\end{align*}
This triangular structure links my  structural model with the literature on (reduced-form) systems of simultaneous equations \citep{chesher2003, matzkin2008, imbensnewey2009} and, more specifically, the related literature on heterogeneous (quantile) treatment effects \citep{chernozhukovhansen2005, vuongxu2017}. 
To identify the structure, one needs to first identify the optimal choice functions. To identify them, I need additional assumptions on the shocks. 
%
%
%

\begin{assumption}[Shocks]\label{ass_shocks} Conditional on $X=x$, 
\begin{enumerate*}[label={\textbf{\upshape{(\roman*)}}}, ref={\theassumption(\roman*)}]
	\item\label{indep_shock} $W$, $\Eta$ and $\Epsilon$ are mutually independent;  
	\item\label{contshock} $\Eta$ is continuously distributed as $\mathcal{U}(0,1)$; 
	\item\label{discshock} $\Epsilon$ is continuously distributed with full support; 
	\item\label{maxu} $\underset{c}{\textrm{max}} \ \tilde{v}_d(c, x, w, \eta) < \infty$ for all $(x, w, \eta, d)$. 
\end{enumerate*}
\end{assumption}

\indent The main independence restriction is that $\Eta$ is independent of $W$ given $X$. The identification of $c_d^*$ requires $\Eta$ to have the same distribution, regardless of the realization of $W$. Other than this, the independence assumption is not as restrictive as it may appear. Indeed, note that the additive term $m_d(x, w, \eta)$ can be interpreted in two ways that cannot be separately identified. In Assumption \ref{instrument}, $m_d$ is an additive part of the payoff $\tilde{v}_d$. However, $m_d$ can also be interpreted as part of a more general additive discrete-choice shock, $\tilde{\epsilon}_d(x, w, \eta) = m_d(x, w, \eta) + \epsilon_d$, in which case $\Epsilon_d$ is the part of the discrete-choice shock that is independent of $\Eta$ and $W$. 
\noindent The continuity of the distribution of $\Eta$ is imposed to obtain smooth conditional distributions of the continuous choices. I cannot identify the distribution of $\Eta$ separately from the utility. Therefore, as is standard in the literature \citep{bmm1997, matzkin2003}, I normalize $\Eta$ to be uniformly distributed (given $X$). This normalization is innocuous. 
Formally, I nonparametrically identify the quantiles of the optimal choices 
and payoffs. 
Similar to the distribution of $\Eta$, the distribution of $\Epsilon$ is not nonparametrically identified in my setup, but this does not affect the nonparametric identification of the optimal choices nor of the payoff function, $v_d$, as long as $\Epsilon$ has full support. 
\noindent Assumption \ref{maxu} is a regularity condition on the functional form ensuring that $0 < \textrm{Pr}(D=d | \Eta=\eta, Z=z) < 1$ for all $(d, \eta, z)$. \\ 
\indent I need one last (testable) condition for identification. 

{ 
\begin{assumption}[Instrument Relevance]\label{relevance}\label{identification2}
For every $x \in \mathcal{X}$,  
\begin{align*}
\textrm{Pr}(D=0 | \Eta = \eta, X=x, W=1) \neq \textrm{Pr}(D=0 | \Eta = \eta, X=x, W=0), 
\end{align*} 
$\text{ for all } \eta \in \mathcal{H} \backslash \mathcal{K}_x,$ 
where $\mathcal{K}_x$ is a (possibly empty) finite set containing $K$ values ($K \geq 0$). 
\end{assumption}
}

\indent Identification of the optimal policies requires that the instrument is \textit{sufficiently relevant}. It needs to be relevant `almost everywhere', but I show that identification still holds even if there is a finite set of values of $\eta$ at which the instrument is not relevant, which could occur if the switching costs vary with $\Eta$.  
Assumption \ref{identification2} yields testable implications for the observed reduced forms distributions of $C$ and $D$. It allows to test whether the structural model is identified, as I discuss in the next section. 
Finally, note that Assumption \ref{identification2}, expressed in terms of the conditional choice probabilities, is equivalent to an assumption on the structural functions $m_d$. Indeed, 
\begin{align*}
&\textrm{Pr}(D = 0 | \Eta=\eta, X=x, W=w) \\
&= \textrm{Pr} \Big(  \Epsilon_0 - \Epsilon_1  > \underset{c}{\textrm{max}} \  v_1 (c, x, \eta) - \underset{c}{\textrm{max}} \ v_{0} (c, x, \eta)  \\
&\quad \quad \quad \quad \quad \quad  + m_1(x, w, \eta) - m_0(x, w, \eta)  \Big| \Eta=\eta, X=x, W=w \Big) .
\end{align*}
Since $\underset{c}{\textrm{max}} \  v_1 (c, x, \eta) - \underset{c}{\textrm{max}} \ v_{0} (c, x, \eta)$ is independent of $W$ and since $\Epsilon_d \perp (W, \Eta) | X = x$, we have that: 
%
\begin{align*}
\textrm{Pr}(D = 0 | \Eta=\eta, X=x, W=0) &\neq \textrm{Pr}(D = 0 | \Eta=\eta, X=x, W=1) \\
\iff \quad m_0(x, 0, \eta) - m_1(x, 0, \eta) \ &\neq \  m_0(x, 1, \eta) - m_1(x, 1, \eta).  \\ 
\end{align*}


\vspace{-1em}

\noindent \textbf{Summary of the setup.} 
I consider a decision problem where an individual selects $(d, c_d)$ to maximize his payoff:
\begin{align*}
\underset{d, c_d}{\textrm{max}}  \quad \Big\{ v_d(c_d, x, \eta) \ + \ m_d(x, w, \eta) \ + \ \epsilon_d \Big\}.
\end{align*}
\noindent The setup applies to a wide range of (static and) dynamic discrete-continuous choice models. I provide an example below that will be developed further in Section \ref{section_dynamic}. \\

\noindent \textbf{Example: Life-cycle model of consumption and labor.} \\
Consider a standard dynamic model where individuals choose how much to consume/save and whether to work or not (or to work part time or full time) every period \citep{altugmiller1998, bcms2016, arellano2017earnings}. 
The individual simultaneously chooses between working ($d=1$) or not ($d=0$), and how much to consume accordingly, $c_d$. The consumption functions can be thought of as `potential consumptions' (potential outcomes), and are completely flexible functions of the labor decision (treatment). 
The vector $x$ contains information about the asset, income, education and other individual characteristics (demographics such as the age, gender, marital status, ...). Note that the asset and income may not affect the current period utility directly, but still affect the conditional value functions, $v_d$, indirectly through their impact on the future (see more discussion in Section \ref{section_dynamic}).  
Implicitly here, I omit the unobserved permanent type $m$, which should have been identified beforehand, and could be thought of as another covariate, affecting both the preferences for work and consumption. The shock $\epsilon_d$ represents individual-specific transitory unobserved preferences for work. The shock $\eta$ represents other unobserved transitory shocks of the individual impacting her preference for consumption, and possibly also her preference for work directly. The higher $\eta$ is, the higher $c_d$ for all $d$. In practice, the greatest challenge is to find a good instrument $w$. Fortunately, the previous labor decision could serve as such an IV. Indeed, in the presence of switching costs, e.g., if $m_d(x, w, \eta) > 0$ when $d \neq w$, and $m_d(x, w, \eta) = 0$ when $d = w$, the previous labor decision is relevant for the current one (Assumption \ref{relevance}). Conditional on the current decision, on the types which capture intrinsic unobserved characteristics of the individuals, and on covariates which capture their wealth and observed characteristics, the previous decision should have no effect on the current consumption choice: it is excluded. Finally, since $\eta$ is purely transitory and occurring in period $t$, the past labor decisions are independent of it (the past labor still depends on the permanent types). 
Thus, the previous labor decision has unique properties that makes it a valid IV (relevant, excluded and exogenous) in dynamic models, because its effect on the current consumption is "subsumed" by the effect of the current labor choice. \\

\noindent \textbf{Discussion of simultaneity.} 
This simultaneous choice framework nests the non-simultaneous timings where either the discrete or the continuous choice is made first and is based on expectations about the other choice (and the corresponding shock). These two timings have testable implications for the optimal choices within the simultaneous choice framework: 
\begin{enumerate}[label=(\roman*)]
\item If the discrete choice is made first (before the realization $\Eta=\eta$ and the continuous choice), then the CCP $\textrm{Pr}(D=d | \Eta=\eta, X=x, W=w)$ does not depend on $\eta$. Indeed, $\eta$ is not yet realized. The discrete choice is only based on expectations about $\eta$ and the corresponding $c^*_d(x, \eta)$. 
\item Conversely, if the continuous choice is made first (before the realization $\Epsilon = \epsilon$ and the discrete choice), then the CCCs $c_d^*(x, \eta)$ do not depend on $d$, i.e., $c_0^*(x, \eta) = c_1^*(x, \eta)$  for all $\eta$ and $x$.
\end{enumerate}
Since I identify the policy functions $c_d^*$ and $\textrm{Pr}(D=d | \Eta=\eta, X=x, W=w)$ in the simultaneous choice framework, I can test the timing of the decisions. 

\subsection{Identification}\label{section_identification}

The unobserved shocks $(\Eta, \Epsilon)$ are independent and identically distributed across individuals. 
I observe data on the variables $(D, C, X, W)$. 
I only observe $C=C_0$ if $D=0$ and $C=C_1$ if $D=1$, where $C_0$ and $C_1$ are the potential choices.  
For all $(x, w, \eta)$ in $\mathcal{X} \times \mathcal{W} \times \mathcal{H}$, I study nonparametric identification of the following objects for $d=0, 1$: the optimal \textit{conditional continuous choices} (CCCs) $c_d^*(x, \eta)$, the optimal \textit{conditional choice probabilities} (CCPs) $\textrm{Pr}(D=d | \Eta=\eta, W=w, X=x)$, the indirect payoff functions (taken at the optimal $c$) $\underset{c}{\textrm{max}} \ v_d(c, x, \eta)$ and $m_d(x, w, \eta)$. 
In this section, I focus on any given value $X=x$ and omit $x$ from the notation in what follows. This is without loss of generality since my assumptions about the distribution of the shocks hold conditional on $X=x$. 
%
%
First, I characterize the reduced forms and constraints imposed by the structural assumptions. Then, I discuss the identification of the optimal policies (CCCs and CCPs) and of the payoffs. \\
\indent In the main text I focus on the case where $D$ is binary. Appendix \ref{appendix_discrete} discusses identification in the case where $D$ is discrete and takes more than two values. 
%
%
%
%
%
%
%

\subsubsection{Reduced forms and constraints}

\noindent In the data, I observe $(D, C, W)$, where $W$ is exogenous while $C$ and $D$ are endogenous choices.  
There is a fundamental observability problem, as I only observe one of the two potential choices $C_0$ and $C_1$ depending on the discrete choice selected: 
\begin{align*}
C = C_0 ( 1 - D ) + C_1 D .
\end{align*}
Therefore, from the data, I only recover the distribution of the potential choice $C_d$ \textit{conditional on} $D=d$ (and $W=w$), that is, $F_{C_d | d, w}(c_d) = \textrm{Pr}(C_d \leq c_d | D=d, W=w)$ for all $d$ and $w$. The functions $F_{C_0 | 1, w}(c_0)$ and $F_{C_1 | 0, w}(c_1)$ are not observed in the data.  
I also recover the conditional probability of selecting $d$ given $W=w$, that is, $p_{D|w}(d) = \textrm{Pr}(D=d | W=w)$. 
The data provide the following reduced-form functions, which exhaust all relevant information:
\begin{align*}
R = \Big\{ F_{C_d|d,w}(\cdot), \  (d,w) \in \mathcal{D}^2\ ; \ p_{D|w}(\cdot), w \in \mathcal{D} \Big\}. 
\end{align*}

A remark on terminology: in this paper, $\textrm{Pr}(D=d | W=w)$ is part of the \textit{reduced form}, while $\textrm{Pr}(D=d | \Eta=\eta, W=w)$ is what I call the \textit{conditional choice probabilities} (CCPs) or selection on unobservables process that I want to identify. 
This differs from the dynamic discrete choice literature, where $\textrm{Pr}(D=d | W=w)$ are actually called CCPs \citep{hm1993, am2011}. Here, however, I have simultaneous choices and a nonseparable shock $\eta$ that affects both choices. Thus, the counterparts to the usual CCPs are $\textrm{Pr}(D=d | \Eta=\eta, W=w)$ for all $d$, hence the different terminology. \\ 
\indent The structural assumptions imply the following constraints on the reduced form. 
\begin{Lemma}\label{distrib_c1}
Under Assumptions \ref{monotone} and \ref{ass_shocks}, $F_{C_d | d, w} (\cdot)$: $\mathcal{C}_d \rightarrow [0,1]$ is continuously differentiable and strictly increasing, for $d=0, 1$. 
\end{Lemma}
\begin{proof} Appendix \ref{appendix_proof_lemma_distrib_c1} \hfill \end{proof}

\begin{Lemma}\label{difference_instru}
Under Assumption \ref{identification2}, 
\begin{align*}
\frac{\partial F_{C_d |d, 1}(c_d) p_{D|1}(d)}{\partial c_d} \neq \frac{\partial F_{C_d |d, 0}(c_d) p_{D|0}(d)}{\partial c_d}  \quad \text{ for all } c_d \in \mathcal{C}_d \backslash \mathcal{K}^{c_d}, \text{ for } d=0,1, 
\end{align*} 
where $\mathcal{K}^{c_d}$ is a (possibly empty) finite set containing $K$ values. 
\end{Lemma}


\begin{proof} Appendix \ref{appendixnonflat} \hfill \end{proof}

\indent Lemmas \ref{distrib_c1} and \ref{difference_instru} fully characterize the impact of the structural assumptions on the reduced-form functions. 
Lemma \ref{distrib_c1} is a regularity result on the distributions implied by the structural form. 
Lemma \ref{difference_instru} provides observable and testable implications of the structural model, specifically of Assumption \ref{relevance}, on the reduced-form functions. 
Indeed, in Assumption \ref{identification2}, $\textrm{Pr}(D=d | \Eta=\eta, W=w)$ is unobserved since $\eta$ is unobserved. 
However, by monotonicity of the optimal continuous choices, the observed conditional distributions of $C_d$ given $D=d$ are transformations of the unobserved conditional distributions of $\Eta$ given $D=d$. Define the difference 
\begin{align*}
	\Delta F_{C_d}(c_d) = F_{C_d |d, 1}(c_d) p_{D|1}(d) - F_{C_d |d, 0}(c_d) p_{D|0}(d) \quad  \text{ for } d=0,1.
\end{align*}
I show that when the instrument is relevant, i.e., when $\textrm{Pr}(D=d | \Eta=\eta, W=1) \neq \textrm{Pr}(D=d | \Eta=\eta, W=0)$, we have $\partial (\textrm{Pr}(\Eta \leq \eta |D=d, W=1) - \textrm{Pr}(\Eta \leq \eta |D=d, W=0))/\partial \eta \neq 0$ and $\partial \Delta F_{C_d}(c_d^*(\eta))/\partial c_d \neq 0$. 
Now, the functions $\Delta F_{C_d}(c_d)$ and $\partial \Delta F_{C_d}(c_d)/\partial c_d$ are well defined (according to Lemma \ref{distrib_c1}) and are directly observable. Therefore, even if we do not observe the conditional distribution of $\Eta$ given $D=d$, we know that if the instrument is sufficiently relevant (Assumption \ref{relevance}), Lemma \ref{difference_instru} holds. I use this to test the relevance of the instrument: if the function $\Delta F_{C_d}(c_d)$ is \textit{flat} over an interval of values $c_d$, then there is a corresponding interval of values $\eta$ where the instrument is not relevant. In this case, the instrument has no impact on the conditional choice probabilities, so the optimal continuous choices are not point identified on this interval of $\eta$. 

\subsubsection{Identification of conditional continuous choices (CCCs)}\label{subsection_ccc}

As in the literature on continuous choice \citep{matzkin2003, bbl2007, hs2010}, I would like to exploit the monotonicity assumption to identify the optimal continuous choices. 
By monotonicity, we have for $d=0,1$, 
\begin{flalign*}
\textrm{Pr}(\Eta \leq \ \eta \ | D=d) &= \textrm{Pr}(C_d \leq \ c_d^*(\eta) \ | D=d)  &\\
\text{and, hence, by Lemma \ref{distrib_c1},} \quad & &\\ 
 c_d^*(\eta) &= F_{C_d | d}^{-1}(\textrm{Pr}(\Eta \leq \eta | D=d)). &
\end{flalign*}
Thus, \textit{if} we knew the distribution of $\Eta$ given $D=d$, we could recover the optimal conditional continuous choices $c_d^*(\eta)$. 
However, here we only know (by normalization) the unconditional distribution of $\Eta$. 
The conditional distributions of $\Eta$ given $D=d$ are unobserved. They depend on a selection on unobservables: $\textrm{Pr}(\Eta \leq \eta | D=d) = \textrm{Pr}(D=d | \Eta \leq \eta) \textrm{Pr}(\Eta \leq \eta)/\textrm{Pr}(D=d)$, which precludes the use of inversion. \\ 
%
%
%
%
%
%
%
%
%
%
%
%
\indent Another way to see the problem is as follows. 
Since $\Eta$ is exogenous, 
\begin{align*}
\textrm{Pr}(\Eta \leq \ \eta) &= \overbrace{\textrm{Pr}( C_d \leq c_d^*(\eta) )}^{\text{unobserved}} \quad \text{ for } d=0,1 \\
&= \textrm{Pr}( C_d \leq c_d^*(\eta); D=0 \text{ or } D=1)  \\
&= \underbrace{\textrm{Pr}( C_0 \leq c_0^*(\eta); \ D=0)}_{\text{observed}} + \underbrace{\textrm{Pr}( C_0 \leq c_0^*(\eta); \ D=1)}_{\text{unobserved}} \\
&= \underbrace{\textrm{Pr}( C_1 \leq c_1^*(\eta); \ D=0)}_{\text{unobserved}} + \underbrace{\textrm{Pr}( C_1 \leq c_1^*(\eta); \ D=1)}_{\text{observed}} .
\end{align*}
\textit{If} we observed both potential choices $C_0$ and $C_1$ for every individual, irrespective of the discrete choice $d$, 
then the unconditional distribution of $C_d$, $F_{C_d}$, would be observed for $d=0,1$. 
Then, knowing that $\Eta$ is uniform, one could exploit monotonicity to recover $c_d^*(\eta)$ by inverting its unconditional distribution:  $c_d^*(\eta) = F_{C_d}^{-1}(\textrm{Pr}(\Eta \leq \eta))$. 
However, we only observe $C_0$ if $D=0$ and $C_1$ if $D=1$. Because of this selection, only the conditional distributions of $C_d$ given $D=d$ are observed, and $c_d^*(\eta)$ is not identified by standard inversion. \\

\noindent \textbf{Identification with the instrument.} \\
\noindent Instead, to identify $c_d^*(h)$, I use the properties of the instrument (Assumption \ref{instrument}) to obtain structural restrictions. We have,  for $\eta \in [0,1]$ and $w=0,1$, 
\vspace{-5pt}
\begin{align}\label{bayes21}
\eta &= \quad \textrm{Pr}(\Eta \leq \eta) \nonumber \\
&= \quad \textrm{Pr}(\Eta \leq \eta| W=w) \nonumber \\
&= \quad \textrm{Pr}(\Eta \leq \eta \ | \  D = 0, W=w) \textrm{Pr}(D=0 | W=w)  \nonumber\\ 
&\quad + \ \textrm{Pr}(\Eta \leq \eta \ | \  D = 1, W=w) \textrm{Pr}(D=1 | W=w) \nonumber \\
&= \quad \textrm{Pr}(C_0 \leq c^*_0(\eta) \ | \  D = 0, W=w) \textrm{Pr}(D=0 | W=w) \nonumber \\
&\quad + \ \textrm{Pr}(C_1 \leq c^*_1(\eta) \ | \  D = 1, W=w) \textrm{Pr}(D=1 | W=w) \nonumber \\ 
&= F_{C_0 |0, w}(c^*_0(\eta)) p_{D|w}(0) + F_{C_1 |1, w}(c^*_1(\eta)) p_{D|w}(1).  
\end{align}
%
\noindent Now, take equation (\ref{bayes21}) at $w=0$ and $w=1$ to obtain a system of two equations to identify two unknown increasing functions, $c^*_0(\cdot)$ and $c^*_1(\cdot)$. The role of the instrument and Assumption \ref{instrument} appears clearly here. First, the exclusion of $w$ from $c_d^*$ is necessary to avoid having four unknown functions $c_d^*(\eta, w)$ (where $d=0,1$ and $w=0,1$) which would not be identified with two equations. 
Similarly, without a relevant instrument (e.g., if $D \perp W$), 
$p_{D|0}(d) = p_{D|1}(d)$ and $F_{C_d | d, 0}(c) = F_{C_d | d, 1}(c)$ for $d=0,1$, so the two equations would coincide, giving one equation for two unknown functions.  \\

\vspace{-10pt}

\begin{theorem}[Identification]\label{identification_theorem}
For every reduced form compatible with the structural model, there exist unique conditional continuous choice (CCC) functions $c_d(h)$ ($d=0, 1$) that are strictly increasing and satisfy
{\leqnomode
\begin{equation}\label{bayes2}
\eta = F_{C_0 |0, w}(c_0(\eta)) p_{D|w}(0) + F_{C_1 |1, w}(c_1(\eta)) p_{D|w}(1)  \quad \text{ for } \eta \in [0, 1], \  w=0,1. 
\end{equation} 
}
\end{theorem} 
\vspace{0.5\baselineskip}
\indent The CCC functions are identified if and only if there exist unique functions $c_d(\eta)$, $d=0,1$, that are strictly increasing in $\eta$, satisfy equation (\ref{bayes2}), and are compatible with the reduced form $R$. \\ 
%
%
\vspace{-10pt}

%
%

\begin{proof}
The existence of a solution is trivial: the reduced form is compatible with the structural model, so, by construction following equation (\ref{bayes21}), $c^*_d(\cdot)$ ($d=0,1$) solve (\ref{bayes2}). 
%
%

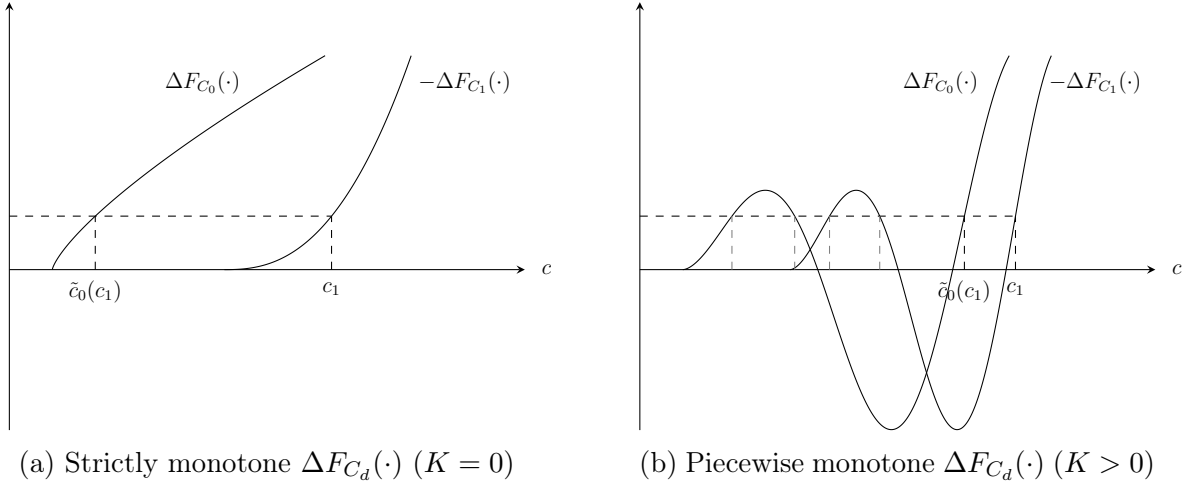
\begin{figure}[!t]
  \centering
  \begin{subfigure}[b]{0.45\textwidth}
    \begin{tikzpicture}[scale=1]
\begin{axis}[
	clip=false,
    axis lines = middle,
    ytick=\empty, xtick=\empty,
    ymin=-3, ymax=5,
    xmin=0, xmax=1.2, 
]
\addplot [domain=0.1:0.735, samples=1000, color=black,]{(10*(x-0.1))^0.75};
\addplot [domain=0.35:0.9353, samples=1000, color=black,]{(4*(x-0.5))^2.5};
\node[scale=0.7] at (axis cs:1.25,0) {$c$};
\node[scale=0.7] at (axis cs:1.06,3.5) {$-\Delta F_{C_1}(\cdot)$};
\node[scale=0.7] at (axis cs:0.45,3.5) {$\Delta F_{C_0}(\cdot)$};
\draw[dashed] (axis cs:0.75,0) -- (axis cs:0.75,1);
\node[below, scale=0.7] at (axis cs:0.75,-0.10) {$c_1$};
\draw[dashed] (axis cs:0.2,0) -- (axis cs:0.2,1);
\node[below, scale=0.7] at (axis cs:0.2,-0.05) {$\tilde{c}_0(c_1)$};
\draw[dashed] (axis cs:0,1) -- (axis cs:0.75,1);
\end{axis}
\end{tikzpicture}
\caption{Strictly monotone $\Delta F_{C_d}(\cdot)$ $(K = 0)$}\label{fig:identification1}
\end{subfigure}%
   \hfill
  \begin{subfigure}[b]{0.45\textwidth}
\begin{tikzpicture}[scale=1]
\begin{axis}[
	clip=false,
    axis lines = middle,
    ytick=\empty, xtick=\empty,
    ymin=-3, ymax=5,
    xmin=0, xmax=1.2, 
]
\addplot [domain=0.1:0.859773, samples=1000, color=black,]{(10*(x-0.1))^0.7*sin(deg(10*(x-0.1)))};
\addplot [domain=0.35:0.957818, samples=1000, color=black,]{(12.5*(x-0.35))^0.7*sin(deg(12.5*(x-0.35)))};
\node[scale=0.7] at (axis cs:1.25,0) {$c$};
\node[scale=0.7] at (axis cs:1.06,3.5) {$-\Delta F_{C_1}(\cdot)$};
\node[scale=0.7] at (axis cs:0.7,3.5) {$\Delta F_{C_0}(\cdot)$};
\draw[dashed] (axis cs:0.87437446,0) -- (axis cs:0.87437446,1);
\node[below, scale=0.7] at (axis cs:0.87437446,-0.10) {$c_1$};
\draw[dashed] (axis cs:0.75546807,0) -- (axis cs:0.75546807,1);
\node[below, scale=0.7] at (axis cs:0.75546807,-0.05) {$\tilde{c}_0(c_1)$};
\draw[dashed] (axis cs:0,1) -- (axis cs:0.87437446,1);
\draw[gray, dashed] (axis cs:0.5583563,0) -- (axis cs:0.5583563,1);
\draw[gray, dashed] (axis cs:0.441506,0) -- (axis cs:0.441506,1);
\draw[gray, dashed] (axis cs:0.2143825,0) -- (axis cs:0.2143825,1);
\draw[gray, dashed] (axis cs:0.36044535,0) -- (axis cs:0.36044535,1);
\end{axis}
\end{tikzpicture}
\caption{Piecewise monotone $\Delta F_{C_d}(\cdot)$ $(K > 0)$}\label{fig:identification2}
\end{subfigure}%
\caption{Intuition behind identification}\label{fig:identification}
\vspace{1em}
\end{figure}

\noindent To prove uniqueness, combine the two equations of (\ref{bayes2}) to give, for $\eta \in \mathcal{H}$, 
\begin{flalign*}
  F_{C_0 |0, 0}(c^*_0(\eta)) p_{D|0}(0) + F_{C_1 |1, 0}(c^*_1(\eta)) p_{D|0}(1) &= F_{C_0 |0, 1}(c^*_0(\eta)) p_{D|1}(0) + F_{C_1 |1, 1}(c^*_1(\eta)) p_{D|1}(1) & \\ 
\text{and hence, after rearranging, } \quad \quad \quad \quad \quad \quad \quad & & \\
\quad   \Delta F_{C_0}(c^*_0(\eta)) &= - \Delta F_{C_1}(c^*_1(\eta)),  &
\end{flalign*}
where the functions $\Delta F_{C_d}(\cdot)$ are directly observed from the data. Now, even without observing $\eta$, if two conditional choices $\tilde{c_0}$ and $\tilde{c_1}$ correspond to the same unobserved $\eta$, we have $\Delta F_{C_0}(\tilde{c_0}) = - \Delta F_{C_1}(\tilde{c_1})$. That is, 
\begin{equation}\label{mapping_sketch}
\Delta F_{C_0}(\tilde{c_0}(c_1)) = - \Delta F_{C_1}(c_1) \text{ for all } c_1 \in c_1^*(\mathcal{H}) = \mathcal{C}_1, 
\end{equation}
where $\tilde{c_0} = c_0^*\circ {c_1^*}^{-1}$ is strictly increasing. 
The mapping $\tilde{c_0}$ is identified if and only if there exists a unique strictly increasing function solving (\ref{mapping_sketch}). 
Notice that the functions $\Delta F_{C_d}(\cdot)$ are transformations (through $c_d^*(\cdot)$) of the same underlying object based on the difference between $\textrm{Pr}(D=0 | \Eta=\eta, W=1) - \textrm{Pr}(D=0 | \Eta = \eta, W=0)$.\footnote{Specifically, as shown in Appendix \ref{appendixnonflat}, for $d=0, 1$, \begin{align*} \Delta F_{C_d}(c) = (-1)^d \int^{(c_d^*)^{-1}(c)}_{0} \Big( \textrm{Pr}(D=0 | \Eta=\eta, W=1) - \textrm{Pr}(D=0|\Eta = \eta, W=0) \Big) d\eta. \end{align*} } 
Thus $- \Delta F_{C_1}(\cdot)$ is a non-constant shift of $\Delta F_{C_0}(\cdot)$, i.e., both functions go through the same values in the same order. 
So, in the simple case where the instrument is relevant everywhere (Figure \ref{fig:identification1}), the functions $\Delta F_{C_0}(\cdot)$ and $-\Delta F_{C_1}(\cdot)$ are strictly monotone and continuously differentiable, they have the same range and we can invert (\ref{mapping_sketch}) to get the unique solution: 
\begin{align*}
\tilde{c_0}(c_1) = \Delta F_{C_0}^{-1}\Big(-\Delta F_{C_1}(c_1)\Big) \text{ for all } c_1 \in \mathcal{C}_1.
\end{align*}

Now, if the instrument is not relevant at $K$ isolated values of $\eta$ (Figure \ref{fig:identification2}), the functions $\Delta F_{C_d}(\cdot)$ are not strictly monotone and not invertible. However, they are piecewise monotone and piecewise invertible. So, using the monotonicity constraint on $\tilde{c_0}(c_1)$, the solution to (\ref{mapping_sketch}) is also unique in this case. 
Indeed, for any value $y$ in the range of $\Delta F_{C_0}(\cdot)$, the $k^{th}$ value $c_0$, denoted $c_0^k$, such that $\Delta F_{C_0}(c_0^k) = y$ for the $k^{th}$ time, is the image of the $k^{th}$ value of $c_1$, denoted $c_1^k$, such that $ - \Delta F_{C_1}(c_1^k) = y$ for the $k^{th}$ time. 
Even though there may exist several solutions (at most $K+1$) such that $y = \Delta F_{C_0}(c_0) = \Delta F_{C_1}(c_1)$ for any given $y$, there is only one $\tilde{c_0}(c_1)$ that is strictly increasing and satisfies (\ref{mapping_sketch}) for all $c_1$.\footnote{For more details on the proof when $K > 0$, see Appendix \ref{appendix_proof_identification}.} \\ 
%
%
\indent The only case in which uniqueness does not hold is when the functions $\Delta F_{C_d}(\cdot)$ are flat on  some interval, i.e., when the instrument is not relevant on an interval of $\eta$. 
In this case, $\tilde{c_0}(c_1)$ is partially identified: it is point identified everywhere except on the flat part where there exists an infinite number of solutions satisfying (\ref{mapping_sketch}). \\
\indent Once we identify $\tilde{c}_0(c_1)$, we recover the corresponding unobserved $\eta$ as  
\begin{align*}
\eta(c_1) =  F_{C_0 |0, w}(\tilde{c}_0(c_1)) p_{D|w}(0) + F_{C_1 |1, w}(c_1) p_{D|w}(1) \quad \text{ for all } c_1 \in \mathcal{C}_1, \text{ for any } w = 0,1. 
\end{align*}
Thus we have a unique increasing solution $(\eta(c_1), \tilde{c}_0(c_1))$ for all $c_1 \in \mathcal{C}_1$. Finally, using strict monotonicity of these functions, we obtain $c_1^*(\eta) = \eta^{-1}(\eta(c_1))$ and $c_0^*(\eta) = \tilde{c_0}(c_1(\eta))$ for all $\eta$. 
Thus the optimal continuous policy functions $(c^*_0(\eta), c^*_1(\eta))$ for all $\eta \in [0,1]$ are identified as the unique solution to system (\ref{bayes2}).  \end{proof}

\indent One key point in this proof is that identifying assumptions can be relaxed by exploiting the strict monotonicity of $c_d^*(\eta)$. Indeed, even though (\ref{bayes2}) may have multiple non-monotone solutions, only one solution is strictly monotone.  
Without monotonicity, I would not obtain general point identification unless I assumed a stronger version of Assumption \ref{relevance}, for example, assuming $\textrm{Pr}(D=0 | \Eta=\eta, W=1) - \textrm{Pr}(D=0 | \Eta=\eta, W=0) > 0$ for all $\eta$. This would be close to the full rank assumption on the effect of the instrument on the selection in order to identify quantile treatment effects \citep{neweypowell2003, chernozhukovhansen2005, chernozhukovhansen2006, chernozhukovhansen2008}, which is in fact stronger than necessary. 
\cite{vuongxu2017} also exploit the power of monotonicity to relax \cite{chernozhukovhansen2005}'s full rank condition and still identify binary treatment effects. Their weaker condition remains at a "high level", while I show that it can be easily expressed in terms of the CCPs (Assumption \ref{relevance}). 
\subsubsection{Identification of conditional choice probabilities (CCPs)}\label{subsection_ccp}

Now that the CCCs, 
$c_d^*(\cdot)$ for $d=0,1$, are identified, identification of the conditional choice probabilities follows readily. 
Since $\Eta \sim \mathcal{U}(0, 1)$, we have $F_{C_d}(c_d) = {c_d^{*}}^{-1}(c_d)$, with corresponding density $f_{C_d}(c_d) = \partial {c_d^{*}}^{-1}(c_d) / \partial c_d$ for $d=0,1$. 
The CCPs are identified from 
\begin{align}\label{eq_identification_ccp}
	\textrm{Pr}(D=d | \Eta = \eta, W=w) &= \textrm{Pr}(D=d | C_d = c_d^*(\eta), W=w)  \notag \\
	&=	f_{D, C | w}(d, c_d^*(\eta)) / f_{C_d | w}(c_d^*(\eta)) \notag \\
	&=	f_{D, C | w}(d, c_d^*(\eta)) / f_{C_d}(c_d^*(\eta)),
\end{align}
where $f_{D, C | w}(\cdot)$ is the joint conditional density of $D$ and $C$ given $W=w$, and $f_{C_d | w}(\cdot)$ is the conditional density of $C_d$ given $W=w$. \\
%
%
%
%
\indent Alternatively, to identify the CCPs notice that, since $c^*_d(.)$ is strictly monotone, one can recover $\Eta$ from observing $(D, C)$ as $\Eta = (c^*_D)^{-1} \big( C \big).$
From there, it is as if $(D, C, W, \Eta)$ was observed from the data. Thus, the CCPs are identified once $\Eta$ is recovered from inverting the CCCs. 

\subsubsection{Identification of the payoffs}\label{subsection_static_payoffs}

Once the optimal policies are nonparametrically identified, we can use them to identify the remaining primitives of the model. 
For example, the differences in payoffs between $D=1$ and $D=0$ at the corresponding optimal continuous choice are semi-parametrically identified by the CCPs using \citeauthor{hm1993}'s (\citeyear{hm1993}): the result depends on the distribution of $\Epsilon$.
Generally, the payoffs can be nonparametrically identified via the CCPs and CCCs, but the identification depends on the model itself, in particular, on how the CCCs relate to the marginal utilities. Typically, in the case of dynamic problems, the marginal utility can be nonparametrically identified using the first order conditions, Euler equations and the CCCs and CCPs. 

\section{Dynamic models}\label{section_dynamic}

In this section, I show how general single-agent (possibly non-stationary) dynamic models can be nonparametrically identified. The main idea is that these dynamic models yield intra-period problems as described in Section \ref{allsection_framework}, thus the optimal choices (CCCs and CCPs) are identified period by period following Section \ref{section_identification}, and we can then use these optimal choices to identify the primitives of the model \citep[in the spirit of][]{bmm1997}. Note that all this Framework is implicitly conditional on unobserved types $m$, which are identified beforehand, following Section \ref{section_types}, and which I abstract from in the notation.  
The framework is general and nests many life-cycle empirical applications of interest \citep[e.g., ][]{bcms2016, iskhakov2017}. I focus on the leading example of a dynamic model of labor and consumption choices.

\subsection{Dynamic life-cycle model of labor and consumption}\label{section_dynamic_example}

Focus on a general dynamic model of labor and consumption choices with a finite horizon ($T < \infty$) (but the arguments also apply when the horizon is infinite).\footnote{In fact, in the case of an infinite horizon with a stationary environment the identification is considerably simplified because the optimal choices are time-independent.} \\ 
\indent Each period $t$ until $T$, the timing of the individual's problem is the following:

\setlength{\baselineskip}{15pt}
\begin{center}
\begin{tikzpicture}[x=2.5cm]
\draw[->,thick,>=latex]
  (0,0) -- (4.5,0) node[below right] {$\scriptstyle time$};

\draw[thick] (0,-.2) -- (0, .2) node[above] {$t$};
\draw[thick] (1,-.2) -- (1, .2) node[above] {$t^+$};
\draw[thick] (2.2,-.2) -- (2.2, .2) node[above] {$t^{++}$};
\draw[thick] (3.2,-.2) -- (3.2, .2) node[above] {$t+1$};

\draw (0,0) node[align=center, below=6pt] {State \\ $z_t=(x_t, w_t)$ \\};
\draw (1,0) node[align=center, below=6pt] {Shocks $(\epsilon_t, \eta_t)$ \\ occur};
\draw (2.2,0) node[align=center, below=6pt] {Individual \\ picks $(d_{t}, c_{dt})$};
\draw (3.2,0) node[align=left, below=15pt] {...};

\draw[decorate, decoration={brace, amplitude=10pt}, thick] 
  (-0.1,0.8) -- (2.3,0.8) node[midway, above=10pt] {intra period $t$}; 
  
\end{tikzpicture}
\end{center}
\setlength{\baselineskip}{21pt}
\vspace{-1.5em}

\indent The current period conditional utility for action $(d_t, c_{dt})$ at time $t$ is given by
\begin{equation}
\mathcal{U}_{dt}(c_{dt}, x_t, w_t, \eta_t, \epsilon_{t}).
\end{equation}
In this example, as explained earlier, $c_t$ is consumption and $c_{dt}$ are potential consumption choices, with $c_t = c_{0t} (1-d_t) + c_{1t} d_t$, $d_t$ is the labor decision to work or not, 
$x_t$ represents all the covariates and $w_t$ is the instrument. The covariates include variables such as age, education and other demographics, impacting current utility. For notational convenience, $x_t$ also includes variables such as assets or income which do not necessarily directly impact preferences but still have an impact on the consumption choice (and labor choice), notably through their transitions. \\ 
%
\indent I impose additional assumptions on the current utilities which are necessary (but not sufficient) such that the intra-period problems of this dynamic setup fit into the framework of Section \ref{section_framework}. 

\begin{assumptionbis}{additive}[Additive Separability]\label{additivebis}
The shock $\epsilon_t$ enters the payoff additively
\begin{align*}
\mathcal{U}_{dt}(c_{dt}, x_t, w_t, \eta_t, \epsilon_t) = \tilde{u}_{dt}(c_{dt}, x_t, w_t, \eta_t) + \epsilon_{dt}, \quad \text{ for } d_t=0, 1.
\end{align*}
\end{assumptionbis}
\begin{assumptionbis}{instrument}[Instrument]\label{instrumentbis} The instrument $w_t \in \mathcal{W} = \mathcal{D}$ is such that
\begin{align*}
\tilde{u}_{dt}(c_{dt}, x_t, w_t, \eta_t) = u_{dt}(c_{dt}, x_t, \eta_t) +  m_{dt}(x_t, w_t, \eta_t), \quad \text{ for } d_t=0, 1.
\end{align*}
\end{assumptionbis}
\begin{assumptionbis}{monotone}[Monotonicity]\label{monotonebis}
The conditional current utility functions $u_{dt}$ are twice continuously differentiable and
\begin{equation*}
\frac{\partial^2 u_{dt}(c_{dt}, x_t, \eta_t)}{\partial c_{dt} \partial \eta_t} > 0, \quad \text{ for } d_t=0, 1. 
\end{equation*}
\end{assumptionbis}

\vspace{10pt}

\indent In a dynamic context, the individual chooses $(d_t, c_{dt})$ to maximize her expected discounted sum of current and future payoffs. She discounts the future utilities at a rate $\beta$ and forms rational expectations about the transition probabilities. The transitions from $(x_t, w_t, \epsilon_t, \eta_t)$ and the current choices $(c_t, d_t)$ to $(x_{t+1}, w_{t+1}, \epsilon_{t+1}, \eta_{t+1})$ matter for the choices. In particular, how the current choices impact these transitions is especially important for optimal choice: for example, individuals do not consume all their wealth in a given period because they are forward-looking and want to save for the future. 
The impacts of the choices on the transitions are often expressed through a budget constraint like $a_{t+1} = (1+r_t) a_t - c_t + y_t d_t,$ where $a_t$ is the value of assets, $r_t$ the return on assets, and $y_t$ is labor income.  
For now let us be more general and only assume the existence of a general transition density of states and errors, which depend on the choices:
\begin{align*}
f_{X_{t+1}, W_{t+1}, \Epsilon_{t+1}, \Eta_{t+1} | x_{t}, w_t, c_t, d_t, \epsilon_t, \eta_t}(x_{t+1}, w_{t+1}, \epsilon_{t+1}, \eta_{t+1}).
\end{align*}
I make additional assumptions on this density for the model to be identified and to fit into the general framework. 

\begin{assumption}[Conditional independence]\label{cond_indep}
For all $x_t \in \mathcal{X}$, $w_t \in \mathcal{W}$, $\epsilon_t \in \mathcal{E}$, $\eta_t \in \mathcal{H}$, 
\begin{align*}
&f_{X_{t+1}, W_{t+1}, \Epsilon_{t+1}, \Eta_{t+1} | x_{t}, w_t, c_t, d_t, \epsilon_t, \eta_t}(x_{t+1}, w_{t+1}, \epsilon_{t+1}, \eta_{t+1}) \\
&= f_{X_{t+1}, W_{t+1} | x_t, w_t, c_t, d_t}(x_{t+1}, w_{t+1}) \ f_{\Epsilon} (\epsilon_{t+1}) \ f_{\Eta} (\eta_{t+1}). 
\end{align*}
\end{assumption}

\begin{assumption}[Instrument transition exclusion]\label{instru_transi}
For all $x_t \in \mathcal{X}$ and $w_t \in \mathcal{W}$, the current instrument is excluded from the transition density:
\vspace{-10pt}
\begin{align*}
f_{X_{t+1}, W_{t+1} | x_t, w_t, c_t, d_t}(x_{t+1}, w_{t+1}) =  f_{X_{t+1}, W_{t+1} | x_t, c_t, d_t}(x_{t+1}, w_{t+1}). 
\end{align*}
\end{assumption}


\indent I also impose the Assumption \ref{ass_shocks} and \ref{relevance} contemporaneously (i.e., adapted with index $t$). I do not rewrite them for simplicity of exposition. 
First, let me show how the intra-period problem of this dynamic model fits into the intra-period problem described in Section \ref{allsection_framework}, and then discuss further the role of these assumptions.  \\
\indent Knowing the transition densities, the individual chooses $(d_t, c_{dt})$ to sequentially maximize her expected discounted sum of payoffs. 
Let $V_t(z_t) = V_t(x_t, w_t)$ be the (ex ante) value function of this discounted sum of payoffs at the beginning of $t$, just before the shocks $(\epsilon_t, \eta_t)$ are revealed and conditional on behaving according to the optimal decision rule. We have
\begin{align*}
V_t(z_t) \equiv \mathbb{E}\Big[ \sum^T_{\tau = t} \beta^{\tau-t} \underset{D, C_{D\tau}}{\textrm{max}} [ \ u_{D\tau} (C_{D\tau}, X_\tau,  \Eta_{\tau} ) + m_{D}(X_\tau, W_\tau, \Eta_{\tau} ) + \Epsilon_{D\tau} ] \Big]. 
\end{align*}
%
%
%
%
Given the state variable $z_t$ and choice $(d, c_{dt})$ in period $t$, the expected value function in period $t+1$ is
\begin{align*}
\mathbb{E}_{Z_{t+1}}[V_{t+1}(Z_{t+1}) | Z_t=z_t, C_t=c_{t}, D_t=d_t]  =  \int_{\mathcal{Z}_{t+1}} V_{t+1} (z_{t+1}) f_{Z_{t+1} | z_t, c_t, d_t}(z_{t+1}) dz_{t+1}.
\end{align*}
By the conditional independence (Assumption \ref{cond_indep}) and instrument exclusion from the transition (Assumption \ref{instru_transi}), we can remove $W_t$ from the conditioning variables,  
\begin{align*}
\mathbb{E}_{Z_{t+1}}[V_{t+1}(Z_{t+1}) | Z_t=z_t, C_t=c_{t}, D_t=d_t ]  &= \mathbb{E}_{Z_{t+1}}[V_{t+1}(z_{t+1}) | X_t=x_t, C_t=c_{t}, D_t=d_t ] 
\end{align*}

\noindent The ex ante value function can be written recursively: 
\begin{align*}
V_t(z_t) = \mathbb{E}_{\Epsilon, \Eta} \bigg[  \underset{d_t, c_{dt}}{\textrm{max}}  \Big[& \ u_{dt} (c_{dt}, x_t,  \Eta_{t} ) + m_{dt}(x_t, w_t, \Eta_{t} ) + \Epsilon_{dt}   \\
& + \beta \mathbb{E}_{Z_{t+1}}[V_{t+1}(Z_{t+1}) |X_t= x_t, C_t=c_{dt}, D_t=d_t ] \Big] \  \bigg]. 
\end{align*}

\noindent Thus, in each period, after observing $(\Epsilon_t, \Eta_t) = (\epsilon_t, \eta_t)$, the individual chooses $d_t$ and $c_{dt}$ to maximize her expected payoff: 
\begin{align*}
\underset{d_t, c_{dt}}{\textrm{max}} \ \left\{ u_{dt} (c_{dt}, x_t,  \eta_{t} ) + \beta \mathbb{E}_{Z_{t+1}}[V_{t+1}(Z_{t+1}) | X_t=x_t, C_t=c_{dt}, D_t=d_t ] \ + m_{dt}(x_t, w_t, \eta_{t} ) + \epsilon_{dt} \right\}.
\end{align*}
Define the \textit{conditional value functions} $v_{dt}(\cdot)$ as
\begin{align}\label{vd_def}
v_{dt}(c_{dt}, x_t, \eta_t) \ = \  u_{dt} (c_{dt}, x_t,  \eta_{t} ) + \beta \mathbb{E}_{Z_{t+1}}[V_{t+1}(Z_{t+1}) | X_t=x_t, C_t=c_{dt}, D_t=d_t ].
\end{align}
\noindent Now the dynamic model yields the same maximization problem as in the general framework of Section \ref{section_framework}. Every period, the individual selects $d_t$ and $c_{dt}$ to solve: \\
\vspace{-24pt}
\begin{align*}
\underset{d_t, c_{dt}}{\textrm{max}} \quad \left\{ v_{dt}(c_{dt}, x_t, \eta_t) \ +  m_{dt}(x_t, w_t, \eta_{t} ) + \epsilon_{dt} \right\}. \\ 
\end{align*}


\vspace{-24pt}
\begin{Lemma}[Dynamic framework]\label{Lemma_dynamic_match} 
Under Assumptions \ref{additivebis}, \ref{instrumentbis}, \ref{monotonebis}, \ref{cond_indep} and \ref{instru_transi}, Assumptions \ref{additive}, \ref{instrument} and \ref{monotone} are satisfied for the conditional value functions defined in equation (\ref{vd_def}) in the dynamic setup. 
\end{Lemma}

Thus, under Assumptions \ref{additivebis}, \ref{instrumentbis}, \ref{monotonebis}, \ref{cond_indep} and \ref{instru_transi}, all the assumptions of the general framework of Section \ref{section_framework} hold (with the shocks and relevance assumptions directly adapted to the dynamic setup).  
If Assumption \ref{additivebis} holds for the current utility function, then, by construction, Assumption \ref{additive} will hold for the conditional value functions in Equation (\ref{vd_def}). Assumptions \ref{instrumentbis} and \ref{monotonebis} on the current utility do not translate directly into Assumptions \ref{instrument} and \ref{monotone} for the conditional value function. One needs additional assumptions about the transitions, i.e., Assumptions \ref{cond_indep} and \ref{instru_transi}. \\ 
\indent Conditional independence assumptions are standard for the identification and empirical tractability of dynamic discrete choice models \citep{rust1987, blevins2014}. Here, Assumption \ref{cond_indep} implies that the transitions of the state variables are independent of the shocks $(\Epsilon_t, \Eta_t)$. Similarly, the shock transitions are independent of the variables here. There is \textit{no time dependence} on the shocks, which are iid across periods. 
\noindent Crucially, here, in addition to the standard conditional independence, Assumption \ref{instru_transi} also implies that conditional on $(D_t, C_t, X_t)$, the \textit{transitions} are \textit{independent} of the current instrument value $w_t$. In particular, the instrument is excluded from its own transition to future values, conditional on $(D_t, C_t, X_t)$, i.e., 
\begin{flalign*}
&W_{t+1} \perp W_t \  | \ C_t=c_t, D_t=d_t, X_t=x_t &\\ 
\text{ or equivalently}  \quad \quad \  &f_{W_{t+1} | x_t, w_t, c_t, d_t}(w_{t+1}) = f_{W_{t+1} | x_t, c_t, d_t}(w_{t+1}). &
\end{flalign*}
This excludes the possibility of having time-independent instrument (e.g., $w_t = w$ for all $t$). 
%
%
Assumption \ref{instrumentbis} combined with Assumption \ref{instru_transi} will satisfy Assumption \ref{instrument} on the conditional value $v_{dt}$ as shown in the computation above. Without the exclusion of the instrument from the transition, $w_t$ could affect the expected future value function and thus enter the conditional value functions, in which case the exclusion restriction of $w_t$ from the payoff (Assumption \ref{instrument}) would not be satisfied. \\ 
%
%
%
%
\indent Similarly, under Assumption \ref{monotonebis} and the conditional independence of the future from current $\eta_t$ (Assumption \ref{cond_indep}), we have
\begin{align*}
\frac{\partial^2  v_{dt}(c_{dt}, x_t, \eta_t)}{\partial c_{dt} \partial \eta_t } = \frac{\partial^2 u_{dt} (c_{dt}, x_t,  \eta_{t} )}{\partial c_{dt} \partial \eta_t} + \underbrace{\frac{\partial^2 \mathbb{E}_{Z_{t+1}}[V_{t+1}(Z_{t+1}) | X_t=x_t, C_t=c_{dt}, D_t=d_t ]}{\partial c_{dt} \partial \eta_t}}_{=0}, 
\end{align*}
and the monotonicity of the conditional value functions $v_{dt}(\cdot)$, Assumption \ref{monotone}, holds. \\ 
\noindent \textbf{Discussion about the instrument.} 
In many dynamic setups, a convenient instrument that satisfies all the assumptions could be the previous discrete choice, $W_t = D_{t-1}$. In this case, the exclusion from the transition (Assumption \ref{instru_transi}) is likely to be satisfied, because (i) $w_{t+1} = d_t$ in this case, thus conditional on the current $d_t$ choice, $w_{t+1}$ is known irrespective of the value of $w_t$, and (ii) conditional on the current $d_t$ and $x_t$, it is unlikely that $d_{t-1}$ impacts $x_{t+1}$. Moreover, the current exclusion restriction (Assumption \ref{instrumentbis}) is also satisfied because conditional on $x_t$, which may, for example, include work experience, it is unlikely that $d_{t-1}$ impacts the current utility $u_{dt}(\cdot)$. Overall, the effect of $d_{t-1}$ on $c_t$ is subsumed by the effect of $d_t$ on $c_t$. Finally, $d_{t-1}$ is a relevant instrument (Assumption \ref{relevance}) if there exists an utility switching cost from exiting or entering the workforce, for example.\footnote{$W_t = D_{t-1}$ is relevant if there is some `autocorrelation' (that I interpret as switching costs) in the discrete choice (conditional on the types). This could be driven by autocorrelation in a general $\tilde{\epsilon}_{dt}(x_t, w_t, \eta_t) = m_{dt}(x_t, w_t, \eta_t) + \epsilon_{dt}$ error term. Thus, the assumption about no correlation in $\epsilon_t$ is less restrictive than it seems. } \\
\indent The reason why I do not allow for autocorrelation in $\Eta_t$ is that I recommend to use the past discrete choice as the default instrument. Indeed, with autocorrelated $\Eta_t$, if $W_t = D_{t-1}$, then in the first period $W_1$ and $\Eta_1$ are not independent as they are both correlated with the unobserved $\Eta_{0}$. However, if one can find another instrument satisfying Assumptions \ref{indep_shock}, \ref{instrumentbis} and \ref{instru_transi} and which do not suffer from this initial period problem, we could include and identify autocorrelation in $\Eta_t$: $f_{\Eta_{t+1} | \eta_t}(\eta_{t+1})$. Because of the exclusion from its own transition (Assumption \ref{instru_transi}), such an instrument is hard to find in dynamic models, and would need to be a purely transitory and unexpected event.  
Now, with $W_t = D_{t-1}$ I can still allow for autocorrelation in the unobservables by including permanent unobserved types \citep{am2011} in the setup (Section \ref{section_types}). The types reintroduce time-dependence in the model and attenuate the effect of the conditional independence assumption.

\subsection{Identification of the dynamic model}\label{subsection_dynamic_identification}

First, I show how the transitions, the CCCs and CCPs are identified in the dynamic model. Then, I show how to use them to nonparametrically identify the marginal utility, the discount factor and the conditional payoffs under additional assumptions.

\subsubsection{Optimal choices: CCCs and CCPs}

Under Lemma \ref{Lemma_dynamic_match}, the dynamic framework described in Section \ref{section_dynamic} fits into the general framework described in Section \ref{section_framework}. Therefore the CCCs and CCPs are identified period by period from $t \geq 2$ onwards, following the proof developed in Section \ref{section_identification}. 
The data $\{ D_t, C_{t}, X_t, W_t, t\}_{t=1}^T$ provides the following reduced-form functions:
\begin{align*}
R = \Big\{ &\textrm{Pr}(D=d_t|X_t=x_t, W_t=w_t), (d_t, x_t, w_t, t) \in \mathcal{D}\times\mathcal{X}_t\times\mathcal{W}\times \{2,..., T\} , \\
 &F_{C_{dt} | d_t, x_t, w_t}(\cdot), (d_t, x_t, w_t, t) \in \mathcal{D}\times\mathcal{X}_t\times\mathcal{W}\times \{2,..., T\} \Big\}. 
\end{align*}
From these reduced forms, following Section \ref{section_identification}, I identify the CCCs and CCPs
\begin{align*}
c^*_{dt}(\eta_t, x_t) \text{ and } \textrm{Pr}(D_t = d_t | \Eta_t=\eta_t, X_t = x_t, W_t=w_t),  
\end{align*}
for all $(d_t, \eta_t, x_t, w_t, t) \in \mathcal{D}\times [0, 1] \times \mathcal{X}_t\times\mathcal{W}\times \{2,..., T\}$. 
\noindent Note that identification does not hold for $t=1$ because $W_t=D_{t-1}$ is not available for $t=1$. \\

\noindent \textbf{Special case: Identification of the choices with terminal/absorbing actions.} \\ 
Suppose $D_t = 1$ is a \textit{terminal action} or an \textit{absorbing state}. For example, $D_t = 1$ if the individual retires, $D_t=0$ if she stays active. Assuming that an individual cannot go back to working life, the retirement choice is absorbing \citep{iskhakov2017, levy2024identification}. 
Now, identification is greatly simplified. Indeed, use $W_t = D_{t-1}$ as the instrument. When $D_{t-1} = 1$, the instrument is `infinitely' relevant: the probability of staying retired is one. Thus by focussing on previously retired individuals ($W_t=1$), equation (\ref{bayes2}) gives
\begin{align*}
	\eta_t = F_{C_{1t} | 1, x_t, 1}(c^*_{1t}(\eta_t, x_t)) \quad \text{ for all } \eta_t \in [0, 1], \ x_t \in \mathcal{X}_t.
\end{align*}
Since $F_{C_1 | 1, X_t, 1, t}(c)$ is invertible (Lemma \ref{distrib_c1}), we recover the continuous choices conditional on being retired as: 
\begin{align*}
c^*_{1t}(\eta_t, x_t) = F^{-1}_{C_{1t} | 1, x_t, 1}(\eta_t) \quad \text{ for all } \eta_t \in [0, 1], x_t \in \mathcal{X}_t.
\end{align*}
It remains to identify the other conditional continuous policy. Take equation (\ref{bayes2}) at $W_t = 0$, i.e., for individuals who did not select the absorbing state yet. It yields
{\small
\begin{flalign*}
\eta_t =&\quad F_{C_{0t} | 0, x_t, 0}(c^*_{0t}(\eta_t, x_t)) \textrm{Pr}(D_t=0|X_t=x_t, W_t=0)  \\
&+ \quad F_{C_{1t} | 1, x_t, 0}(c^*_{1t}(\eta_t, x_t)) \textrm{Pr}(D_t=1|X_t=x_t, W_t=0) &\\
\text{and hence, } &  &\\
c^*_{0t}(\eta_t, x_t) =& F^{-1}_{C_{0t} | 0, x_t, 0} \left( \frac{\eta_t - F_{C_{1t} | 1, x_t, 0}(c^*_{1t}(\eta_t, x_t)) \textrm{Pr}(D_t=1|X_t=x_t, W_t=0)}{\textrm{Pr}(D_t=0|X_t=x_t, W_t=0)} \right). &
\end{flalign*}
}
This identifies the optimal continuous choice of active individuals ($D_t=0$), since all the terms on the right hand side are already known. Once the CCCs are identified, we proceed as previously to identify the CCPs.

\subsubsection{Transitions} 
The transition density $f_{X_{t+1} | X_t=x_t, C_t=c_t, D_t=d_t}(x_{t+1})$ is identified directly from the data by observing the conditional transitions of the variables between consecutive periods $t$ and $t+1$. 
The transition of the instrument is known by construction if $w_{t+1} = d_t$. With other instruments, it can also be recovered from the data.  
\noindent As is standard in the dynamic choice literature, I assume individuals are rational, so that the observed transition density coincides with the transition densities expected by the individuals. Then, the transitions recovered from the data can be used to build the individual's expectations at each time $t$, and help recover the primitives. 

\subsubsection{Primitives}\label{section_structure}

Once the CCCs, CCPs and transitions are identified, I build upon existing literature to identify the primitives of the model \citep{hm1993, bmm1997, mt2002, escancianoetal2021}. I need to introduce additional structure on the covariates' transition and on the current utility function for nonparametric identification of the utility and the discount factor. \\
\indent \textbf{Budget constraint.} The asset transitions are given by the budget constraint\footnote{The budget constraint could be more sophisticated and include taxes, benefits... The effect of income could be a general function $\kappa_{d_t}(y_t)$, where non-working individuals would still receive a part of their income. See the simulations for another example with part-time versus full-time work.} 
\begin{equation}\label{eq_budget}
a_{t+1} = (1+r_t) a_t - c_t + y_t d_t. 
\end{equation} 
The asset here plays a different role than the other covariates. Indeed, the transition from $a_t$ to $a_{t+1}$ is directly impacted by the choice $c_t$ through the budget constraint (\ref{eq_budget}). 
Denote the covariates as $x_t = (\tilde{x}_t, a_t)$ to emphasize the distinct role of the asset.\footnote{Note that $y_t$ and $r_t$ are included in $\tilde{x}_t$, even though, in most applications, they will also be excluded from the current period utility. 
For notational simplicity and generality, I include them in $\tilde{x}_t$, which enters the current utility and represents all covariates other than $a_t$, i.e., all covariates whose transitions are not impacted by $c_t$ (Assumption \ref{cov_transi}). \vspace{-12pt}}
\begin{assumption}[Asset exclusion]\label{asset_exclusion}
The asset is excluded from the current period utility, i.e., $\partial u_{dt}(c_{dt}, \tilde{x}_t, a_t, \eta_t)/\partial a_t = 0$. 
\end{assumption}

\begin{assumption}[General covariates transitions]\label{cov_transi}
For all $\tilde{x}_t \in \tilde{\mathcal{X}}_t, d_t \in \mathcal{D},$ and $c_t \in \mathcal{C}$, $c_t$ does not impact the transitions of $\tilde{x}_t$ and $w_t$, i.e., 
\begin{align*}
f_{\tilde{X}_{t+1}, W_{t+1} | \tilde{x}_{t}, c_t, d_t}(\tilde{x}_{t+1}, w_{t+1}) = f_{\tilde{X}_{t+1}, W_{t+1} | \tilde{x}_{t}, d_t}(\tilde{x}_{t+1}, w_{t+1}).
\end{align*}
\end{assumption}
%
%
%
%
%
%

\begin{assumption}[Stationary utility]\label{statio_uti}
The current period utility is independent of time, i.e., 
$u_{dt}(c_{dt}, x_t, \eta_t) = u_d (c_{dt}, x_t, \eta_t)$. 
\end{assumption}
\begin{assumption}[Monotone utility]\label{monotone_u_c}
The utility is strictly increasing in $c$, i.e.,
\begin{align*}
\frac{\partial u_{dt}(c_{dt}, x_t, \eta_t)}{\partial c_{dt}} > 0  \quad  \text{ for all } (d, c_{dt}, x_t, \eta_t) \in \mathcal{D} \times \mathcal{C}_{dt} \times \mathcal{X}_t \times [0, 1]. \\ 
\end{align*}
\end{assumption}

\vspace{-24pt}

\begin{Lemma}[Marginal utilities and discount factor]\label{Lemma_escanciano}
Following \cite{escancianoetal2021}, under Assumptions \ref{additivebis}-\ref{monotonebis} and \ref{ass_shocks}-\ref{monotone_u_c}, the conditional marginal utilities at the optimal continuous choices, 
\begin{align*}
u'^*_{d}(x_t, \eta_t) = \frac{\partial}{\partial c_{dt}} u_{d}(c_{dt}, x_t, \eta_t) |_{c_{dt} = c^*_{dt}(\eta_t, x_t)},
\end{align*}
and the discount factor $\beta$ are nonparametrically point identified by the Euler equation for all $(d, x_t, \eta_t) \in \mathcal{D} \times \mathcal{X}_t \times [0, 1]$.
\end{Lemma}

\begin{proof}
Since $a_t$ is excluded from $u_d(\cdot)$ by Assumption \ref{asset_exclusion}, $u'^*(\cdot)$ only depends on $a_t$ through the optimal CCCs, . 
Then, given the budget constraint and since $c_t$ only affects the asset transition (Assumption \ref{cov_transi}), the \textit{Euler equations} are,  for  $d_t=0,1$, 
\begin{equation}\label{euler_marg}
u'^*_{d_t}(x_t, \eta_t) = \ \beta (1+r_t) \mathbb{E}_t \Big[  u'^*_{D_{t+1}}(X_{t+1}, \Eta_{t+1}) \ \Big| X_t=x_t, C_{t}=c^*_{dt}(\eta_t, x_t), D_t=d_t \Big]. 
\end{equation}
\noindent We have a system  of two equations with two unknown functions $u'^*_0(\cdot)$ and $u'^*_1(\cdot)$ (and the unknown discount factor $\beta$). Hence the importance of stationarity (Assumption \ref{statio_uti}), since otherwise we would have a different unknown function on each side of the equation. 
Now, under Assumptions \ref{monotonebis} and \ref{monotone_u_c}, the optimal marginal utilities are positive, 
\begin{align*}
u'^*_d(x_t, \eta_t)  > 0 \quad \text{ for all } d, x_t, \eta_t.
\end{align*}
Now, Theorem $2$ of \cite{escancianoetal2021} shows that the discount factor $\beta$ and the marginal utility functions are nonparametrically globally point identified by the system of Euler equations (\ref{euler_marg}). 
\end{proof}

\indent Once the marginal utilities are identified, I follow \cite{bmm1997} to identify the conditional value functions. Note that even though the marginal utilities are stationary, we still have a non-stationary problem because the conditional value functions are time-dependent with finite horizon.  

\begin{Lemma}[\cite{bmm1997}]\label{Lemma_bmm}
Under Assumptions \ref{additivebis}-\ref{monotonebis} and \ref{ass_shocks}-\ref{monotone_u_c}, the conditional value functions at optimal choices, $v_{dt}(c^*_{dt}(\eta_t, x_t), x_t, \eta_t)$, are identified up to an unknown constant of integration $O_{dt}$ independent from the asset, i.e., 
\begin{align*}
v_{dt}(c^*_{dt}(\eta_t, x_t), x_t, \eta_t) = G_{dt}(\tilde{x}_t, a_t, \eta_t) + O_{dt}(\tilde{x}_t, \eta_t) \quad \text{ for all } d, x_t, \eta_t,
\end{align*}
where $G_{dt}$ and $K_{dt}$ are defined in the proof.
\end{Lemma}

\begin{proof}
We have the first order conditions, holding at optimal CCCs for all $d$: 
\begin{equation}\label{foc_opti2}
 \frac{\partial}{\partial a_t} v_{dt}(c_{dt}, \tilde{x}_t, a_t, \eta_t)  = (1+r_t) \frac{\partial}{\partial c_{dt}} u_{d}(c_{dt}, x_t, \eta_t) \ |_{c_{dt} =c^*_{dt}(\eta_t, \tilde{x}_t, a_t)}.
\end{equation}
With $v^*_d(\cdot)$ as the conditional value function taken at the optimal continuous choice, we can rewrite the FOC as  
\begin{align}
\forall a_t: \quad \quad \frac{\partial}{\partial a_t} v^*_{dt}(\tilde{x}_t, a_t, \eta_t)  &= (1+r_t) \  u'^*_{d}(\tilde{x}_t, a_t, \eta_t).
\end{align}
Crucially, following Assumption \ref{asset_exclusion}, the asset is excluded from the current period utilities and marginal utilities. The identification strategy relies on this exclusion. Now, integration gives 
\begin{align*}
v^*_{dt}(\tilde{x}_t, a_t, \eta_t)  &= \int^{a_t}_0 \ (1+r_t) \  u'^*_{d}(\tilde{x}_t, a, \eta_t) \ da,
\end{align*}
where the lower bound $0$ is taken arbitrarily. 
Since $u'^*_d$ is identified, we can identify the optimal conditional value functions nonparametrically as
\begin{align*}
v^*_{dt}(\tilde{x}_t, a_t, \eta_t) = G_{dt}(\tilde{x}_t, a_t, \eta_t) + O_{dt}(\tilde{x}_t, \eta_t),
\end{align*}
up to unknown constant of integration $O_{dt}(\tilde{x}_t, \eta_t)$, independent from $a_t$ and depending on the arbitrary lower bound of integration. 
\end{proof}

\indent Finally, by specifying a distribution for $\Epsilon$ (e.g., generalized extreme value), the differences in the additive terms of the utility, $\Delta m_t(\cdot) = m_{1t}(\cdot) - m_{0t}(\cdot)$, are \textit{semi-parametrically} identified. Indeed, the difference in total conditional values, $\Delta v_t^*(\cdot) + \Delta m_t(\cdot)$, are identified by the CCPs through \citeauthor{hm1993}'s (\citeyear{hm1993}) inversion. Thus, if I impose a normalization of the constant, e.g. $O_{dt} = 0$, $\Delta m_t(x_t, w_t, \eta_t)$ is identified. 

\section{Unobserved types}\label{section_types}

So far, I assumed no autocorrelation in purely transitory $\Eta_t$ (Assumption \ref{cond_indep}), in order to be able to use the previous discrete choice as a relevant instrument ($W_t=D_{t-1}$) without violating the independence between the instrument and $\Eta_t$. 
However, including only iid transitory period-specific shocks is fairly restrictive in dynamic models, where we often observe serial correlation in the choices. I handle this by including permanent \textit{unobserved types} into the model, following the standard approach in the dynamic discrete choice literature \citep{am2011}. 
These types capture intrinsic latent differences between individuals, while $\Eta$ and $\Epsilon$ are transitory shocks affecting the decisions.
In this section, I show how to adapt the identification arguments with unobserved types, by identifying the unobserved types beforehand. 

\subsection{Identification with unobserved types}\label{subsection_identification_types}

I assume throughout this section that $W_t = D_{t-1}$.\footnote{If $W_t$ is not $D_{t-1}$ and is a period-$t$ variable, then the identification of unobserved types still holds. It is simplified and only requires $T \geq 3$ time periods in the panel, as in Section 3.1 of \cite{kasaharashimotsu2009}. }  
I observe panel data $\{D_t, C_{t}, X_t\}_{t=1}^{T}$ with $T\geq 6$ and $C_t = C_{0t} (1-D_t) + C_{1t} D_t$. The instrument $W_t=D_{t-1}$ is included in the observations of $\{D_t\}_{t=1}^{T}$.
Each individual has a time-invariant/permanent type $\mu$ with finite values $m \in \{1, ..., M\}$. 
The type is unobserved by the researcher. The probability of belonging to type $m$ is $\text{Pr}(\mu=m) = \pi^m$ and is time-invariant and independent of the covariates.\footnote{The setup can be extended to allow for time-varying types (e.g., first-order Markov), time-varying type probabilities, as well as type probabilities that depend on the covariates, using \cite{kasaharashimotsu2009} and \cite{hushum2012}.} 

\subsubsection{Adaptation of the framework with types} 

The adjustments to include types are fairly straightforward. Types act similarly to a covariate in $X$, except that it is unobserved by the researcher.  
The functions $u_{dt}(c_{dt}, x_t, m, \eta_t)$, $m_{dt}(x_t, m,  w_t, \eta_t)$, $V_t(x_t, m, w_t)$, and $v_{dt}(c_{dt}, x_t, m, \eta_t)$ are all \textit{type-dependent}, and I now make this dependence explicit by writing them with an $m$ supperscript as $u_{dt}^m(\cdot), m_{dt}^m(\cdot)$, $V_t^m(\cdot)$ and $v_{dt}^m(\cdot)$. Assumption \ref{ass_shocks} now conditions on $X=x$ and $\mu = m$. For simplicity, the covariate transition densities are assumed to remain type-independent: $f_{X_t | x_{t-1}, c_{t-1}, d_{t-1}}(x_t)$ for all $m$.\footnote{Again, this can be relaxed and we can identify type specific transitions $f_t^m(x_t | X_{t-1}=x_{t-1}, C_{t-1}=c_{t-1}, D_{t-1}=d_{t-1})$, following Section $3.2$ of \cite{kasaharashimotsu2009}.} 
\noindent I only add one assumption on how types enter the model. 
\begin{assumption}[Type-independent Support]\label{type_support}
Types enter the utilities of the model in a way such that, for all $t, c_t, d_t, x_t$, 
\begin{align*}
f_{D_t, C_t | x_t, d_{t-1}}(d_t, c_{t}) > 0 \iff f^m_{D_t, C_t | x_t, d_{t-1}}(d_t, c_{t}) > 0, \text{ for all } m=1,...,M.
\end{align*}
\end{assumption}
Here, $f^m_{D_t, C_t | x_t, d_{t-1}}(\cdot)$ is the type-dependent conditional joint density. 
Assumption \ref{type_support} restricts how types enter the utility functions: it must not affect the support of the optimal choices, especially the continuous one. 
%
%
In terms of identification, it means that any possible observation $(d_t, c_{t}, x_t)$ can come from any type $m \in \{1, ..., M\}$.   

\subsubsection{Identification of the type-dependent conditional joint densities}

\indent Given the framework, the joint densities of the choices depend on $X_t=x_t, W_t=d_{t-1}$ and now also depend on the type $\mu=m$. So the reduced form to identify the optimal choices are now type-specific and not directly observable from the data: I need to identify them first to identify the optimal choices following Section \ref{section_identification} afterwards. \\
\vspace{-12pt}

\noindent \textbf{Result}:
	 If $T \geq 6$, the \textit{type probabilities}, $\pi^m$, and the \textit{type-dependent conditional joint densities}, $f^m_{D_t, C_t | x_t, d_{t-1}}(d_t, c_t)$,  
are identified from observed serial data $\{D_t, C_{t}, X_t\}_{t=1}^{T}$ for all $m, t, d_t, c_{t}, x_t, d_{t-1}$. \\
%
%
%
%
%
%
%
\vspace{-12pt}

\indent The idea is to identify the unobserved types by using the identification power of the observed serial correlations of $\{D_t, C_{t}, X_t\}_{t=1}^{T}$. 
Notice that, except for the first-order autocorrelation between $d_t$ and $d_{t-1}$ (relevance condition), I did not use the observed autocorrelations of the choices to identify the dynamic model before. This is the reason why I can exploit them to identify the type-specific conditional joint densities in a first step, independent of the rest of the identification (which proceeds period by period). 
Formally, to show the identification of the type-dependent conditional joint densities, I extend the identification proof of \cite{kasaharashimotsu2009} to joint choices with both time-dependent conditional choice probabilities and a lagged dependent variable. I also make specific adjustments because my covariates include the value of assets which has a deterministic transition given the choices, violating Assumption $1 (c)$ in \cite{kasaharashimotsu2009}. 
The identification proof is given in Appendix \ref{appendix_types}.

%


\subsection{Identification of the dynamic model with unobserved types}

I obtain the type-dependent reduced-form functions from the type-dependent joint choices densities for all $m \in \{1, ..., M\}$
\begin{align*}
R^m = \Big\{ &\textrm{Pr}^m(D=d_t|X_t=x_t, W_t=w_t), (d_t, x_t, w_t, t) \in \mathcal{D}\times\mathcal{X}_t\times\mathcal{W}\times \{2,..., T\} , \\
 &F^m_{C_{dt} | d_t, x_t, w_t}(\cdot), (d_t, x_t, w_t, t) \in \mathcal{D}\times\mathcal{X}_t\times\mathcal{W}\times \{2,..., T\} \Big\}. 
\end{align*}
\noindent From these $R^m$, following Section \ref{section_identification}, I identify type-dependent CCCs and CCPs
\begin{align*}
c^{m*}_{dt}(\eta_t, x_t) \text{ and } \textrm{Pr}^m(D_t = d_t | \Eta_t=\eta_t, X_t = x_t, W_t=w_t),  
\end{align*}
for all $m \in \{1, ..., M\}$ and $(d_t, \eta_t, x_t, w_t, t) \in \mathcal{D}\times [0, 1] \times \mathcal{X}_t\times\mathcal{W}\times \{2,..., T\}$. \\ 
\indent The transitions are type-independent by assumption, so I identify them directly from the data as before. 
Then, 
the identification of the primitives of the dynamic model follows Section \ref{section_structure}, replacing the optimal choices by their type-dependent counterparts, and conditioning everything on the type $m$. 


\section{Estimation}\label{section_estimation}

I build a two-step estimation procedure.  
In the first stage, I estimate type-dependent conditional continuous choices (CCCs) and conditional choice probabilities (CCPs). First, I estimate the type probabilities using an expectation-maximization (EM) algorithm \citep{aj2003, am2011}. 
Then, I estimate the type-dependent optimal choices given these type probabilities. This step is data-driven and is independent of the structural model specification. 
In the second stage, I use these estimated optimal policies to estimate the primitives (structural parameters) of the model. To do so, I use the fact that the optimal choices are obtained via the optimality conditions of the model taken at the true parameters: the true parameters are the only parameters that generate these optimal policies, and satisfy the optimality conditions taken at these true policies. 
\noindent The estimation is analogous to that of \cite{hm1993}, \cite{am2011} and \cite{hmss1994} but extended to discrete-continuous choice models. \\
\indent The main appeal of this estimation is computational gains. By estimating the optimal choices only once, directly from the data, and taking them as given in the next stage, the computational burden of the estimation is significantly reduced. Indeed, one does not need to solve for the value function or the likelihood for each new set of selected parameters. 
This allows us to estimate models that were previously computationally intractable. 
I describe the estimation method in this section, and show the estimator's performance using Monte Carlo simulations in Section \ref{section_comparison_estimation}. 

\subsection{1st step: conditional choices}

Use $W_t = D_{t-1}$ as suggested before. 
First, I estimate the type-independent covariates transitions directly from the data. Asset transition is known by the budget constraint, instrument transition is known since $w_t = d_{t-1}$, and the other covariate transitions can be estimated using auto-regressive processes of order $1$. This yields the transitions  $\hat{f}_{X_t | x_{t-1}, c_{t-1}, d_{t-1}}(x_t)$, where $c_{t-1}$ only affects the asset transition (Assumption \ref{cov_transi}). 
Then, I estimate the type-dependent reduced forms using an expectation-mazimization (EM) algorithm in the spirit of \cite{am2011} (Section \ref{subsection_EM}). Using these type-specific probabilities for each individuals, I estimate type-specific CCCs and CCPs building upon the identification arguments (Section \ref{subsection_cccestimate}). 
%
%

\subsubsection{EM algorithm for type-dependent reduced forms}\label{subsection_EM}

Suppose the type-dependent joint densities $f^m_{D_t, C_t | x_t, d_{t-1}}(d_t, c_t, \theta_r)$ 
are fully parametrized by $\theta_r$.\footnote{Note that this includes the initial period joint density $f_{D_1, C_1, X_1}^m(d_1, c_1, x_1, \theta_{\text{init}})$ when the instrument $D_{0}$ is unobserved.} 
We want a nonparametric sieve-estimator where the number of parameters in $\theta_r$ increases with the sample size. To estimate $\theta_r$ and $\pi^m$ (the type probabilities) we proceed by iteration, starting from an initial guess $(\theta_r^{(k)}, \pi^{(k)})$, with $k=1$, where $\pi^{(k)} = \{{\pi^{m}}^{(k)}\}_{m=1}^M$. \\
%
%
%
%

\noindent \textbf{Expectation step.} 
Given the $k^{th}$ guess, the likelihood of observing $\{ d_{it}, c_{it}, x_{it}\}_{t=1}^T$ given the type $m$ for individual $i$ is
\begin{align*}
	L^m_i(\theta_r^{(k)}) = f_{D_1, C_1, X_1}^m(d_{i1}, c_{i1}, x_{i1}, \theta_{\text{init}}^{(k)}) \prod^T_{t=2} f^m_{D_t, C_t | x_{it}, d_{it-1}}(d_{it}, c_{it}, \theta_r^{(k)}) \ 
	\hat{f}_{X_t | x_{it-1}, c_{it-1}, d_{it-1}}(x_{it})
\end{align*}
%
%
%
\noindent Then, the likelihood of observing $\{ d_{it}, c_{it}, x_{it}\}_{t=1}^T$ for $i$, unconditional on type, is
\begin{align*}
	L_i(\theta_r^{(k)}, \pi^{(k)}) &= \sum^M_{m=1} {\pi^m}^{(k)} L^m_i(\theta_r^{(k)}). 
\end{align*}
The updated likelihood that individual $i$ belongs to type $m$, denoted $q_i(m)$, is 
\begin{align*}
	q_i^{(k+1)}(m) = \frac{L^m_i(\theta_r^{(k)})}{L_i(\theta_r^{(k)}, \pi^{(k)})}.
\end{align*}
Given a sample of $N$ individuals, we update $\pi^{(k)}$ to $\pi^{(k+1)}$ for each type $m$ as 
\begin{align*}
	\pi^{m(k+1)}= \frac{1}{N} \sum^N_{i=1} q_i^{(k+1)}(m). 
\end{align*}

\noindent \textbf{Maximization step.} 
Given $q^{(k+1)}= \{q_i^{(k+1)}(m)$  for all $m\}_{i=1}^N$, we can compute the sample likelihood for any $\theta_r$:
\begin{align*}
	L(\theta_r, q^{(k+1)}) = \prod^N_{i=1} \sum^M_{m=1} q_i^{(k+1)}(m) L^m_i(\theta_r).
\end{align*}
We update $\theta_r^{(k)}$ to $\theta_r^{(k+1)}$ by finding the $\theta_r$ which maximizes the log-likelihood 
\begin{align*}
	\theta^{(k+1)} = \underset{\theta_r}{\text{argmax}} \ \textrm{log} \ L(\theta_r, q^{(k+1)}).
\end{align*}
Notice that the empirical conditional joint densities weighted by $q^{(k+1)}$ maximize the log-likelihood \citep[as in (5.9) of][]{am2011}. Thus we can directly nonparametrically estimate it, without running any numerical optimization algorithm. \\

\noindent \textbf{EM estimation.} 
Select initial values $(\theta_r^{(1)}, \pi^{(1)})$. For example, randomly assign a type to every individual and estimate the initial joint densities given this guess  to obtain the initial values.  
Starting from these initial values and iterating the expectation and maximization steps, the EM algorithm converges to $(\hat{\theta}_r, \hat{\pi})$, which maximizes the likelihood of the sample. The estimates $(\hat{\theta}_r, \hat{\pi})$ give estimates of the \textit{type-dependent reduced forms} $\widehat{R}^m$ and provide estimates of the \textit{type probabilities} of each individual, $\hat{q_i}(m)$, that we use in the next steps. 

\subsubsection{Type-dependent CCCs and CCPs}\label{subsection_cccestimate}

Once the type-dependent probabilities $\hat{q_i}(m)$ are estimated, we can use them to estimate the type-dependent optimal choices. \\ 

\noindent \textbf{Conditional continuous choices (CCCs).} 
We build upon the link between the intra-period problem and the IV-Quantile model of \cite{chernozhukovhansen2005} established in Section \ref{section_framework}, and adapt existing IVQR estimation procedures \citep{chernozhukovhansen2006, kaido2021decentralization} to estimate the CCCs. More precisely, since $W_t=D_{t-1}$ is a valid instrument, we estimate $c_{0t}^m(\Eta_t, x_t)$ and $c_{1t}^m(\Eta_t, x_t)$ for any rank $\Eta_t$ and for all period $t$, covariates $x_t$, and type $m$, by running the weighted IV-Quantile regression of $C_t$ on $D_t$ at each quantile $\Eta_t$, conditional on $X_t=x_t$ and weighted by the estimated type-$m$ probabilities $\hat{q_i}(m)$.\footnote{There are several manners to condition on the covariates $X_t$. Typically, for discrete covariates, we can run separate IV-quantile regressions on each subsamples with $X_t=x_t$, provided that these subsamples contain enough observations. Otherwise, continuous covariates (e.g., the assets) enter additively in the IVQR specification, which effectively restricts the heterogeneity of the effect of $D$ on $C$ with respect to these covariates at each quantile. In theory, we could also split the continuous covariate in subgroups with sufficiently enough observations. Or we could do kernel-based IV quantile regressions to account for these continuous covariates nonparametrically with weights, but this requires a large sample. } 
This IVQR approach allows to flexibly estimate heterogenous effects of $D$ on $C$ at each quantiles $\Eta$. \\

\noindent \textbf{Conditional choice probabilities (CCPs).} Once the CCCs are estimated, one can invert them to estimate the unobserved shock $\eta_{it}$ for every individual, i.e., 
\begin{align*}
	\hat{\eta}_{it} = \left({c_{d_{it}t}^m}\right)^{-1}\Big( c_{it}, x_{it}\Big). 
\end{align*} 
Using these estimated unobserved individual shocks $\hat{\Eta}_{it}$ as a generated covariate, we can directly estimate the type-dependent CCPs, $\textrm{Pr}^m(D_t | \hat{\Eta}_{t}, X_t, W_t)$, using weighted nonparametric kernels or flexible weighted logit/probit regressions of $D_t$ on $X_t, W_t$ and $\hat{\Eta}_t$, weighted by type-$m$ probabilities, $q_i(m)$. 

\subsubsection{Alternative estimation methods}
There are many alternative ways to estimate the optimal choices.  
For example, the CCCs can be estimated nonparametrically or semi-parametrically by building upon the identification arguments using empirical counterparts of the functions $\Delta F_{C_d}^m(\cdot)$, obtained using the estimated type-dependent reduced forms.  This alternative approach has the advantage of working under weaker relevance conditions than the ones imposed by \cite{chernozhukovhansen2005}, i.e., even with piecewise monotone $\Delta F_{C_d}^m$ functions, as in Figure \ref{fig:identification2}.  

\subsection{2nd step: structural model}

Suppose the primitives can be fully parametrized by $\theta = (\beta, \theta_C, \theta_P)$, where $\theta_C$ characterize the marginal utility with respect to the continuous choice and $\theta_P$ does not.\footnote{More precisely, assume there is a one-to-one mapping between the parameters $\theta$ and the primitives of the model, i.e., each different value of $\theta$ generates different primitives.} 
In other words, $u^m_d(\cdot, \theta_C)$ is parametrized by $\theta_C$, while $\theta_P$ only impacts the difference $\Delta m_{t}^m(\cdot, \theta_P)$.  Denote $\theta^0$ the true parameters that generated the data.  \\ 
%
%
\indent Using the nonparametric identification arguments developed previously (Section \ref{subsection_dynamic_identification}), there is a one-to-one mapping between the primitives of the model and the optimal choices (CCCs and CCPs). Each set of parameters $\theta$ characterizing the primitives of the model is associated with distinct optimality conditions (e.g., Euler equations and differences of conditional value functions) which, in turn, yield distinct optimal choices. Consequently, the true CCCs and CCPs which have been consistently estimated directly from the data in the first stage, can only be rationalized by the true value of the parameters, $\theta^0$. 
In theory, we could estimate the model using standard methods of simulated moments with these CCCs and CCPs as the moments.\footnote{An even more standard approach would be to use moments directly available in the data to estimate the model, e.g., observed quantiles of $C$ given $D$, $X$, and $W$, and estimated probability of selecting $D$ given $X$ and $W$. The key take-away from the identification being that one needs to use moments which depends on the instrument $W$, otherwise the model would not be identified. } A typical method of simulated moment estimator would be as follows: for each value of the parameter $\theta$, one would compute the optimal value function and the corresponding theoretical optimal choices. Then, the estimated $\theta$ would be the set of parameters which make these theoretical optimal choices the closest to the true observed optimal choices moments estimated in the first stage. 
\noindent While theoretically simple, this standard simulated method of moments is impractical for dynamic models. Indeed, even the fastest methods to compute the value functions, namely the endogenous grid method \citep[EGM, ][]{carroll2006, iskhakov2017}, is still long, even for only dynamic discrete choice models, and even more so for dynamic discrete \textit{and} continuous choice models which require additional numerical optimization to solve for the optimal continuous choices. \\
\indent Fortunately, extending what \cite{hm1993, hmss1994, am2011} have proposed to estimate dynamic discrete choice models, I propose a faster alternative estimation method that does not require to compute the value function and numerically solve for the optimal choices for each evaluated set of parameters. The key intuition is to directly use the link between the optimality conditions of the model and the optimal choices. Given the known (estimated in the first stage) optimal choices, the first order conditions are only satisfied for the true value of the parameters, $\theta^0$. So, we estimate these true parameters by minimizing the error in the first order conditions where we plugged-in the known optimal choices. The computational difficulty is that these first order conditions involve expectations about the future. In order to compute these, we use forward simulations, as  \cite{hmss1994} did for dynamic discrete choice models. I split the estimation into two types of first order conditions (i) the Euler equation which determines the CCCs and will allow to estimate $\theta_C^0$ and $\beta^0$, and (ii) the conditional value function comparison which determines the CCPs. The full estimation is described below. \\ 


\noindent \textbf{Moment selection.} Select a set of $S$ moments corresponding to $S$ covariates values: $\{ t^s, \eta^s, x^s, w^s, m^s\}_{s=1}^S$. The CCCs and CCPs have been consistently estimated in the first stage, so these moments corresponds to moments expressed in terms of $C$ and $D$, i.e., 
$c_d^s = c_{dt^s}^{m^s}(\eta^s, x^s)$ for each $d \in \mathcal{D}$ and $p_d^s = \textrm{Pr}^{m^s}(D_{t^s}=d | \Eta_t = \eta^s, X_t=x^s, W_t = w^s)$. The set of moments needs to be large enough such that there is a one-to-one mapping between the model parameters $\theta$ and all the moments.\footnote{It is possible that two different set of parameters are observationally equivalent locally, for some CCCs and CCPs taken at specific values of the covariates and type. However, if the model is properly parametrized (with no "redundant" parameters), there do not exist two distinct sets of parameters that yield observationally equivalent CCCs and CCPs for every values of $t, X, W,$ and $m$. One cannot test every value of these covariates as moments, but one needs to take sufficiently many different moments such that there is only one optimal set of parameters that generate them.}  If the model specification is correct, the optimal choices estimated in the first stage are generated by the true parameters, $\theta^0$. Furthermore, there is a one-to-one mapping between the model and the optimal choices, and these observed moments can only be rationalized by the true $\theta^0=(\beta^0, \theta_C^0, \theta_P^0)$, and no other value of $\theta$. \\ 


%
%
%

\noindent \textbf{Euler objective, estimation of $(\theta_C, \beta)$.}  
Recall that under the true model $\theta^0$, Euler Equation \eqref{euler_marg} holds, i.e., for any moment $s$ with $\{t^s, \eta^s, x^s, w^s, m^s\}$,\footnote{For the Euler equation, the value of the instrument, $w^s$, in the list of moment does not matter  because it is excluded from the optimal CCC.} and each $d \in \mathcal{D}$, \\
\vspace{-3em}
\begin{adjustwidth}{-0.5cm}{-0.5cm}  
\begin{align}\label{eq:euler_est}
& u^{{m^s}'^{**}}_{d}(x^s, \eta^s, \theta_C^0) = \  \beta (1+r_{t^s}) \mathbb{E}_{t^s} \Big[  u^{{m^s}'^{**}}_{D_{t^s+1}}(X_{t^s+1}, \Eta_{t^s+1}, \theta_C^0)  \ \Big| X_{t^s}=x^s, C_{dt^s}=c^{m^s}_{dt}(\eta^s, x^s), D_{t^s}=d_t \Big] \nonumber \\
&\overset{def}{\iff} \quad \quad b_{1}(d, t^s, \eta^s, x^s, m^s, \theta_C^0)\  = \ b_{2}(d, t^s, \eta^s, x^s, m^s, \theta_C^0, \beta^0).
\end{align}
\end{adjustwidth}
The functions $u^{{m^s}'^{**}}_{d}(\cdot, \theta_C)$ are the marginal utility taken at the optimal choices, \textit{taking the optimal choices as estimated in the first stage}, $c^{m}_{dt}(\eta_t, x_t)$, i.e., 
\begin{align*}
u^{{m}'^{**}}_{d}(x_t, \eta_t, \theta_C) = \frac{\partial}{\partial c_{dt}} u^m_{d}(c_{dt}, x_t, \eta_t, \theta_C) \Big|_{c_{dt} = c^{m}_{dt}(\eta_t, x_t)}.
\end{align*}
Regardless of the parameters $\theta_C$, the function is evaluated at the true $c^{m}_{dt}(\eta_t, x_t)$ which correspond to the true parameters, $(\theta_C^0, \beta^0)$.
This Euler equation uniquely determines the CCCs. Given the nonparametric identification, and the uniqueness of the mapping between the optimal choices and the primitives of the model, there is no alternative set of parameters $\tilde{\theta}_C \neq \theta_C^0$ such that this equation \eqref{eq:euler_est} would hold for all moments $s$. This is because, here we plugged-in the optimal choices of the first stage which correspond to $\theta^0$, and any distinct set of parameters $\theta$ would require different CCCs and CCPs in order to hold for all $s$, due to the uniqueness of the optimum. 
As a consequence, the idea behind the estimation is to minimize the difference between both sides of the Euler equation, $b_{1}$ and $b_{2}$, taken at the optimal choices estimated in the first stage.\footnote{Equivalently, recall that the marginal utilities are strictly increasing in $C$ by monotonicity. So one could express the Euler equation not in terms of the marginal utilities, but in terms of the optimal $C$ they determine. Then, the objective of the estimator is to find the true values of $\theta_C$ and $\beta$ which minimize the difference between the estimated CCCs in the first stage, and the corresponding theoretical CCCs pinned down by the Euler equation for any moment $s$. } While the left hand side of the Euler equation, $b_{1}(\cdot)$, can be directly estimated consistently for any $\theta_C$ by plugging in the first stage optimal CCCs, the right hand side $b_2(\cdot)$ contains an expectation over the next period optimal marginal utilities. To compute it, we use \textit{one-period ahead forward simulations}, using the estimated covariates transitions and the next-period type-dependent optimal choices (CCCs and CCPs) estimated in the first stage. 
Then $(\theta_C, \beta)$ can be estimated by minimizing the sum of squared differences $\widehat{b_{1}}(\cdot) - \widehat{b_{2}}(\cdot)$  over all moments $s$ and alternative $d$, i.e.,   
\begin{align}\label{eq:euler_criterium}
(\hat{\theta}_C, \hat{\beta}) = &\ \underset{\theta_C, \beta}{\textrm{argmin}} \sum_{s=1}^S \sum_{d\in \mathcal{D}} \Big( \widehat{b_{1}}(d, t^s, \eta^s, x^s, m^s, \theta_C) - \widehat{b_{2}}(d, t^s, \eta^s, x^s, m^s, \theta_C, \beta)  \Big)^2. 
\end{align}
\vspace{-1em}

\noindent This `Euler objective' \eqref{eq:euler_criterium} consistently estimates $\theta_C$ and $\beta$. Indeed, since the optimal choices are consistently estimated in the first stage, the minimum should only be reached at the true values of the primitive parameters $\theta_C^0$ and $\beta^0$ according to the Euler equation \eqref{eq:euler_est}.  \\

\noindent \textbf{Probability objective, estimation of $\theta_P$.} 
Following a similar intuition, we can semi-parametrically estimate the remaining parameters $\theta_P$ impacting the differences of the additive term, $m_d(\cdot)$ but not the marginal utility with respect to the continuous choice. Given a known distribution of $\Epsilon$, there is a one-to-one mapping between the CCPs and the conditional value functions which are determined by the primitives of the model, $\theta^0$. This is  the \citeauthor{hm1993}'s (\citeyear{hm1993}) inversion. The only adjustment with respect to \cite{hm1993} is that the mapping is with respect to the conditional value functions taken at the optimal continuous choice, denoted $v_{dt}^{m*}(\cdot)$. For example, if $\Epsilon$ is extreme-value type I, for any moment $s$ with with $\{t^s, \eta^s, x^s, w^s, m^s\}$, we know that the CCPs estimated in the first stage satisfy 
\begin{align}\label{eq:formulaccp}
	\textrm{Pr}^{m^s}(D_{t^s} = d &| \Eta_t = \eta^s, X_t = x^s, W_t = w^s) \nonumber \\
	&= \frac{\textrm{exp}\big( v^{{m^s}*}_{d{t^s}}(x^s, \eta^s, \theta^0) + m^m_{d{t^s}}(x^s, w^s, \eta^s, \theta^0) \big)}{\sum^{J}_{j=0} \textrm{exp}\big( v^{{m^s}*}_{j{t^s}}(x^s, \eta^s, \theta^0) + m^m_{j{t^s}}(x^s, w^s, \eta^s, \theta^0) \big) }  \nonumber \\
	& := p^{model}(d, t^s, \eta^s, x^s, w^s, m^s, \theta^0) \quad  \text{ for all } d \in \mathcal{D},
\end{align}
where $v_{dt}^{m*}(\cdot, \theta)$ are the \textit{conditional values taken at the true optimal choices estimated in the first stage} but with parameter $\theta$.\footnote{Attention, as for the optimal marginal utilities, except if $\theta = \theta^0$, these are \textit{not} the traditional `conditional values functions'. This is because, we plug-in the optimal choices estimated in the first stage (which correspond to the true parameters $\theta^0$) in all the future periods, and not the optimal choices corresponding to $\theta$. The choice of $\theta$ only affects how much utility is derived from these already given optimal choices every period. } Since \citeauthor{hm1993}'s (\citeyear{hm1993}) mapping is unique, with the choices estimated in the first stage, \eqref{eq:formulaccp} only holds at the optimal value of the parameters, $\theta^0$ for all moments $s$. We use this known link between the CCPs and the primitives of the model to estimate $\theta=(\beta, \theta_C, \theta_P)$. In fact, we take $\widehat{\theta}_C$ and $\widehat{\beta}$ estimated via the Euler equation as given, and use \eqref{eq:formulaccp} to estimate the remaining parameters, $\theta_P$.\footnote{In theory, one could estimate all the parameters $\theta$ using only this probability criterium. I recommend to split in two separate estimation steps because the parameters $\theta_C$ and $\beta$ have a larger impact on the Euler equation, and are more precisely estimated using the Euler-criterium.  Moreover, the Euler equation estimation is faster because it involves only one-period ahead simulations, while the probability estimation requires forward simulations of the complete remaining life-cycle. }  \\
\indent In order to estimate $p^{model}(\cdot, \hat{\beta}, \hat{\theta}_C, \theta_P)$ for any $\theta_P$, we need to compute these conditional value functions at the true optimal choices for any moment $s$. These conditional values contain expectations about the next period value function, and thus about the entire future life-cycle of individuals from time $t^s$ onwards, taking the true optimal choices (estimated in the first stage) as given. In order to estimate these expectations without solving for the value function, I follow the insights of \cite{hmss1994} and use forward simulations of the entire remaining life-cycle of individuals.\footnote{\textit{Forward simulation details.} 
For any $d \in \mathcal{D}$, recall that the conditional value functions at the optimal choices estimated in the first stage are defined as
\begin{align}\label{eq:esti_optivd}
	&v_{dt^s}^{m^s*}(x^s, \eta^s, \theta) \\ 
	&= u_{dt^s}^{m^s}(c_{dt^s}^{m^s}(\eta^s, x^s), x^s, \eta^s, \theta_C) + \beta \mathbb{E}_{t^s}[V^m_{t^s+1}(X_{t+1}, W_{t+1}, \theta) | X_{t^s}=x^s, C_{t^s}=c_{dt^s}^{m^s}(\eta^s, x^s), D_{t^s}=d ], \nonumber
\end{align}
where for any $(t, m, x_t, \eta_t)$, the value functions given the optimal first stage choices are given by
\vspace{-1em}
\begin{align}\label{eq:esti_value}
&V_t^{m}(x_t, w_t, \theta) = \mathbb{E}_t\Big[ \sum^T_{\tau = t} \beta^{\tau-t} [ \ u^m_{D^*_\tau\tau} (c^{*m}_{d^*\tau}(\Eta_\tau, X_\tau), X_\tau,  \Eta_{\tau}, \theta_C ) + m^m_{D^*_\tau}(X_\tau, W_\tau, \Eta_{\tau}, \theta_P) + \Epsilon_{D^*_\tau\tau} ] \Big],  \\ 
&\text{where } \quad D^*_\tau =  d \text{ with first stage probability } \textrm{Pr}^{m}(D_{\tau} = d | \Eta_\tau, X_\tau, W_\tau). \nonumber 
\end{align}
These value functions at the optimal choices can be estimated by forward simulating the entire life-cycle of individuals from $t$ onwards. 
For a large number of simulations, $N_S$, we draw new state variables and unobserved $\Eta_\tau$ and $\Epsilon_\tau$ each period, given the previous state variables and choices. Given these new state variables, we use the true optimal choices which have been consistently estimated in the first stage, to draw the discrete and continuous choices of the period. We repeat this process every period until period $T$ is reached (or, if the horizon is infinite, until the discount is so large that the additional period has negligeable impact on the value). Then, given the entire pre-simulated histories of $\{ \{ C^k_\tau, D^k_\tau, X^k_\tau, W^k_\tau, \Eta^k_\tau, \Epsilon^k_\tau \}_{\tau = t+1}^T \}_{k=1}^{N_S}$ for $N_S$ different simulations, for any value of $\theta$ we can compute the corresponding period utility every period, and thus, the value of every simulated life-cycle. Taking the average of these values over the $N_S$ simulations provides an estimate of the value function given $\theta$, described in Equation \eqref{eq:esti_value}. Once the value functions are estimated, we can recover the conditional value functions \eqref{eq:esti_optivd}, and then the theoretical CCP, $p^{model}(\cdot, \theta)$ for any parameter $\theta$ and moment $s$. 
}
Thanks to these forward simulation, we can estimate $\widehat{p^{model}}(\cdot, \widehat{\beta}, \widehat{\theta}_C, \theta_P)$ for any $\theta_P$ and for all moments $s$. Then, we estimate $\theta_P$ as the unique set of parameters satisfying the theoretical property \eqref{eq:formulaccp}. In practice, we estimate $\theta_P$ by minimizing the sum of squared differences over all moments $s$, between the optimal CCPs which have been consistently estimated in the first stage, and their theoretical counterpart in the model, $\widehat{p^{model}}(\cdot, \theta_P)$, i.e.,\footnote{Note that we take the sum over all $d \in \mathcal{D}$ except the reference $d=0$, since all probabilities sum to one so one of the alternative is redundant.} 
\begin{align}\label{eq:probacriterium}
\hat{\theta}_P = \underset{\theta_P}{\textrm{argmin}} 
\sum_{s=1}^S \sum_{d \in \mathcal{D} \setminus \{0\}} \Big( &\widehat{\textrm{Pr}}^m(D_{t^s}=d | \Eta_{t^s} = \eta^s, X_{t^s} = x^s, W_{t^s} = w^s) \nonumber \\
&- \quad \widehat{p^{model}}(d, t^s, \eta^s, x^s, w^s, m^s, \widehat{\beta}, \widehat{\theta}_C, \theta_P) \Big)^2. 
\end{align}
%
%
This probability objective \eqref{eq:probacriterium} consistently estimates the remaining parameters $\theta_P$ as the minimum should only be reached at the true values of $\theta_P^0$ according to \eqref{eq:formulaccp}. \\


\noindent \textbf{Additional remarks.} This second stage estimation is completely independent of the data, the data only affected the estimation of the optimal policies in the first stage. In this sense, it is similar to the method of simulated moments where the data only affects the estimation of the moments. Given the first stage estimated policies, the second stage estimates only depend on the selected set of moments and on the number of forward simulations used to estimate the expectations in the Euler equation and in the conditional value functions. Increasing the number of simulations increases the precision at the cost of increased computational time. \\
\indent Notice also that the CCCs and CCPs estimated for all $t, \Eta_t, X_t, W_t, m$, affect the estimation, and not only the ones at the selected moments. This is because, in the forward simulations, the entire range of covariates and $\Eta$ can be drawn. So the first stage optimal choices need to be well estimated at any $\Eta_t$, even at the tails. 




\section{Estimator performance}\label{section_comparison_estimation}

I illustrate the estimator's performance with Monte Carlo simulations of the estimation of a parametric toy model of simultaneous labor and consumption choices. 
\subsection{Toy model} 
\noindent \textbf{Period utility.} 
Each period from age $t=1$ to $T$, 
individuals choose to work full time ($D_t=1$) or part time ($D_t=0$) and to consume ($C_t$). Their period-$t$ utilities are 
\begin{align*}
	u_{dt}^m(c_t, \eta_t) + m_{dt}^m(w_t, \eta_t) + \epsilon_{dt},  
\end{align*}
where $u_{dt}$ is the `main utility' function which depends on the consumption choice, $m_{dt}$ is the `additive part of the utility', which does not depend on the consumption but depends on the instrument (lagged labor choice, $W_t=D_{t-1}$ here), and $\Epsilon_{dt}$ (taking values $\epsilon_{dt}$) are extreme-value type I additive idiosyncratic shocks impacting the preferences for full time work. 
Both parts of the utility depend on the individuals' time-invariant types $\mu$ taking values $m \in \{1, 2\}$, with $\text{Pr}(\mu=1) = \pi^1$. These types are unobserved by the econometrician, but known by the individuals. \\
\indent We parametrize the main utility, $u(\cdot)$ as a CES utility function 
\begin{align*}
	u_{dt}^m(c_t, x_t, \eta_t) = \left(\frac{c_t}{exp(\eta_t)}\right)^{1-\sigma_d^m}/(1-\sigma_d^m), 
\end{align*}
where the $\sigma_d^m$ represents type $m$ and discrete choice $d$-specific risk aversion or intertemporal elasticity of substitution. We allow this main risk aversion to vary with the labor tenure decision. In this model, contrary to standard CES models, the marginal utilities of consumption are heterogenous (even at fixed covariates, $X$, and type, $m$) because they depend on idiosyncratic preference shocks, $\Eta_t$, taking values $\eta_t \in [0, 1]$. As in the general model, $m$ captures the permanent differences (types) between individuals, while $\Eta_t$ captures period-specific transitory idiosyncratic preference shocks. $m$ and $\Eta_t$ are iid for every individuals. Both $m$ and $\Eta_t$ are unobserved by the econometrician.    \\
\indent The additive part of the payoff is given by
\begin{align*}
	m_{0t}^m(w_t, \eta_t) &= 0, 
	\quad \text{ and } \quad m_{1t}^m(w_t, \eta_t) = \gamma^m + \alpha^m \eta_t + \omega^m (1-w_t), 
\end{align*}
such that $m_{1t}^m - m_{0t}^m = m_{1t}^m$ represents the additive utility gains (or cost) of choosing to work full time ($D_t=1$) compared to working part time ($D_t=0$). 
More precisely, we model this cost as a linear function of the unobserved idiosyncratic preference for consumption, $\Eta_t$, where $\gamma^m$ represent the intercept/constant gain of working full time for individuals with $\Eta_t = 0$, while $\alpha^m$ represent the slope of how this gain changes with $\Eta_t$. The parameters $\alpha^m$ captures additional complementarity/substitutability between consumption and labor, in addition to the ones implicitly present in the main utility $u(\cdot)$. 
Finally, $\omega^m$ represent the utility switching cost (if $\omega^m$ is negative) endured by previously part time workers $(w_t = d_{t-1} = 0$) who switch to full time ($d_t=1$). Again, all the parameters $\gamma^m, \alpha^m, \omega^m$ are type-specific. \\
\indent The model is dynamic and the individuals discount their future utility with a factor $\beta$.\footnote{Instead, we could have specified a type-specific discount factor, $\beta^m$. It would also be identified and precisely estimated.} Thus, the main parameters entering the Euler equations and affecting the consumption choices are $\beta$ and the parameters entering $u_{dt}^m(\cdot)$, i.e., $\theta_C = (\sigma_0^1, \sigma_0^2, \sigma_1^1, \sigma_1^2)$, while the parameters $\theta_P = (\gamma^1, \gamma^2, \alpha^1, \alpha^2, \omega^1, \omega^2)$ only affect the labor supply choices. The parameters $\theta = (\beta, \theta_C, \theta_P)$ describe the primitives of the models and represent the main parameters we want to estimate. \\ 

\noindent \textbf{Dynamics and covariates transition.} 
\noindent The asset, $a_t$, evolves according to the budget constraint 
\begin{align*}
	a_{t+1} = (1+r) a_t - c_t + (d_t + 0.5(1-d_t))y_t,
\end{align*}
where $y_t$ represents the full-time equivalent yearly income of individuals. Individuals who work full time ($d_t=1$) obtain $y_t$, while individuals who work part time ($d_t=0$) obtain half of it. The income $y_t$ is a random variable, taking only two values for simplicity: $y_L$ and $y_H$, for low and high income, respectively. The income transition is given by \begin{align*}
\textrm{Pr}(y_1 = y_H | D_0 = d_0, Y_0=y_0) = \Pi(d_0, y_0) = 
\begin{pmatrix} 
\pi_{0L} & \pi_{0H} \\
\pi_{1L} & \pi_{1H} 
\end{pmatrix},
\end{align*}
and is directly estimated from the observed data on income. \\
\indent Asset and income are the only two observable covariates, i.e., $X_t = (a_t, y_t)$.\footnote{One could easily complexify this model, by adding more individual characteristics, a more complex income process, type-dependent budget constraint, more realistic pension plans for the retirees... I choose to model only the key features of a standard life-cycle model, as it is sufficient to illustrate the performance of the estimator in terms of precision and computation time. } 
Even though the utility does not directly depend on asset and income, the optimal consumption and labor choices depend on these through the dynamics of the problem. \\

\noindent \textbf{Retirement.} At age $T+1$, the individuals retire for $T^{\text{retire}}$ periods, and then dies. During retirement, they only consume and can no longer work. 
For simplicity, we specify that every period they obtain the period utility of a part-time working individual with a median $\Eta_t = 0.5$ and without the additive shock $\Epsilon$. 
They obtain a pension set to $50\%$ of their last full-time equivalent income, $y_T$.  There is no bequest motive. The retirement problem has a closed form solution, easily solved for any parameters. \\

\noindent \textbf{Instrument validity.} Since this model enter the more general Framework, the previous labor choice, $D_{t-1}$, can be used as an instrument $W_t$ for identification here. Indeed, $W_t=D_{t-1}$ is excluded conditional on $D_t, X_t$ and $m$, and it will also be relevant provided that the switching costs $\omega^m \neq 0$. 


\begin{figure}[!p]

	\vspace{-0.5cm}
    \centering
    \captionsetup{justification=centering}
    \caption{Estimated policies, CCCs and CCPs}
    \label{fig:optipolicy}
    \vspace{-1em}

    \begin{subfigure}[b]{0.7\textwidth}
		\centering
        \includegraphics[width=1\textwidth]{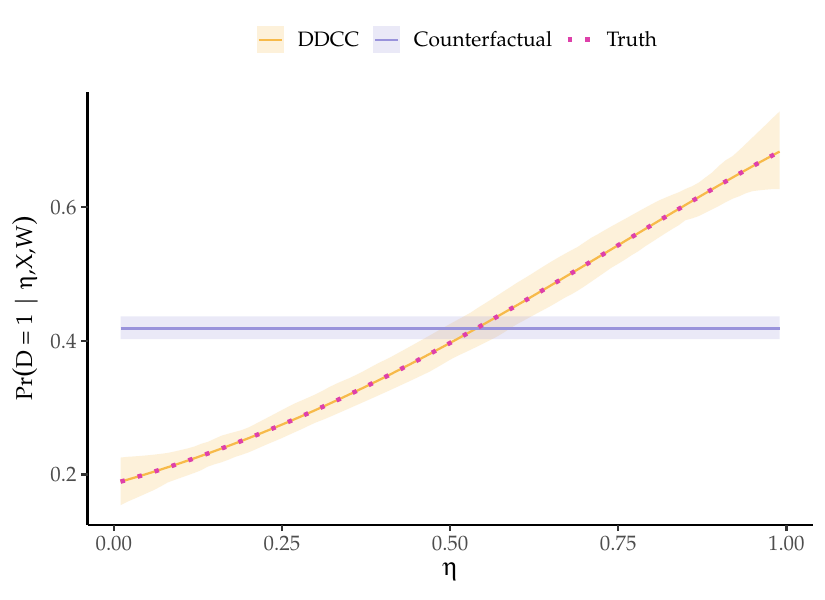}
        \caption{CCP, $\textrm{Pr}^m(D_t=1|\Eta_t, X_t, W_t)$}
        \label{fig:ccp}	
    \end{subfigure}


	\begin{adjustbox}{max width=1\textwidth,center}
    \begin{subfigure}[b]{0.5\textwidth}
        \centering
        \includegraphics[width=1\textwidth]{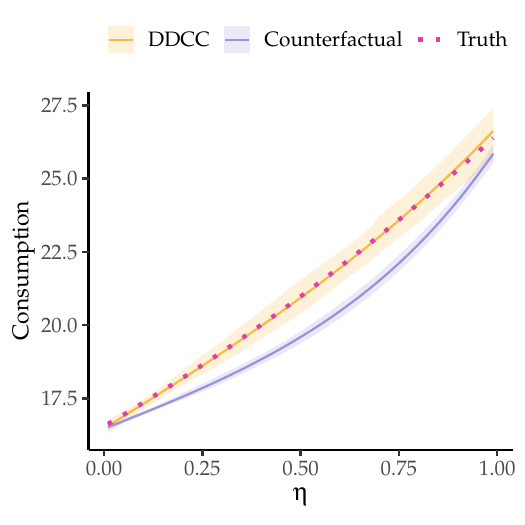}
        \caption{CCC, $c_{0t}^m(\Eta_t, X_t)$}
        \label{fig:c0}
    \end{subfigure}
    \hfill 
    \begin{subfigure}[b]{0.5\textwidth}
        \centering
        \includegraphics[width=1\textwidth]{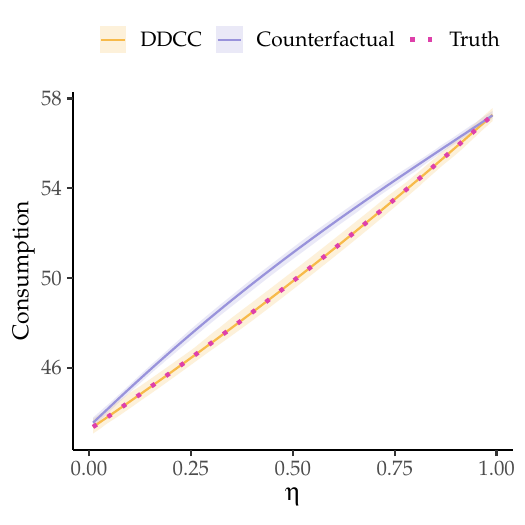}
        \caption{CCC, $c_{1t}^m(\Eta_t, X_t)$}
        \label{fig:c1}
    \end{subfigure}
    \end{adjustbox}

    \vspace{1em} 

   \captionsetup{font=footnotesize, justification=justified}
	\caption*{\textit{Notes:} The figures report the estimated policies (CCCs and CCPs) over $100$ Monte Carlo simulations of the model, with the specification described in Table \ref{table_comparison_results}. The DDCC estimates correspond to my estimator, the counterfactual estimates correspond to an estimation assuming that $D$ is exogenous (or equivalently, sequential choice of $D$ before $\Eta$ is realized). 
	We report the average estimate as well as the $[10\%, 90\%]$ interval of the estimates at each values of $\Eta$. For simplicity, we only report the results for a specific individual, with an arbitrary value of the type, age, covariates, and instrument. Here the results are reported for type $m=2$, $t=2$, and the covariates $X_t$ set to an income $=y_L = 30$ and asset holding of $30$ (thousands of dollars). For the CCPs, they are also reported for individuals who were working part time in the previous period, i.e., with $W_t = D_{t-1} = 0$. The estimated CCCs and CCPs may be more or less noisy for other choices of $X_t, W_t$, and $m$. 
	 }
\end{figure}

\newgeometry{top=0.5in, bottom=0.8in, left=1in, right=1in} 

\begin{table}[!p] \centering 
  \caption{Monte Carlo simulations of the life-cycle model}\label{table_comparison_results}
\scalebox{0.8}{
\begin{tabular}{lC{2.9cm}C{2.9cm}C{2.9cm}C{2.9cm}}	
\\[-1.8ex]\hline 
\hline \\[-1.8ex] 
 & & \multicolumn{3}{c}{\textit{Method}}  \\ 
\cline{3-5} 
& & & & \\ 
& Truth & \multicolumn{1}{c}{DDCC} &  \multicolumn{1}{c}{Counterfactual} & Known policies\\ 
& & & ($D$ exogenous) & \\
\\[-1.8ex] 
\hline \\[-1.8ex] 
\textbf{Euler parameters, $\theta_C$ and $\beta$.} \\
Risk aversion, $\sigma_d^m$: \\
$\sigma_0^1$ & 1.700 & 1.691 & 1.680 & 1.700 \\
& & (0.0325) & (0.0076) & (0.0005) \\
$\sigma_0^2$ & 2.000 & 2.027 & 2.061 & 2.000 \\
& & (0.0193) & (0.0129) & (0.0007) \\
$\sigma_1^1$ & 1.600 & 1.599 & 1.541 & 1.600 \\
& & (0.0334) & (0.0073) & (0.0004) \\
$\sigma_1^2$ & 1.500 & 1.514 & 1.473 & 1.500 \\
& & (0.0147) & (0.0094) & (0.0001) \\
Discount factor: $\beta$ & 0.950 & 0.948 & 0.961 & 0.950 \\
& & (0.0075) & (0.0025) & (0.0004) \\ 

\textbf{Probability parameters, $\theta_P$.} \\ 
Utility cost of labor, $\gamma^m$: \\
$\gamma^1$ & 0.000 & -0.026 & 0.960 & 0.007 \\
& & (0.1363) & (0.0283) & (0.0281) \\ 
$\gamma^2$ & -0.500 & -0.473 & 0.709 & -0.488 \\
& & (0.1033) & (0.0248) & (0.0210) \\
Effect of $\Eta$ on the cost, $\alpha^m$ \\
$\alpha^1$ & 2.000 & 1.985 & 0.025 & 1.988 \\
& & (0.2829) & (0.0135) & (0.0582) \\
$\alpha^2$ & 2.500 & 2.388 & 0.122 & 2.475 \\
& & (0.2089) & (0.0099) & (0.0443) \\
Tenure switching cost, $\omega^m$ \\
$\omega^1$ & -2.000 & -2.036 & -1.894 & -2.001 \\
& & (0.0660) & (0.0466) & (0.020) \\
$\omega^2$ & -2.200 & -2.199 & -1.997 & -2.198 \\
& & (0.0572) & (0.0361) & (0.0182) \\

\textbf{Type Probability.} \\
$\pi^1$ & 0.600 & 0.601 & 0.601 & Known \\
& & (0.0055) & (0.0055) & \\
  \hline \\[-1.8ex] 
  & \multicolumn{4}{c}{$N=10,000$ and $T=10$} \\
  \textbf{Estimation Time.} \\
  \textit{1st stage time} (Types + Policies) & & 143s (43 + 100) & 77s (43 + 34) & 0s (0+0) \\
  
  \textit{2nd stage time} (Euler + Probability) & & 45s (6+39) & 39s (5+34) & 82s (6+76) \\
  
  \textit{Total} & & 188s & 116s & 82s \\

\hline 
\hline \\[-1.8ex]

\end{tabular} 
}
\captionsetup{font=footnotesize, justification=justified}
    \caption*{\textit{Notes:} The results correspond to averages and standard deviations over $100$ Monte Carlo simulations. The DDCC column corresponds to my estimator, the counterfactual corresponds to an estimation assuming that $D$ is exogenous (or equivalently, sequential choice of $D$ before $\Eta$ is realized), and the last column corresponds to the estimation of the second stage given the true first stage optimal policies. \\
    \textit{Specification details.} 
    The income specification used is one with $y_L = 30$, $y_H = 60$, with income transition probabilities given by $\pi_{0L} = 0$, $\pi_{0H} = 0.3$, $\pi_{1L} = 0.5$, $\pi_{1H} = 0.75$. These transition probabilities are directly estimated from the data in each simulation. 
    For the initialization, in the first period, $\textrm{Pr}(W_1 = 1) = 0.5$, $\textrm{Pr}(\text{income}=y_H) = 0.5$, and the initial asset are drawn from a normal distribution with mean $40$ and standard deviation $10$, and minimum set to $0$ for asset draws below $0$. 
    $T=10$ periods and individuals live for $T^{retire}=5$ periods after retirement. Once they retire, they are left with a pension equal to $50\%$ of their last income every year. The yearly interest rate $r=1.5\%$. 
    In each simulation, we use $1,000$ one-period ahead forward simulation to approximate the next period expectation in the Euler equation, and $500$ forward life-cycle simulations to approximate the conditional value functions for the probability parameters. \\ 
    \textit{Moments used.} Individuals with asset $30$ (average) and all possible combinations of income $y_L$ or $y_H$, starting age $5$ or $8$, $\Eta$ equals $25\%, 50\%,$ or $75\%$, and $W$ equals $0$ or $1$. \\
    \textit{Computation details.} All the Monte Carlo simulations are ran in \textit{R} on a Macbook Pro M3 Max processor without parallelization.  }

\end{table}

\restoregeometry
\newpage

\subsection{Monte Carlo simulation results} 
To assess the performance of the estimator developed in this paper -- denoted DDCC for dynamic discrete-continuous choices -- I run Monte Carlo simulations of the life-cycle model described previously. Each simulation simulates a panel of $N=10,000$ individuals observed for their entire life-cycle of $T=10$ periods. This approximately corresponds to the sample size of real surveys used to estimate life-cycle models (e.g., the PSID), which allows to assess the performance of the estimator in a realistic context. Figure \ref{fig:optipolicy} shows the estimated optimal policies (CCCs and CCPs) in the first stage, for a given type $m$, covariates $X=x$ and instrument $W=w$. Table \ref{table_comparison_results} shows the corresponding estimates of the deep parameters of the model, $\theta$, using the first stage policies previously estimated. \\

\noindent \textbf{Estimation performances.} 
The DDCC estimates (orange in Figure \ref{fig:optipolicy}, column DDCC in Table \ref{table_comparison_results}) correspond to the estimator described in Section \ref{section_estimation}.\footnote{\textit{Specification details.}  For the DDCC method , the CCCs, $c_{0t}^m(\eta, x)$ and $c_{1t}^m(\eta, x)$, are estimated via a weighted (by estimated type-$m$ probabilities) IVQR of $C$ on $D$ given $X$, instrumented by $W$ \citep{chernozhukovhansen2006, kaido2021decentralization}. The continuous covariates (asset) enters linearly in the specification, while we run a separate IVQR on each subsample of income and age. We estimate the model at each $5\%$ percentiles, from $5\%$ to $95\%$. To extrapolate at the tails, we run a supplemental regression of a shape constrained additive model \citep[scam package in R,][]{pya2015shape} imposing monotonicity of $C$ with respect to a flexible spline of order $4$ in $\Eta$. From the estimated CCCs, we can recover $\hat{\Eta}$ for every observations. Then, the CCPs are estimated using a flexible weighted (by type-$m$ probabilities) logit regression of $D$ on $\hat{\Eta}$, $X$, and $W$. More precisely, we regress $D$ on a polynomial of order 3 in $\Eta$, on $W$, age and income dummies, and a polynomial of order 3 in assets, and also interactions of the polynomial in $\Eta$ with $W$. 
Then the primitives of the models are estimated as described in Section \ref{section_estimation} by taking these optimal policies as given in forward simulations of the models to approximate (i) the next period marginal utility expectation (Euler criterion) and (ii) the conditional value functions.  } 
With a balanced panel of $10,000$ individuals observed over $10$ periods, the types are well estimated in the initial EM algorithm and the optimal policies are precisely estimated without bias around the truth (Figure \ref{fig:optipolicy}). Then, the forward simulations approximate well the expected next period marginal utility (Euler criterion) and the conditional value functions (probability criterion), and as a consequence, all the parameters of the models are precisely estimated.\footnote{The probability parameters are more noisy because small changes in these parameters only induce small changes in the observable CCPs.} \\
\indent The `known policies' column in Table \ref{table_comparison_results} corresponds to an estimation of the second stage using the true first stage policies as if they were known. By comparing the DDCC results to these, we can separate the variance in the estimates caused by the second-stage forward simulation' approximations from the variance due to errors in the first-stage policies that carry over to the second stage. One can see that a large part of the variance is driven by the (unbiased) estimation of the first stage, even though some variance remains purely from the simulation process in the second stage (especially for the probability parameters). Note that the precision of the second stage can be increased by increasing the number of forward simulations, at the cost of increasing computation time as well.  \\

\noindent \textbf{Counterfactual exogeneity/sequentiality assumption.} The `Counterfactual' estimation (blue in Figure \ref{fig:optipolicy}) corresponds to an estimation obtained by taking the standard empirical approach of (wrongly) assuming the timing that $D$ is decided before $\Eta$ is realized (sequential choices), or, equivalently, that $D$ is exogenous with respect to $C$.  
As visible in Figure \ref{fig:optipolicy}, this exogeneity assumption leads to severely biased CCCs and CCPs estimates. This is because, by assuming that $D$ is independent of $\Eta$, one assumes that the CCPs are flat in $\Eta$ (see Figure \ref{fig:ccp}). Consequently, the CCCs, $c_{dt}^m(\eta, x)$, are wrongly estimated by the observed quantiles of $C$ in the subsample $D=d$ (given $X=x$ and type $m$), ignoring the fact that the distributions of $\Eta$ differs in the $D=1$ and $D=0$ subsamples. In other words, the main error with this counterfactual approach is that the CCCs are estimated via simple quantile regression (weighted by estimated type-$m$ probabilities), instead of using the proper (weighted) IV-quantile regression approach to correct for the endogeneity.\footnote{\textit{Specification details.} The CCCs are estimated via simple weighted (by estimated type-$m$ probabilities) quantile regression of $C$ on $D$ on each subsamples of income and age, and including the asset as a linear covariate in the regression. For the CCP, since $D$ is assumed exogenous, one simply estimate $\textrm{Pr}(D=1 | W, X, m)$ in the data using a weighted logit of $D$ on $W$, age and income dummies, and a polynomial of order $3$ in assets. } \\
\indent These wrong estimates of the first stage policies induce biased estimates of the structural parameters of the model. The bias is especially severe for parameters related to the probability criteria, but we also obtain wrong estimates of the risk aversion and discount factor. This is problematic and means that taking this wrong (but relatively standard) approach can lead to biased counterfactual policy analysis, and wrong policy recommendations. \\

\noindent \textbf{Computational performances.} As visible in Table \ref{table_comparison_results}, the estimation with the DDCC estimator of Section \ref{section_comparison_estimation} is very fast: it takes about $3$ minutes to estimate the entire model, $143$ seconds to estimate the first stage optimal policies (including the type probabilities) and $45$ seconds for the second stage primitive parameters. Note that these results were obtained using R, a popular language among economists but rarely used for structural estimation due to its slower performance compared to compiled languages like C or Fortran. By enabling rapid estimation in widespread languages like R or Python, the DDCC estimator lowers barriers, making structural modeling accessible for a broader range of applied researchers. \\
\indent To contextualize this performance, I tried to compare the DDCC estimator with the best alternative, i.e., state-of-the-art indirect inference estimation using the endogenous grid method \citep[EGM][]{carroll2006, iskhakov2017} to solve for the value functions and optimal policies. The problem is that the EGM estimation is orders of magnitude (at least 200 times) longer: the computation of a single value function of this life-cycle model takes about 8 minutes, and one needs to compute hundreds (if not thousands) of value functions to find the optimal parameters. As a consequence, one value function computation following \cite{iskhakov2017} is longer than the entire estimation with my method. 
\noindent Therefore, my method yields sizeable computational gains, even to the point where one can estimate models that would otherwise be considered intractable. This is because the two-step method, even though it introduces a fixed computational cost for the first stage optimal choices, drastically reduces the computational burden by avoiding the computation of value functions for each evaluated set of parameters in the second stage.
\footnote{Another advantage of the DDCC estimator with respect to indirect inference with the EGM is that I do not solve numerically for the optimal choices. As a consequence, I do not run into optimization problems and I do not need to smooth potential kinks introduced by the joint discrete-continuous choices, contrary to \cite{iskhakov2017} for example. }  
The more complicated the model, the larger the computational gains.  



%
%
%
%
%
%
%



\section{Conclusion}\label{section_conclusion}

This paper develops a general class of dynamic discrete-continuous choice models including a wide range of unobserved heterogeneity with transitory shocks and permanent unobserved types. I provide a constructive identification proof for this class of models. 
%
Given the identification, I provide a new estimation procedure yielding sizeable computational gains relative to existing alternatives for the estimation of dynamic models. The gains are so large that they should facilitate the practical use of complex dynamic discrete-continuous models in many fields (labor, housing, education, industrial organization, etc.) in the future. \\
\indent This discrete-continuous choice single-agent framework also adapts to stationary infinite horizon (dynamic) games with private information and unobserved market types. The adaptation of the framework is straightforward and similar to how the dynamic discrete choice framework of \cite{am2011} adapts to dynamic discrete games. See Appendix \ref{appendix_dynamic_games} for more details.   


\bibliographystyle{ecta}
{\begin{spacing}{0}
\setlength{\baselineskip}{0pt} 
\footnotesize
\bibliography{references}
\end{spacing}
}

\pagebreak

\pagenumbering{arabic}
\appendix

\noindent {\LARGE \textbf{Supplementary materials: online appendices}} 

\section{Identification with discrete number of alternatives}\label{appendix_discrete}

Assume that the discrete choice $D$ has support $\mathcal{D} = \{0, ..., J\}$, i.e.,  there are $J+1$ alternatives. In the main text I developed the reasoning with binary choice ($J=1$). I focus here on the general case where there are more than $2$ alternatives (i.e., $J \geq 2$). 
Let us assume that the instrument also has a support of (at least) $J+1$ values. This is typically true if $W_t = D_{t-1}$ as suggested. 
\noindent In this section, I only address the identification of the CCCs with $J \geq 2$, because the adjustments of the rest of the main text are straightforward. 
I impose the same assumptions as in Section \ref{section_framework}, with a generalized relevance condition. 
\begin{assumption}[Instrument Relevance $J \geq 2$]\label{relevance_J} The matrix 
\begin{align*}
		\tilde{M}(\eta) = \begin{bmatrix}
		\textrm{Pr}(D=0 | \Eta = \eta, W=0) & \cdots  & \textrm{Pr}(D=J | \Eta = \eta, W=0) \\
		\vdots &  & \vdots\\
		\textrm{Pr}(D=0 | \Eta = \eta, W=J) & \cdots & \textrm{Pr}(D=J | \Eta = \eta, W=J)
	\end{bmatrix} 
\end{align*} 
is such that $\textrm{det} \ \tilde{M}(\eta) \neq 0$ for all $\eta \in \mathcal{H} \backslash \mathcal{K}$, where $\mathcal{K}$ is a (possibly empty) finite set containing $K$ values ($K \geq 0$).
\end{assumption}

Assumption \ref{relevance_J} means that $\tilde{M}(\eta)$ is invertible for $\eta \in [0, 1]$, except possibly on a set of isolated noncritical singular values. In other words, if $\tilde{\eta}$ is a singular value with $\textrm{det}  \ \tilde{M}(\tilde{\eta}) = 0$, it is an isolated/noncritical one, i.e., $\frac{d(\textrm{det}  \ \tilde{M}) (\tilde{\eta})}{d\eta} \neq 0$. 
When $J=1$ (binary choice), Assumption \ref{relevance_J} corresponds to the relevance Assumption \ref{relevance}.\footnote{Indeed, the singular points are characterized by, $\textrm{det} \ \tilde{M}(\eta) = 0$, i.e.,  
\begin{align*}
	\frac{\textrm{Pr}(D=0 | \Eta = \eta, W=1)}{\textrm{Pr}(D=1 | \eta = h, W=1)} &= \frac{\textrm{Pr}(D=0 | \Eta = \eta, W=0)}{\textrm{Pr}(D=1 | \Eta = \eta, W=0)}.
\end{align*}
Since $\textrm{Pr}(D=0 | \Eta = \eta, W=w) = 1 - \textrm{Pr}(D=1 | \Eta = \eta, W=w)$, the condition is equivalent to
\begin{align*}
	\frac{\textrm{Pr}(D=0 | \Eta = \eta, W=1)}{1 - \textrm{Pr}(D=0 | \Eta = \eta, W=1)} &= \frac{\textrm{Pr}(D=0 | \Eta = \eta, W=0)}{1 - \textrm{Pr}(D=0 | \Eta = \eta, W=0)}. 
\end{align*}
Since $f(x) = x/(1-x)$ is strictly increasing on $[0, 1]$, the condition is equivalent to
\begin{align*}
	\textrm{Pr}(D=0 | \Eta = \eta, W=1) = \textrm{Pr}(D=0 | \Eta = \eta, W=0), 
\end{align*}
which is exactly the condition in Assumption \ref{relevance}.}  
Assumption \ref{relevance_J} is a generalization of Assumption \ref{relevance} in the general $J \geq 2$ discrete choice case. 
Notice that it is much weaker than usual full rank identification assumptions \citep[e.g.][]{chernozhukovhansen2005} on the effect of the instrument. $\tilde{M}(\eta)$ does not need  to be invertible for all $\eta \in [0, 1]$, it can have some singularities, as long as they are isolated. 

To identify the CCCs, we have a counterpart to system (\ref{bayes21}) with $J+1$ alternatives 
\begin{align}\label{bayes21_appendix}
\eta &= \sum^{J}_{d=0} \textrm{Pr}(C_d \leq c^*_d(\eta), D=d | W=w)  \quad \text{ for all } \eta \in [0,1] \text{ and } w \in \{0,..., J\}. 
\end{align}
%

\begin{theorem}[Identification with $J \geq 2$]\label{identification_theorem_J}
For every reduced form compatible with the structural model, there exist unique conditional continuous choice (CCC) functions $c_d(h)$ (for $d = 0, ..., J$) mapping $[0, 1]$ into $\mathcal{C}_d$, that are strictly increasing and satisfy
\begin{equation}\label{bayes2_appendix}
\eta = \sum^{J}_{d=0} \textrm{Pr}(C_d \leq c_d(\eta), D=d | W=w)  \quad \text{ for all } \eta \in [0,1] \text{ and } w \in \{0,..., J\}. 
\end{equation} 
\end{theorem}

\vspace{0.5\baselineskip}
\indent The conditional continuous choice (CCC) functions are identified if and only if there exist unique functions $c_d(\eta)$ strictly increasing with respect to $\eta$, which satisfy equation (\ref{bayes2_appendix}), and are compatible with the reduced form $R$. 
Thus, by Theorem \ref{identification_theorem}, the optimal CCCs, $c_d^*(\eta)$ ($d = 0,..., J$), are point identified from the reduced form $R$ as the unique strictly increasing solutions to (\ref{bayes2_appendix}). \\

\vspace{-0.5\baselineskip}

\begin{proof} 
\textit{Existence} of the solution is trivial: the reduced form is compatible with the structural model, so, by construction and as shown in equation (\ref{bayes21_appendix}), the  $c_d^*(\cdot)$ ($d = 0, ..., J$) solve system (\ref{bayes2_appendix}). 
\noindent To prove \textit{uniqueness} of the solution, let us derive system (\ref{bayes2_appendix}) with respect to $\eta$, for all $\eta \in [0, 1]$, 
\begin{align*}
	1 = \sum^{J}_{d=0} \frac{\partial c_d(\eta)}{\partial \eta} f_{D, C| w}(d, c_d(\eta))  \quad \text{ for all } w \in \{0,..., J\}, 
\end{align*}
where $f_{D, C| w}(d, c_d(\eta))$ are the joint densities of $D$ and $C$ conditional on $W=w$. 
This quasilinear implicit differential equation can be written under matrix form: 
\begin{align}\label{system_identification_appendix}
	&\begin{bmatrix}
		f_{D, C| 0}(0, c_0(\eta))  & \cdots & f_{D, C| 0}(J, c_J(\eta)) \\
		\vdots & & \vdots \\
		f_{D, C| J}(0, c_0(\eta))  & \cdots & f_{D, C| J}(J, c_J(\eta)) \\
	\end{bmatrix}
		\begin{bmatrix}
\partial c_0(\eta)/\partial \eta \\
\vdots \\
\partial c_J(\eta)/\partial \eta \\
	\end{bmatrix}
	 = 		\begin{bmatrix}
1 \\
\vdots \\
1 \\
	\end{bmatrix} \notag \\
	\overset{def}{\iff}&   \quad \quad \quad \quad \quad \quad \quad   M(\mathbf{c}(\eta)) \ \frac{\partial \mathbf{c}(\eta)}{\partial \eta} = \begin{bmatrix}
1 \\
\vdots \\
1 \\
	\end{bmatrix}, 
\end{align}
where $\mathbf{c}(\eta)$ is a notation for the vector of the $c_d(\eta)$ for all $d=0, ..., J$. \\
\indent Notice that, if $M$ is invertible, we can write the explicit differential equation 
\begin{align}\label{eq_appendix_invertible}
	\frac{\partial \mathbf{c}(\eta)}{\partial \eta} = M^{-1}(\mathbf{c}(\eta)) \ [1 \hdots 1]^T 
\end{align}
and this system of differential equations uniquely determines the derivatives $\partial c_d(\eta)/\partial \eta$ and thus the evolution of $c_d(\eta)$. 
Thus, if $M(\mathbf{c}^*(\eta))$ is invertible for all $\eta$ on the true optimal path, $\mathbf{c}^*(\eta)$, the true CCCs are the unique solution to the initial value problem starting from the known $\mathbf{c}^*(0)$ and solving system (\ref{system_identification_appendix}). 
Indeed, since $\mathbf{c}^*(\eta)$ is a solution to the system, and the derivatives are uniquely determined by (\ref{eq_appendix_invertible}), the uniquely determined derivatives correspond to $\partial \mathbf{c}^*(\eta)/\partial \eta$. Thus, starting from the known true $\mathbf{c}^*(0)$, we recover the entire true CCCs, $\mathbf{c}^*(\eta)$, for all $\eta \in [0, 1]$ by solving system (\ref{system_identification_appendix}). \\
\indent Now, notice that at the optimal CCCs, for all $d = 0, ..., J$ and $w = 0, ..., J$, 
\begin{align*}
f_{D, C| w}(d, c^*_d(\eta))  = \textrm{Pr}(D=d | \Eta = \eta, W=w) \  \partial {c_d^*}^{-1}(c_d)/\partial c_d.
\end{align*} 
Thus, on the optimal path of the true CCCs, we have 
\begin{align*}
	&M(\mathbf{c^*}(\eta)) =  \\ 
	&\underbrace{\begin{bmatrix}
		\textrm{Pr}(D=0 | \Eta = \eta, W=0) & \cdots  & \textrm{Pr}(D=J | \Eta = \eta, W=0) \\
		\vdots &  & \vdots\\
		\textrm{Pr}(D=0 | \Eta = \eta, W=J) & \cdots & \textrm{Pr}(D=J | \Eta = \eta, W=J)
	\end{bmatrix}}_{= \tilde{M}(\eta)}  \ \begin{bmatrix}
		\frac{\partial {c_0^*}^{-1}(c_0)}{\partial c_0} & & 0 \\
		& \ddots & \\
		0 & & \frac{\partial {c_J^*}^{-1}(c_J) }{\partial c_J} 
	\end{bmatrix}.
\end{align*}
Since $c_d^*(\eta)$ are strictly monotone functions of $\eta$, ${c_d^*}^{-1}(c_d)$ are also strictly increasing functions of $c_d$, thus $\partial {c_d^*}^{-1}(c_d) / \partial c_d$ are strictly positive for all $d$. So $M(\mathbf{c^*}(\eta)) $ is invertible if and only if $\tilde{M}(\eta)$ is invertible. \\ 
\indent \textit{Full rank case ($K=0$)}: under a strong version of Assumption \ref{relevance_J} with $K=0$, $\tilde{M}(\eta)$ is full rank and invertible, and so is $M(\mathbf{c^*}(\eta))$ for all $\eta \in [0, 1]$. In this case, the system of differential equation (\ref{system_identification_appendix}) uniquely identifies the true CCCs, $\mathbf{c^*}(h)$, as the solution to the initial value problem starting from the known $\mathbf{c}^*(0)$.\footnote{See \cite{bruneel2024identification} for a similar proof in the full rank case that does not directly rely on the system of differential equations. }  \\
\indent \textit{General case ($K \geq 0$)}: even when there are $K > 0$ isolated singularities on the optimal path, one can show there still exists a unique \textit{strictly increasing} solution to (\ref{bayes2_appendix}), using results from the literature on autonomous quasilinear implicit differential algebraic equations (DAEs) \citep[for e.g.,][]{rabier1989, sotomayor2001, marszalek2005, riaza2008}. 
\noindent First, let us rewrite (\ref{system_identification_appendix}) as an autonomous system with respect to $c$, 
\begin{align}\label{autonomous_system}
	M(\mathbf{c}) \mathbf{c'} = [1 \hdots 1]^T.
\end{align}
Let us denote $\mathbf{\tilde{c}}$ an isolated singularity. $\mathbf{\tilde{c}}$ is such that 
\begin{align*}
	\textrm{det} \ M(\mathbf{\tilde{c}}) = 0 &\text{ and } (\textrm{det}  \ M)' (\mathbf{\tilde{c}}) \neq 0.
\end{align*}
Obviously, in between the isolated singularities, $M(\mathbf{c})$ is invertible and we can proceed as in the previous case to solve the differential equation. 
At a singularity $\mathbf{\tilde{c}}$,  (\ref{autonomous_system}) is not invertible, we cannot write (\ref{eq_appendix_invertible}). Instead, define the \textit{canonical system}: 
\begin{align}\label{appendix_canonical}
	\textrm{det} \ M(\mathbf{c}) \mathbf{c}' = \textrm{adj} \ M(\mathbf{c}) \ [1 \hdots 1]^T. 
\end{align}
Take it at $\mathbf{\tilde{c}}$ to obtain: $\textrm{det} \ M(\mathbf{c}) \mathbf{c}' = \textrm{adj} \ M(\mathbf{c}) \ [1 \hdots 1]^T$.  
A solution to the canonical system (\ref{appendix_canonical}) also solves (\ref{autonomous_system}) and (\ref{system_identification_appendix}) given appropriate initial values. \\
\indent On the optimal path, under Assumption \ref{relevance_J}, we will only encounter $K$ ($\geq 0$) singular values $\tilde{\eta}^k$ for $k=1, ..., K$, with 
\begin{align*}
	\textrm{det} \ M(\mathbf{c}(\tilde{\eta}^k)) = 0 &\text{ and } \frac{d(\textrm{det}  \ M) (\mathbf{c}(\tilde{\eta}^k))}{d \eta} \neq 0.
\end{align*} 
They correspond to the $K$ singular values of the vector $\mathbf{\tilde{c}^k} = \mathbf{c^*}(\tilde{\eta}^k)$.  
These singularities on the optimal path are isolated \textit{geometric} singularities in the terminology of \cite{marszalek2005}. Indeed, $[1 \hdots 1]^T \in \textrm{image of} \ M(\mathbf{\tilde{c}})$, since the true optimal $\mathbf{{c^*}'}$ exist and satisfy (\ref{autonomous_system}) if $\mathbf{\tilde{c}}$ is on the optimal path. 
Now, we are only looking for a solution under the \textit{monotonicity constraint} that $\mathbf{c}' > 0$. The monotonicity constraint eliminates eventual multiplicity of the solutions to the canonical system (\ref{appendix_canonical}) at singular points. 
Indeed, following \cite{marszalek2005} (Theorem 1 and 2), there can only be one strictly increasing solution that goes smoothly through any geometric singularity $\mathbf{\tilde{c}^k}$ present on the optimal path. 
Thus, starting from the true initial value, $\mathbf{c}(0)$, there is a unique strictly increasing solution to  (\ref{bayes2_appendix}) even in the presence of isolated singularities on the optimal path. This solution is the true CCCs, $\mathbf{c^*}(\cdot)$.  \hfill \end{proof} 
%
%
%


\section{Proofs}
\subsection{Proof: Lemma \ref{distrib_c1}}\label{appendix_proof_lemma_distrib_c1}

\begin{proof}
The distribution of $\Eta$ is continuously differentiable and strictly increasing. Under Assumption \ref{ass_shocks}, $0 < \textrm{Pr}(d | \eta, w) < 1$. Thus, the distribution of $\Eta$ given $D=d$ and $W=w$ is also continuously differentiable and strictly increasing.  
Moreover, by the monotonicity Assumption \ref{monotone}, the distribution functions of $C_d$ (given $w$) are strictly monotone transformations of the distribution of $\Eta | d$, and we have
\begin{align*}
\underbrace{\textrm{Pr}( \Eta \leq \ \eta \ | D=d, W=w)}_{=F_{\Eta | d, w} (\eta)} \ = \ \underbrace{\textrm{Pr}(C_d \leq \ c_d^*(\eta) \ | D=d, W=w)}_{=F_{C_d | d, w} (c_d^*(\eta))}   \quad \forall d, w.
\end{align*}
Therefore, since $F_{\Eta |d, w}(\eta)$ is continuously differentiable and strictly increasing (with respect to $\eta$), $F_{C_d | d, w}(c_d^*(\eta))$ is also continuously differentiable and strictly increasing (with respect to $\eta$). Finally, since $c_d^*(\eta)$ are continuously differentiable and strictly increasing with respect to $\eta$ (Assumption \ref{monotone}), $F_{C_d | d, w}(c_d)$ are also continuously differentiable and strictly increasing with respect to $c_d$ for all $d$. 
\end{proof}

\subsection{Proof: Lemma \ref{difference_instru}}\label{appendixnonflat}

%
%
%
%

\begin{proof}
\noindent First, let us relate $\textrm{Pr}(D=d | \Eta=\eta, W=1) - \textrm{Pr}(D=d | \Eta=\eta, W=0)$ to the distributions/quantiles. We have for $\eta \in [0, 1]$ and $w=0, 1$: 
\begin{align}
\eta &= \quad \textrm{Pr}(\Eta \leq \eta) = \textrm{Pr}(\Eta \leq h | W=w) \nonumber \\
&= \quad \textrm{Pr}(\Eta \leq \eta \ | \  D = 0, W=w) \textrm{Pr}(D=0 | W=w) \nonumber \\
&\quad + \textrm{Pr}(\Eta \leq \eta \ | \  D = 1, W=w) \textrm{Pr}(D=1 | W=w) \nonumber \\
&= \quad F_{\Eta|D=0, w}(\eta) p_{D|w}(0) + F_{\Eta|D=1, w}(\eta) p_{D|w}(1). \label{bayes2eta}
\end{align}
Then, combining (\ref{bayes2eta}) at $w=1$ and $w=0$, we obtain for all $\eta$: 
\begin{align}
 &F_{\Eta |D=0, W=1}(\eta) p_{D|1}(0) - F_{\Eta |D=0, W=0}(\eta) p_{D|0}(0)  \nonumber \\
 = \  &- \ \Big( F_{\Eta | D=1, W=1}(\eta) p_{D|1}(1)  -  F_{\Eta | D=1, W=0}(\eta) p_{D|0}(1) \Big) \nonumber \\
\overset{def}{\iff}&  \quad \quad \Delta F_{\Eta_0}(\eta) \ = \ - \Delta F_{\Eta_1}(\eta). \label{diffweta}
\end{align}
Moreover, notice that we have
\begin{align}
F_{\Eta |D=d, w}(\eta) = \textrm{Pr}(D=d | \Eta \leq \eta, W=w) \textrm{Pr}(\Eta \leq \eta | W=w) / \textrm{Pr}(D=d | W=w).  \label{rewrittenFeta}
\end{align}

\noindent First focus on $D=0$ (symmetric reasoning for $D=1$) and plug (\ref{rewrittenFeta}) into (\ref{diffweta}):
\begin{align*}
\Delta F_{\Eta_0}(\eta) =& \ \textrm{Pr}(D=0 | \Eta \leq \eta, W=1) \textrm{Pr}(\Eta \leq \eta | W=1) \\
&-  \textrm{Pr}(D=0 | \Eta \leq \eta, W=0) \textrm{Pr}(\Eta \leq \eta | W=0).
\end{align*}
Moreover, since $\Eta \perp W$: $\textrm{Pr}(\Eta \leq \eta | W=1) = \textrm{Pr}(\Eta \leq \eta | W=0) = \textrm{Pr}(\Eta \leq \eta)$, thus
\begin{align}
\Delta F_{\Eta_0}(\eta) &= \big[ \ \textrm{Pr}(D=0 | \Eta \leq \eta, W=1) - \textrm{Pr}(D=0 | \Eta \leq \eta, W=0) \ \big] \ \textrm{Pr}(\Eta \leq \eta). \label{rewrittenFeta2}
\end{align}
\noindent Under the normalization $\Eta \sim \mathcal{U}(0,1)$,  $\textrm{Pr}(\Eta \leq \eta) = \eta$, and  
\begin{align*}
\textrm{Pr}(D=0 | \Eta \leq \eta, w) = \int^{\eta}_0 \textrm{Pr}(D=0 | \Eta = \tilde{\eta}, w) / \eta \ d\tilde{\eta}. 
\end{align*}
So we can rewrite (\ref{rewrittenFeta2}) for all $\eta$, as
\begin{align}
\Delta F_{\Eta_0}(\eta) 
& = \int^{\eta}_0 \Big(\textrm{Pr}(D=0 | \Eta = \tilde{\eta}, W=1) - \textrm{Pr}(D=0 | \Eta = \tilde{\eta}, W=0) \Big) d\tilde{\eta}. \label{rewrittenFeta3}
\end{align}
Thus, we relate $\Delta F_{\Eta_0}(\cdot)$ with the difference in CCPs present in Assumption \ref{identification2}: 
\begin{align*}
\frac{\partial \Delta F_{\Eta_0}(\eta)}{\partial\eta} = \textrm{Pr}(D=0 | \Eta = \eta, W=1) - \textrm{Pr}(D=0 | \Eta = \eta, W=0).
\end{align*}
\indent Now let us relate these with the observed $C_0$. The true CCCs, $c_d^*(\eta)$, are strictly monotone, thus $F_{\Eta |D=d, W=w}(h) = F_{C_d | D=d, W=w} (c_d^*(\eta))$ and
\begin{align*}
\Delta F_{\Eta_0} (\eta) &= F_{\Eta | D=0, W=1}(\eta) p_{D|1}(0) - F_{\Eta |D=0, W=0} (h) p_{D|0}(0) \\
&= F_{C_0 |D=0, W=1}(c_0^*(\eta))p_{D|1}(0) - F_{C_0 |D=0, W=0} (c_0^*(\eta)) p_{D|0}(0)  \overset{def}{=} \Delta F_{C_0}(c_0^*(\eta)).   
\end{align*}
%
Thus
\begin{align*}
 \frac{\partial \Delta F_{\Eta_0}(h)}{\partial h} &= \textrm{Pr}(D=0 | \Eta = \eta, W=1) - \textrm{Pr}(D=0 | \Eta = \eta, W=0)  \\
&= \frac{\partial \Delta F_{C_0} (c_0^*(\eta))}{\partial\eta} = \frac{\partial \Delta F_{C_0} (c_0)}{\partial c_0}   \underbrace{\frac{\partial c_0^*(\eta)} {\partial\eta}}_{ > 0}.
\end{align*}



\noindent So, if $\textrm{Pr}(D=0 | \Eta = \eta, W=1) - \textrm{Pr}(D=0 | \Eta = \eta, W=0) \neq 0$ (Assumption \ref{identification2}), then $\partial \Delta F_{C_0} (c_0)/ \partial c_0 \neq 0$ by strict monotonicity of $c_0^*(\eta)$. And reversely, at the  $K$ ($\geq 0$) isolated values of $\eta \in [0, 1]$ such that $\textrm{Pr}(D=0 | \Eta = \eta, W=1) = \textrm{Pr}(D=0 | \Eta = \eta, W=0)$, then there is also $K$ isolated values of $c_0 \in \mathcal{C}_0$ such that $\partial \Delta F_{C_0} (c_0)/ \partial c_0 = 0$. \\
%
%
\indent We can follow exactly the same reasoning for $D=1$. Using the fact that $\textrm{Pr}(D=0 | \Eta = \eta, W=w) = 1 - \textrm{Pr}(D=1 | \Eta = \eta, W=w)$ for all $\eta$ and $w=0,1$, we have that if $\textrm{Pr}(D=0 | \Eta = \eta, W=1) - \textrm{Pr}(D=0 | \Eta = \eta, W=0) \neq 0$ (Assumption \ref{identification2}), then $\partial \Delta F_{C_1} (c_1)/ \partial c_1 \neq 0$ by strict monotonicity of $c_1^*(\eta)$. And reversely, at the  $K$ ($\geq 0$) isolated values of $\eta \in [0, 1]$ such that $\textrm{Pr}(D=0 | \Eta = \eta, W=1) = \textrm{Pr}(D=0 | \Eta = \eta, W=0)$, then there is also $K$ isolated values of $c_1 \in \mathcal{C}_1$ such that $\partial \Delta F_{C_1} (c_1)/ \partial c_1 = 0$. 
Moreover, the $K$ values of $c_1$ such that $\partial \Delta F_{C_1} (c_1)/ \partial c_1 = 0$ corresponds to the same $h$ as the $K$ values of $c_0$ such that $\partial \Delta F_{C_0} (c_0)/ \partial c_0 = 0$. \hfill \end{proof}  

\subsection{Details of the proof of Theorem \ref{identification_theorem} when K $> 0$}\label{appendix_proof_identification}

There is a finite number $K < \infty$ of $c_0$ and $c_1$ such that $\partial \Delta F_{C_d}(c_d)/\partial c_d = 0$, denoted $c_0^1 < c_0^2 < ... < c_0^K$ and $c_1^1 < c_1^2 < ... < c_1^K$. These points corresponds to the same set of underlying $\eta^k$ and are thus the image of each other  (cf proof of Lemma \ref{difference_instru}).  
Thus, we necessarily have 
\begin{align*}
	\tilde{c_0}(c_1^k) = c_0^k \text{ for } k = 1, ..., K, 
\end{align*}
since otherwise $\tilde{c_0}(\cdot)$ would not be strictly increasing with respect to $c_1$. 

Now, we show that $\tilde{c_0}(c_1)$ is also unique in between the $c_1^k$. We use that the function $\Delta F_{C_d}(\cdot)$ are \textit{piecewise monotone and invertible} between the $K$ points of null derivative.
First, for $d=0,1$, split the compact set $\mathcal{C}_d$ of $c_d$ into $K+1$ sub-intervals $\mathcal{C}_d^k$: 
$\mathcal{C}_d^1 = [c_d^*(0), c_d^1], \ \mathcal{C}_d^2 = [c_d^1, c_d^2],\  ... \ , \  \mathcal{C}_d^{K+1} = [c_d^K, c_d^*(1)]$  such that $\mathcal{C}_d = \underset{k \in \{1, ..., K+1\}}{\cup} \mathcal{C}_d^k $, and where $c_d^*(0)$ and $c_d^*(1)$ are identified as the minimum and maximum observed $C_d$.   
Denote $\mathcal{S}_0^k$ and $\mathcal{S}_1^k$ the image of those subsets by $\Delta F_{C_0}(\cdot)$ and $- \Delta F_{C_1}(\cdot)$ respectively, i.e., $\Big((-1)^d \Delta F_{C_d}\Big): \mathcal{C}_d^k \rightarrow \mathcal{S}_d^k$. 
Let us show that, on each subintervals, the images $\mathcal{S}_0^k = \mathcal{S}_1^k$ correspond to each other.
$\Delta F_{C_d}(\cdot)$ are strictly monotone and invertible (since continuously differentiable) between the points of null derivative. Thus, $\mathcal{S}_d^k$ are compact sets, as image of compact sets by strictly monotone functions.
Moreover, at the global boundaries, $F_{C_d | d, w}(c_d^*(0)) = 0$ for $d=0,1$ and $w=0,1$, thus, $\Delta F_{C_d}(c_d^*(0)) = 0$ for $d = 0, 1$, 
and $F_{C_d | d, w}(c_d^*(1)) = 1$ for $d=0,1$ and $w=0,1$, thus, $\Delta F_{C_0}(c_0^*(1)) = p_{D|1}(0) - p_{D|0}(0) = (1 - p_{D|1}(1)) - (1 - p_{D|0}(1)) = - (p_{D|1}(1) - p_{D|0}(1)) = - \Delta F_{C_1}(c_1^*(1))$. 
For the interval boundaries, they are images of each other so we have, $\Delta F_{C_0}(c_0^k) = -\Delta F_{C_1}(c_1^k)$ for $k = 1, ..., K$. 
So, it implies that $\mathcal{S}_0^k = \mathcal{S}_1^k$ and we denote them $\mathcal{S}^k$ for all $k = 1, ..., K+1$. 
We have: $\mathcal{S}^0 = [0, -\Delta F_{C_1}(c_1^1)], \ \mathcal{S}^1 = [-\Delta F_{C_1}(c_1^1), -\Delta F_{C_1}(c_1^2) ], \ ..., \ \mathcal{S}^{K+1} = [-\Delta F_{C_1}(c_1^K), -\Delta F_{C_1}(c_1^*(1)) ]$. \\
\indent Now, we are looking for an increasing mapping solution to (\ref{mapping_sketch}). By \textit{monotonicity}, we know that for a solution $\tilde{c_0}(\cdot)$, we have $\tilde{c_0}: \mathcal{C}_1^k \rightarrow \mathcal{C}_0^k$ since the bounds of these sets are image of each other. 
On each subintervals $\mathcal{C}_d^k$, the corresponding functions $\Delta F_{C_d}(\cdot)$ are strictly monotone and continuously differentiable for $d=0, 1$. Moreover, they have the same image, $\Delta F_{C_0} : \mathcal{C}_0^k \rightarrow S^k$ and $(- \Delta F_{C_1}): \mathcal{C}_1^k \rightarrow S^k$.  So we can invert the $\Delta F_{C_0}(\cdot)$ segment by segment and get for any $k = 1, ..., K+1$: 
\begin{align*}
\tilde{c_0}(c_1) = ( \Delta F_{C_0} )^{-1} ( \Delta F_{C_1}(c_1) )   \text{ for all } c_1 \in \mathcal{C}_1^k.  
\end{align*}
This uniquely define the solution $\tilde{c_0}(c_1)$ on $\mathcal{C}_1^k$. 
We repeat it for all $k = 1, ..., K+1$, and obtain the unique mapping $\tilde{c_0}(c_1)$ solution to (\ref{mapping_sketch}) for all $c_1 \in \mathcal{C}_1$.


\section{Unobserved types}\label{appendix_types}

\indent At any time $t \in \{1, ..., T\}$, we observe the joint conditional densities $f_{D_t, C_t | x_t, d_{t-1}, t}(d_t, c_t)$, and by construction, we know that
\begin{align*}
f_{D_t, C_t | x_t, d_{t-1}, t}(d_t, c_t) = \sum^M_{m=1} \pi^m f^m_{D_t, C_t | x_t, d_{t-1}}(d_t, c_t)
\end{align*}
Over several periods, we have the joint density 
\begin{align}\label{eq_panel}
f(\{d_{t}, c_t, x_t\}_{t=1}^T) = \sum^M_{m=1} &\pi^m f^{m}_{D_1, C_1, X_1}(d_1, c_1, x_1)
 \prod^T_{t=2} f_{X_t| x_{t-1}, c_{t-1}, d_{t-1}}(x_t) f^m_{D_t, C_t | x_t, d_{t-1}}(d_t, c_t).
\end{align}
where $f^{m}_{D_1, C_1, X_1}(\cdot)$ is the density of the first period, when we do not observe past choices.  
Define $s_t = (d_t, c_{t}, x_t)$, $q^{*m}(s_1) = f^{m}_{D_1, C_1, X_1}(d_1, c_1, x_1)$, and 
\begin{align*}
Q^m_t(s_t | s_{t-1}) = f_{X_t| x_{t-1}, c_{t-1}, d_{t-1}}(x_t) \  f^m_{D_t, C_t | x_t, d_{t-1}, t}(d_t, c_t). 
\end{align*}
$s_t$ follows a first-order Markov process.\footnote{In fact, here, $f^m_{D_t, C_t | x_t, d_{t-1}, t}(d_t, c_t)$ does not depend on $c_{t-1}$, even though we write it generally in $Q^m_t(s_t | s_{t-1})$. But it could be identified even if it depended on $c_{t-1}$.} 
Rewrite equation (\ref{eq_panel}) as
\begin{equation}
f(\{s_t\}^T_{t=1}) = \sum^M_{m=1} \pi^m q^{*m}(s_1) \prod^T_{t=2} Q^m_t(s_t | s_{t-1}). \tag{\ref{eq_panel}'}
\end{equation}
Notice that $s_t$ appears in both $Q^m_t(s_t | s_{t-1})$ and $Q^m_{t+1}(s_{t+1} | s_t)$, which creates a dependence between these two terms. 
As in \cite{kasaharashimotsu2009}, I solve this dependence problem by using the first-order Markov property of $s_t$ and looking at \textit{every two periods} in order to break the dependence of $s_t$ across periods. \\
%
%
%
%
%
\indent Without loss of generality, first focus on even time periods. The time horizon is $T \geq 6$ and $T$ is even such that $T$ is the last \textit{even} period.\footnote{If $T$ is odd, just follow the same reasoning with $T-1$ (i.e., the last even period) instead of $T$.}  Fix a predefined \textbf{sequence} $\mathbf{\bar{s}} = \{ \bar{s}_1, \bar{s}_3, ..., \bar{s}_{T-1} \}$ for odd time periods. 
First focus on the identification of the joint densities given this predefined sequence. As a difference with \cite{kasaharashimotsu2009}, I allow for $\bar{s}_1 \neq \bar{s}_3 \neq ... \neq \bar{s}_{T-1}$. This modification is made to account for the presence of the asset in the covariates, because of which, a given $\bar{s}_{t+1}$ is only reachable by picking some $s_t$ given $\bar{s}_{t-1}$. I can pick any predefined sequence, as long as it is observable in the data, i.e., for any $t$, there must exist at least one $s_t$ such that $f(\bar{s}_{t+1}, s_t, \bar{s}_{t-1}) > 0$. Identification requires the existence of $M-1$ values of $s_t$ satisfying this condition. 
%
%
%
%
%
Conditional on this predefined sequence $\mathbf{\bar{s}}$, also define
\begin{align*}
\lambda_{t, \ \mathbf{\bar{s}}}^m (s_t) &= Q^m_{t+1}(\bar{s}_{t+1} | s_t) \  Q^m_t(s_t | \bar{s}_{t-1})  \quad \text{ for } t = 2, 4, ..., T-2 \\
\text{and } \quad \lambda_{T, \ \mathbf{\bar{s}}}^{*m} (s_T) &= Q^m_{T}(s_{T} | \bar{s}_{T-1}) \text{ for the last (even) period } T.
\end{align*}
Then we have that the observable $f(\{s_t\}_{t=1}^T | s_t = \bar{s}_t \text{ for t odd} )$ is 
\begin{equation}
f(\{s_t\}_{t=1}^T | s_t = \bar{s}_t \text{ for t odd} ) =  \sum^M_{m=1} \pi^m q^{*m}(s_1) \left( \prod_{t=2, 4, ..}^{T-2} \lambda^m_{t, \ \mathbf{\bar{s}}} (s_t) \right) \ \lambda_{T, \ \mathbf{\bar{s}}}^m (s_T).
\end{equation}
Let us define the sets $\mathcal{S}_t|\mathbf{\bar{s}}$ for $t=2,4,...,T$ as the sets of elements which are `compatible' with the predefined sequence $\mathbf{\bar{s}}$, i.e., 
\begin{align*}
\{ \xi \in \mathcal{S}_t|\mathbf{\bar{s}}: f(\bar{s}_{t+1}, s_t=\xi, \bar{s}_{t-1}) > 0 \}.
\end{align*}
Let $\{\xi_j^t\}_{t=2, 4, ..., T-2}$, for $j=1, ..., M-1$ be different elements of $\mathcal{S}_2|\mathbf{\bar{s}} \times \mathcal{S}_4|\mathbf{\bar{s}} \times ... \times \mathcal{S}_{T-2}|\mathbf{\bar{s}}$, let $\xi^T$ be an element of $\mathcal{S}_T|\mathbf{\bar{s}}$ and define 
\begin{align*}
\underset{(M \times M)}{L_{t, \ \mathbf{\bar{s}} }} &= \begin{bmatrix}
1 & \lambda^1_{t, \ \mathbf{\bar{s}}}(\xi^t_1) & \cdots & \lambda^1_{t, \ \mathbf{\bar{s}}}(\xi^t_{M-1}) \\
\vdots & \vdots & \ddots & \vdots \\
1 & \lambda^{M}_{t, \ \mathbf{\bar{s}}}(\xi^t_1) & \cdots & \lambda^{M}_{t, \ \mathbf{\bar{s}}}(\xi^t_{M-1}) \\
\end{bmatrix}, \ \\ 
\underset{(M \times M)}{D_{\xi^T, \ \mathbf{\bar{s}}}} &= \textrm{diag}(\lambda^{*1}_{T, \ \mathbf{\bar{s}}}(\xi^T), \hdots, \lambda^{*M}_{T, \ \mathbf{\bar{s}}}(\xi^T)) 
\text{ and } \quad \underset{(M \times M)}{V} = \textrm{diag}( \pi^1, ..., \pi^m). 
\end{align*}
The elements of $L_{t, \ \mathbf{\bar{s}} }$ and $D_{t, \ \mathbf{\bar{s}} }$ and $V$ are the parameters of the mixture models we want to identify. \\
\indent Some remarks about the `\textit{compatible sets}': 
first, notice that by Assumption \ref{type_support}, if $\xi \in \mathcal{S}_t|\mathbf{\bar{s}}$ then we also have $f^m(\bar{s}_{t+1}, s_t=\xi, \bar{s}_{t-1}) > 0$ for all $m$. Thus, $\lambda^m_{t, \ \mathbf{\bar{s}}}(\xi) > 0$. This is necessary for identification as we will require that $\{\xi_j^t\}$ are such that $L_{t, \ \mathbf{\bar{s}} }$ is nonsingular. Hence the focus on these `compatible sets'. 
Second, the introduction of these sets is only important if $f_t(x_t | x_{t-1}, d_{t-1}, c_{t-1}) = 0$ for some $(t, x_t, x_{t-1}, d_{t-1}, c_{t-1})$, i.e., if assumption 1(c) of \cite{kasaharashimotsu2009} is violated. It is the case in the example with consumption and labor choice if the assets are in the covariates. Indeed the asset has deterministic transition through a budget constraint: $\text{asset}_{t} = (1+r_{t-1}) \text{asset}_{t-1} - c_{t-1} + \text{income}_{t-1} d_{t-1}$. Now fix the sequence $\bar{\mathbf{s}}$. For any even $t$, the fixed $\bar{s}_{t-1}$ yields an unique value of $\tilde{\text{asset}}_t$ for the asset at time $t$. Any $s_t$ with an asset different from $\tilde{\text{asset}}_t$ is not compatible with the predefined sequence. The sequence $\mathbf{\bar{s}}$ uniquely determines the assets for all time periods $\{\text{asset}\}_{t=1}^T$. 
So, in order to identify $f^m_{D_t, C_t | x_t, d_{t-1}, t}(d_t, c_t)$ conditional on different asset values, one needs to adjust the sequence $\mathbf{\bar{s}}$.  
Moreover, for a given $(\text{income}_t, r_t, d_t)$, $c_t$ can only take one value: $c_t = \bar{\text{asset}}_{t+1} - (1+r_t) \tilde{\text{asset}_t} - d_t \text{income}_t$. Thus the set of values $s_t$ can take given a fixed sequence $\mathbf{\bar{s}}$ is quite limited. To identify the joint densities for all $c_{dt} \in \mathcal{C}_{dt}$ given a fixed $\bar{s}_{t-1}$, one needs to adjust $\bar{s}_{t+1}$ in the sequence so that $c_t$ can adjust. 
Finally notice that if $M$ is large, it is useful to have covariates $x^1$ which do not enter the budget constraint and for which $f_t(x^1_t | x_{t-1}, d_{t-1}, c_{t-1}) > 0$ for all $(t, x^1_t, x_{t-1}, d_{t-1}, c_{t-1})$ in order to have more elements in $\mathcal{S}_t|\mathbf{\bar{s}}$ from which to pick $M-1$ elements $\xi$ from. \\
\indent Now, let us define notation for what is observable in the data, given the fixed sequence $\mathbf{\bar{s}}$. If $t$ is even, using the first order markov property of $s$, we define
\allowdisplaybreaks
\begin{align*}
F^t_{\mathbf{\bar{s}}}(s_t) &= f(s_t, s_{t+1} | s_\tau = \bar{s}_\tau \text{ for } \tau \text{ odd}) = \sum_{m=1}^M \pi^m \lambda^m_{t, \ \mathbf{\bar{s}}} (s_t), \\
\text{and for } T, \quad F^{*T}_{\mathbf{\bar{s}}}(s_T) &= f(s_T | s_\tau = \bar{s}_\tau \text{ for } \tau \text{ odd}) = \sum_{m=1}^M \pi^m \lambda^{*m}_{T, \  \mathbf{\bar{s}}}(s_T). 
\end{align*}
Similarly, for transitions over two (non consecutive) even periods, define
\begin{align*}
F^{t, t+2}_{\mathbf{\bar{s}}}(s_t, s_{t+2}) = f(s_t, s_{t+1}, s_{t+2}, s_{t+3} | s_\tau = \bar{s}_\tau \text{ for } \tau \text{ odd}) &= \sum^M_{m=1} \pi^m \lambda^m_{t+2, \ \mathbf{\bar{s}}} (s_{t+2})\  \lambda^m_{t, \ \mathbf{\bar{s}}} (s_{t}), \\
\text{Including $T$, } F^{*t, T}_{\mathbf{\bar{s}}}(s_t, s_T) = f(s_t, s_{t+1}, s_{T} | s_\tau = \bar{s}_\tau \text{ for } \tau \text{ odd}) &= \sum^M_{m=1} \pi^m \lambda^m_{t, \ \mathbf{\bar{s}}} (s_{t})\  \lambda^{*m}_{T, \ \mathbf{\bar{s}}} (s_{T}).
\end{align*}
Finally, for the transitions over three (non consecutive) periods, including $T$, define
\begin{align*}
F^{*t, t+2, T}_{\mathbf{\bar{s}}}(s_t, s_{t+2}, s_T) &= f(s_t, s_{t+1}, s_{t+2}, s_{t+3}, s_T | s_\tau = \bar{s}_\tau \text{ for } \tau \text{ odd}) \\
 &= \sum^M_{m=1} \pi^m \lambda^m_{t, \ \mathbf{\bar{s}}} (s_{t}) \ \lambda^m_{t+2, \ \mathbf{\bar{s}}} (s_{t+2}) \ \lambda^{*m}_{T, \ \mathbf{\bar{s}}} (s_{T}),
\end{align*}
which is equal to $f(\{ s_t \}_{t=2}^T | s_\tau = \bar{s}_\tau \text{ for } \tau \text{ odd})$ when $T=6$. \\
\indent So, the observable probabilities $F$, can be related with type-dependent parameters of the mixture models we want to identify. Let us evaluate these marginals at combinations of selected $s_t = \{\xi^t_j\}_{t, t+2}$ for $j = 1, ..., M-1$, and $\{\xi^T\}$. Arrange them into two $M \times M$ matrices: \\
\begin{align*}
\underset{(M \times M)}{P_{t, \ \mathbf{\bar{s}} }} &= 
\begin{bmatrix}
1 & F^{t+2}_{\mathbf{\bar{s}}}(\xi^{t+2}_1) & \cdots & F^{t+2}_{\mathbf{\bar{s}}}(\xi^{t+2}_{M-1}) \\
& & & \\
F^{t}_{\mathbf{\bar{s}}}(\xi^{t}_1) & F^{t, t+2}_{\mathbf{\bar{s}}}(\xi^t_1, \xi^{t+2}_1) & \cdots & F^{t, t+2}_{\mathbf{\bar{s}}}(\xi^t_1, \xi^{t+2}_{M-1}) \\
& & & \\
\vdots & \vdots & \ddots & \vdots \\
& & & \\
F^{t}_{\mathbf{\bar{s}}}(\xi^{t}_{M-1}) & F^{t, t+2}_{\mathbf{\bar{s}}}(\xi^t_{M-1}, \xi^{t+2}_1) & \cdots & F^{t, t+2}_{\mathbf{\bar{s}}}(\xi^t_{M-1}, \xi^{t+2}_{M-1}) \\
\end{bmatrix}, \\
\\
\underset{(M \times M)}{P^*_{t,\  \xi^T, \ \mathbf{\bar{s}} }} &= 
\begin{bmatrix}
F^{*T}_{\mathbf{\bar{s}}}(\xi^T) & F^{*t+2, T}_{\mathbf{\bar{s}}}(\xi^{t+2}_1, \xi^T) & \cdots & F^{*t+2, T}_{\mathbf{\bar{s}}}(\xi^{t+2}_{M-1}, \xi^T) \\
& & & \\
F^{*t, T}_{\mathbf{\bar{s}}}(\xi^{t}_1, \xi^T) & F^{*t, t+2, T}_{\mathbf{\bar{s}}}(\xi^t_1, \xi^{t+2}_1, \xi^T) & \cdots & F^{*t, t+2, T}_{\mathbf{\bar{s}}}(\xi^t_1, \xi^{t+2}_{M-1}, \xi^T) \\
& & & \\
\vdots & \vdots & \ddots & \vdots \\
& & & \\
F^{*t, T}_{\mathbf{\bar{s}}}(\xi^{t}_{M-1}, \xi^T) & F^{*t, t+2, T}_{\mathbf{\bar{s}}}(\xi^t_{M-1}, \xi^{t+2}_1, \xi^T) & \cdots & F^{*t, t+2, T}_{\mathbf{\bar{s}}}(\xi^t_{M-1}, \xi^{t+2}_{M-1}, \xi^T) \\
\end{bmatrix}. 
\end{align*}
Now, since $\sum^M_{m=1} \pi^m = 1$, we obtain \textit{factorization equations} for any even $t \leq T-4$: 
\begin{equation}
P_{t, \mathbf{\bar{s}} } = L'_{t,\mathbf{\bar{s}} } \ V \ L_{t+2, \mathbf{\bar{s}} }, \quad \text{ and } P^*_{t, \ \xi^T, \mathbf{\bar{s}} } = L'_{t, \mathbf{\bar{s}} } \ D_{\xi^T, \mathbf{\bar{s}} } \ V \ L_{t+2, \mathbf{\bar{s}} }.
\end{equation}
We use these known relation to identify $\{\lambda^m_{t, \ \mathbf{\bar{s}}}(\xi^t_j)\}_{t=2, 4, ..., T-2}$ for $j=1,.., M-1$ and for all $m = 1,..., M$ and $\lambda^{*m}_{T, \ \mathbf{\bar{s}}}(\xi^t_j)$ for all $m = 1,..., M$. 

\begin{proposition}[Identification given $\mathbf{\bar{s}}$ for even $t$]\label{proposition_even_sbar} Suppose that $s_t$ follows a first-order Markov process and assume $T \geq 6$ and is even. Fix an observable predefined sequence $\mathbf{\bar{s}}$ for odd $t$ periods. For $t$ even, let $\xi^t_j$, $j=1,..., M-1$ be elements of $\mathcal{S}_t | \mathbf{\bar{s}}$. 
Suppose $q^{*m}(\bar{s}_1) > 0$ for all $m$, and for any even $t$, there exists  $M-1$ elements $\{\xi^t_j\}_{j=1,..., M-1}$ of $\mathcal{S}_t | \mathbf{\bar{s}}$,  such that $L_{t, \ \mathbf{\bar{s}} }$ is nonsingular for Suppose also that there exists $\xi^T \in \mathcal{S}_T|\mathbf{\bar{s}}$ such that $\lambda^{*m}_{T, \ \mathbf{\bar{s}}}(\xi^T) > 0$ for all $m$ and $\lambda^{*m}_{T, \ \mathbf{\bar{s}}}(\xi^T) \neq \lambda^{*n}_{T, \ \mathbf{\bar{s}}}(\xi^T)$ for any $m \neq n$. 
Then $\{\pi^m, \{ \lambda^{*m}_{T, \ \mathbf{\bar{s}}}(\xi)\}_{\xi \in \mathcal{S}_T | \mathbf{\bar{s}}}, \{\{ \lambda^{m}_{t, \ \mathbf{\bar{s}}}(\xi^t)\}_{\xi^t \in \mathcal{S}_t | \mathbf{\bar{s}}} \}_{t=2, 4, ..., T} \}^M_{m=1}$ is uniquely determined from $\{ f(\{s_t\}^T_{t=1} | s_t = \bar{s}_t \text{ for t odd}): \{s_t\}_{t=2, 4, ..., T} \in \mathcal{S}_2|\mathbf{\bar{s}} \times \mathcal{S}_4|\mathbf{\bar{s}} \times ... \times \mathcal{S}_T|\mathbf{\bar{s}} \}$. 
\end{proposition}


\begin{proof} 
Using the factorization equations, we can identify the parameters following \cite{kasaharashimotsu2009}. Let us consider an even $t \leq T-4$.  
First, $L_{t, \ \mathbf{\bar{s}} }$ and $L_{t+2, \ \mathbf{\bar{s}} }$ are nonsingular, so we can define: 
\begin{align*}
A_{\xi^T} := P_{t, \ \mathbf{\bar{s}} }^{-1} P^*_{t, \ \xi^T, \ \mathbf{\bar{s}} } = L_{t+2, \ \mathbf{\bar{s}} }^{-1} D_{\xi^T, \mathbf{\bar{s}}} L_{t+2, \ \mathbf{\bar{s}} }
\end{align*}
Thus, we can proceed to the eigen decomposition of $A_{\xi^T}$. The eigenvalues of $A_{\xi^T}$ gives the diagonal elements of $D_{\xi^T, \mathbf{\bar{s}}}$ (up to an arbitrary ordering of the types). The eigenvectors of $A_{\xi^T}$ determine the columns of $L_{t+2, \ \mathbf{\bar{s}} }^{-1}$ up to multiplicative constants. Denote these eigenvectors by  $L_{t+2, \ \mathbf{\bar{s}} }^{-1} K$ where $K$ is some diagonal matrix. Then, we can determine $VK$ from the first row of $P_{t, \ \mathbf{\bar{s}} } L_{t+2, \ \mathbf{\bar{s}} }^{-1} K $, since $P_{t, \ \mathbf{\bar{s}} } L_{t+2, \ \mathbf{\bar{s}} }^{-1} K = L'_{t, \ \mathbf{\bar{s}} } VK$ and the first row of $L'_{t, \ \mathbf{\bar{s}} }$ is a vector of ones. 
From here, we can uniquely identify $L'_{t, \ \mathbf{\bar{s}}}$ (and thus $L_{t, \ \mathbf{\bar{s}}}$) as $L'_{t, \ \mathbf{\bar{s}}} = (L'_{t, \ \mathbf{\bar{s}}} V K) (VK)^{-1}$. 
Then we can determine $V$ and $L_{t+2, \ \mathbf{\bar{s}}}$ from $VL_{t+2, \ \mathbf{\bar{s}}} = (L'_{t, \ \mathbf{\bar{s}}})^{-1} P_{t, \ \mathbf{\bar{s}}}$ since the first column of $VL_{t+2, \ \mathbf{\bar{s}}}$ equals the diagonal of $V$, and thus, $L_{t+2, \ \mathbf{\bar{s}}} = V^{-1} (VL_{t+2, \ \mathbf{\bar{s}}})$. 
Thus, we identified $D_{\xi^T, \mathbf{\bar{s}}}$, $V$, $L_{t, \ \mathbf{\bar{s}}}$ and $L_{t+2, \ \mathbf{\bar{s}}}$ and all their elements:
\begin{align*}
\Big\{\{\pi^m\}, \ \{ \lambda^{*m}_{T, \ \mathbf{\bar{s}}}(\xi^T)\}, \ \{ \lambda^{m}_{t, \ \mathbf{\bar{s}}}(\xi_j^t)\}_{j=1, ..., M-1}^{t=2, 4, ..., T} \Big\}_{m=1}^M. 
\end{align*}

Now we can also identify $\lambda$ for elements of $\mathcal{S}|\mathbf{\bar{s}}$ different from the ones we selected. 
First, we can identify $\{ \lambda^{*m}_{T, \ \mathbf{\bar{s}}}(\zeta) \}_{m=1}^{M}$ for any $\zeta \in \mathcal{S}_T | \mathbf{\bar{s}}$. Define  

\begin{align*}
\underset{(M \times M)}{D_{\zeta, \ \mathbf{\bar{s}}}} = \begin{bmatrix}
\lambda^{*1}_{T, \ \mathbf{\bar{s}}}(\zeta) & & 0 \\ 
& \ddots & \\
0 & & \lambda^{*M}_{T, \ \mathbf{\bar{s}}}(\zeta)
\end{bmatrix}, 
\end{align*}
and construct $P^*_{t, \ \zeta, \ \mathbf{\bar{s}} }$ the same way we constructed $P^*_{t, \  \xi^T, \ \mathbf{\bar{s}} }$, i.e.,
\begin{align*}
\underset{(M \times M)}{P^*_{t, \ \zeta, \ \mathbf{\bar{s}} }} = 
\begin{bmatrix}
F^{*T}_{\mathbf{\bar{s}}}(\zeta) & F^{*t+2, T}_{\mathbf{\bar{s}}}(\xi^{t+2}_1, \zeta) & \cdots & F^{*t+2, T}_{\mathbf{\bar{s}}}(\xi^{t+2}_{M-1}, \zeta) \\
& & & \\
F^{*t, T}_{\mathbf{\bar{s}}}(\xi^{t}_1, \zeta) & F^{*t, t+2, T}_{\mathbf{\bar{s}}}(\xi^t_1, \xi^{t+2}_1, \zeta) & \cdots & F^{*t, t+2, T}_{\mathbf{\bar{s}}}(\xi^t_1, \xi^{t+2}_{M-1}, \zeta) \\
& & & \\
\vdots & \vdots & \ddots & \vdots \\
& & & \\
F^{*t, T}_{\mathbf{\bar{s}}}(\xi^{t}_{M-1}, \zeta) & F^{*t, t+2, T}_{\mathbf{\bar{s}}}(\xi^t_{M-1}, \xi^{t+2}_1, \zeta) & \cdots & F^{*t, t+2, T}_{\mathbf{\bar{s}}}(\xi^t_{M-1}, \xi^{t+2}_{M-1}, \zeta) \\
\end{bmatrix}. 
\end{align*}
We identify the elements of $D_{\zeta, \ \mathbf{\bar{s}}}$ using 
\begin{align*}
D_{\zeta, \ \mathbf{\bar{s}}} = (L'_{t, \ \mathbf{\bar{s}} } V)^{-1} P^*_{t, \ \zeta, \ \mathbf{\bar{s}} } (L_{t+2, \ \mathbf{\bar{s}}})^{-1}.
\end{align*}

\indent Similarly, for $t = 2, 4, ..., T-2$, we identify $\{ \lambda^{m}_{t, \ \mathbf{\bar{s}}}(\zeta) \}_{m=1}^{M}$ for any $\zeta \in \mathcal{S}_t | \mathbf{\bar{s}}$. Define
\begin{align*}
\underset{(M \times 2)}{L_{t, \ \mathbf{\bar{s}} }^\zeta} = \begin{bmatrix}
1 & \lambda^1_{t, \ \mathbf{\bar{s}}}(\zeta)  \\
\vdots & \vdots  \\
1 & \lambda^{M}_{t, \ \mathbf{\bar{s}}}(\zeta)  \\
\end{bmatrix}.
\end{align*}
Then we can construct $P_{t, \  \zeta, \ \mathbf{\bar{s}} } = (L^\zeta_{t, \ \mathbf{\bar{s}}})' V L_{t+2, \ \mathbf{\bar{s}}}$, where $P_{t, \  \zeta, \ \mathbf{\bar{s}} }$ is observable from the data. Thus we can identify $L^\zeta_{t, \ \mathbf{\bar{s}}}$ for any $t = 2, 4, ..., T-2$ as:\footnote{For  $t=T-2$, one can construct the `next period' $L_{T, \ \mathbf{\bar{s}}}$ using identified elements of $D_{\zeta, \ \mathbf{\bar{s}}}$. Or alternatively, one can just build $L_{T-2, \ \mathbf{\bar{s}} }^\zeta$ to replace $L_{t+2, \ \mathbf{\bar{s}} }$ and not $L_{t, \ \mathbf{\bar{s}} }$ in the factorization equation. }
\begin{align*}
(L^\zeta_{t, \ \mathbf{\bar{s}}})' = P_{t, \  \zeta, \ \mathbf{\bar{s}} } (V L_{t+2, \ \mathbf{\bar{s}}})^{-1}.
\end{align*}

Thus, we identified
\begin{align*}
\Big\{\{\pi^m\}, \ \{ \lambda^{*m}_{T, \ \mathbf{\bar{s}}}(\zeta), \forall \zeta \in \mathcal{S}_T | \mathbf{\bar{s}} \}, \ \{ \lambda^{m}_{t, \ \mathbf{\bar{s}}}(\zeta^t) , \forall \zeta^t \in \mathcal{S}_t | \mathbf{\bar{s}}\}_{t=2, 4, ..., T-2} \Big\}_{m=1}^M.  
\end{align*} 
This completes the identification proof when $t$ is even, conditional on a given $\mathbf{\bar{s}}$. \hfill \end{proof}

\noindent \textbf{Identification for any sequence $\mathbf{\bar{s}}$ when t is even: }\\
Proposition \ref{proposition_even_sbar} provides identification of $\lambda$ conditional on a specific sequence $\mathbf{\bar{s}}$. 
We would like to identify $\lambda_{t}^m (s_{t-1}, s_{t}, s_{t+1})$ for all $(s_{t-1}, s_t, s_{t+1}) \in \mathcal{S}_{t-1} \times \mathcal{S}_t \times \mathcal{S}_{t+1}$ and  $\lambda_{t}^{*m} (s_{T-1}, s_{T})$ for all $(s_{T-1}, s_T) \in \mathcal{S}_{T-1} \times \mathcal{S}_T$, for all $m$, where 
\begin{align*}
\lambda_{t}^m (s_{t-1}, s_{t}, s_{t+1}) &= Q^m_{t+1}(s_{t+1} | s_t) Q^m_t(s_t | s_{t-1})  \quad \text{ for } t = 2, 4, ..., T-2 \\
\text{and } \quad \quad \quad \lambda_{T}^{*m} (s_{T-1}, s_T) &= Q^m_{T}(s_{T} | s_{T-1}) \quad \quad \text{ for the last (even) period } T.
\end{align*}

First, notice that we know that $\lambda_{t}^m (s_{t-1}, s_{t}, s_{t+1}) = 0$ for all $m$ for non existing transition patterns, i.e., $\text{if we observe }  f(s_{t+1}, s_t, s_{t-1}) = 0 \text{ then } \lambda_{t}^m (s_{t-1}, s_{t}, s_{t+1}) = 0 \ \text{for all } m$. And similarly, $\lambda_{T}^{*m} (s_{T-1}, s_T) = 0$ for all $m$ if $f(s_{T+1}, s_T) = 0$. With the presence of the asset in the variables $x_t$, there are many impossible transitions. \\
\indent For any other observable combination with non-zero transition probability, we will apply Proposition \ref{proposition_even_sbar} to identify the $\lambda$s. A sequence $\mathbf{\bar{s}}$ is `observable' if $f(s_t = \mathbf{\bar{s}}$ for $t$ odd$) > 0$. 
Let us assume that the transition pattern is \textit{sufficiently heterogenous across different types}, and that there exists covariates $x^1$ not included in the budget constraint for which $f_{X_t^1 | x_{t-1}, d_{t-1}, c_{t-1}}(x^1_t) > 0$ for all $(t, x^1_t, x_{t-1}, d_{t-1}, c_{t-1})$, with the number of elements in $\mathcal{X}^1_t$ is largely greater than $M$ for all $t$.\footnote{The second condition is only necessary when $M$ is large. If $M=2$ for example, I only need $M-1 = 1$ element $\xi \in \mathcal{S}_t|\mathbf{\bar{s}}$. I can always find such an element even if I only have covariates entering the budget constraint. Indeed, even if the income and the interest rate only take one value (giving no possible variations), just take $c_{t}^d = \text{asset}_{t+1} - (1+r) \text{asset}_t - d_t \text{income}_t$ $= \text{asset}_{t+1} - (1+r) ((1+r) \text{asset}_{t-1} - c_{t-1} - d_{t-1} \text{income}_{t-1})  - d_t \text{income}_t$. This gives (at most) two possible values for $(c_t, d_t)$: $(c_t^0, d_t=0)$ and $(c_t^1, d_t=1)$. \mysingleq{At most}, because depending on the value of the assets, income and interest rate, one of the two computed consumption may be negative. However since the transition is observed in the data, I know that at least one of these two consumptions will be positive, giving me an existing path. Now, if $M$ is large, and the support of the income and interest rate contains only a small finite number of elements, I need other covariates not entering the budget constraint in order to find $M-1$ different elements compatible with $\mathbf{\bar{s}}$.  } 
In this case, for any observable $\mathbf{\bar{s}}$, we can find $M-1$ elements $\{\xi^t_j\}_{t=2, 4, ..., T-2}$ such that $L_{t, \mathbf{\bar{s}}}$ is nonsingular. Similarly, for any observable $\mathbf{\bar{s}}$, we can find $\xi^T \in \mathcal{S}_T|\mathbf{\bar{s}}$ such that $\lambda^{*m}_{T, \ \mathbf{\bar{s}}}(\xi^T) > 0$ for all $m$ and $\lambda^{*m}_{T, \ \mathbf{\bar{s}}}(\xi^T) \neq \lambda^{*n}_{T, \ \mathbf{\bar{s}}}(\xi^T)$ for any $m \neq n$. 
Thus we can apply Proposition \ref{proposition_even_sbar} and identify $\Big\{\{\pi^m\}, \ \{ \lambda^{*m}_{T, \ \mathbf{\bar{s}}}(\zeta), \forall \zeta \in \mathcal{S}_T | \mathbf{\bar{s}} \}, \ \{ \lambda^{m}_{t, \ \mathbf{\bar{s}}}(\zeta^t) , \forall \zeta^t \in \mathcal{S}_t | \mathbf{\bar{s}}\}_{t=2, 4, ..., T-2} \Big\}_{m=1}^M$ for any observable $\mathbf{\bar{s}}$. \\
\indent Now, across several $\mathbf{\bar{s}}$, the types are identified up to an arbitrary order. Assume $\pi^m$ is different for all types (e.g., $\pi^n \neq \pi^k$ if $k \neq n$). In this case, since one identifies $\pi^m$ for each $\mathbf{\bar{s}}$ with Proposition \ref{proposition_even_sbar}, one can match the $\lambda^m_{t | \mathbf{\bar{s}}}$ to their respective type $m$ across different values of $\mathbf{\bar{s}}$.  
\noindent Thus, we cover the space of all possible transition patterns, and identify
\begin{align*}
\Big\{ &\{\pi^m\}, \ \{ \lambda^{*m}_{T}(s_{T-1}, s_T), \forall (s_{T-1}, s_T) \in \mathcal{S}_{T-1} \times \mathcal{S}_T  \}, \\
 \ &\{ \lambda^{m}_{t}(s_{t-1}, s_t, s_{t+1}) , \forall (s_{t-1}, s_t, s_{t+1}) \in \mathcal{S}_{t-1} \times \mathcal{S}_{t} \times \mathcal{S}_{t+1}\}_{t=2, 4, ..., T-2} \Big\}_{m=1}^M. 
\end{align*}

\noindent \textbf{Identification for odd time periods:} \\
We have identified $\lambda^{m}_{t}(\cdot)$ for even time periods. One can just proceed exactly the same way to identify the $\lambda^{m}_{t}(\cdot)$ for odd time periods. 
Just notice that now, we focus on the last \textit{odd} time period, i.e., $T-1 \geq 5$ if we assumed $T$ even.
Thus in the matrix $D$ we focus on $\lambda_{T-1, \ \mathbf{\bar{s}}}^{*m} (s_T) = Q^m_{T-1}(s_{T-1} | \bar{s}_{T-2})$ for the last (odd) period $T-1$. And we do not consider the transition from $T-1$ to the last period $T$. 
And for $t=1$, $Q^m_1(s_1 | s_{0})$ is undefined, so we replace it by the initial distribution $Q^m_1(s_1) := q^{*m}(s_1)$. 
The rest of the demonstration is straightforward, by replacing \textit{even} with \textit{odd} time periods and by fixing predefined sequences for even time periods $\mathbf{\bar{s}} = \{\bar{s}_2, \bar{s}_4, ..., \bar{s}_{T-2}\}$ in the previous development. I skip the complete development to simplify the exposition. \\

\noindent \textbf{Identification for all time periods:} \\
Trivially, if instead we assumed that $T \geq 6$ is odd, the identification with odd periods will identify up to $T$ while the identification of even periods would identify $\lambda$ up to $T-1$. 
Therefore, in any case ($T$ even or odd), if $T \geq 6$ we identify 
\begin{align*}
\Big\{ &\{\pi^m\}, \ 
\{ \lambda^{*m}_{T}(s_{T-1}, s_{T}), \forall (s_{T-1}, s_{T}) \in \mathcal{S}_{T-1} \times \mathcal{S}_{T}  \}, \\
&\{ \lambda^{*m}_{T-1}(s_{T-2}, s_{T-1}), \forall (s_{T-2}, s_{T-1}) \in \mathcal{S}_{T-2} \times \mathcal{S}_{T-1}  \}, \\
 \ &\{ \lambda^{m}_{t}(s_{t-1}, s_t, s_{t+1}) , \forall (s_{t-1}, s_t, s_{t+1}) \in \mathcal{S}_{t-1} \times \mathcal{S}_{t} \times \mathcal{S}_{t+1}\}_{t=1, 2, ..., T-2} \Big\}_{m=1}^M. 
\end{align*}

\indent Now we want to identify the $Q^m_t(s_t | s_{t-1})$ separately for all $t, s_t, s_{t-1}$. For $t \in \{T-1, T\}$, $Q^m_t(s_t | s_{t-1})$ are identified directly as they are equal to $\lambda^{*m}_T(s_{T-1}, s_T)$ and $\lambda^{*m}_{T-1}(s_{T-2}, s_{T-1})$. For $t \leq T-2$, we identified $\lambda_{t}^m (s_{t-1}, s_{t}, s_{t+1}) = Q^m_{t+1}(s_{t+1} | s_t) Q^m_t(s_t | s_{t-1})$. Thus, if we know $Q^m_{t+1}(s_{t+1} | s_t)$, then we identify $Q^m_t(s_t | s_{t-1}) = \lambda_{t}^m (s_{t-1}, s_{t}, s_{t+1})/Q^m_{t+1}(s_{t+1} | s_t)$. Thus, given that we know $Q^m_{T-1}(\cdot)$, we can proceed backwards to identify recursively $Q^m_t()$ for all $t \leq T-2$.  \\
\indent Moreover, recall that $Q^m_t(s_t | s_{t-1}) = f_{X_t | x_{t-1}, c_{t-1}, d_{t-1}}(x_t) f^m_{D_t, C_t | x_t, d_{t-1}}(d_{t}, c_t)$, and that the covariates transition density is type-independent and identified directly from the data. As a consequence, since $Q^m_t()$ are identified for all $t$, we also identify the type-dependent conditional choice joint densities. To conclude, we identify the type-dependent conditional choice joint densities and the type probabilities, $\text{ for all } m \in \{1, ..., M\}$, i.e.,  
\begin{align*}
f^m_{D_t, C_t | x_t, d_{t-1}}(d_t, c_t), \ (d_t, c_t, x_t, d_{t-1}) \in  \mathcal{D}_t  \times \mathcal{C}_{dt|x_t} \times \mathcal{X}_t \times \mathcal{D}_{t-1} \times \{2, ..., T\}   \  \text{ and } \pi^m. 
\end{align*}

\section{Dynamic games}\label{appendix_dynamic_games}
This Appendix describes how the dynamic discrete-continuous choice single-agent Framework of the main text can be extended to dynamic discrete-continuous games. \\
\indent There are $I$ firms in each of many markets. The payoffs of the $i$th firm depends on its own choices $(d^{(i)}_t, c_{d^{(i)}t}^{(i)})$, but also the choices of the other firms in the market $d_t^{(-i)} := (d_t^{(1)}, ..., d_t^{(i-1)}, d_t^{(i+1)}, ..., d_t^{(I)})$ and $c_{d^{(-i)}t}^{(-i)} := (c_{d^{(1)}t}^{(1)}, ..., c_{d^{(i-1)}t}^{(i-1)}, c_{d^{(i+1)}t}^{(i+1)}, ..., c_{d^{(I)}t}^{(I)})$. 
The payoff also depends on state variables $z_t^{(i)} = (x_t^{(i)}, w_t^{(i)})$. The instrument for firm $i$ is again $w_t^{(i)} = d_{t-1}^{(i)}$. 
The covariate $x_t^{(i)}$ includes variables $\tilde{x}_t$ which are common to all the firms, or firm-specific but observed by all the firms in the market. For notational simplicity, it also includes the \textit{market type} $m$, which is observed by all the firms.\footnote{One can have firms specific types $m^{(i)}$ taking values into $\mathcal{M} = \{1, ..., M\}$. As long as each firm on the market knows the types of the others, the common knowledge for everyone is $m=(m^{(1)}, ..., m^{(I)}) \in \mathcal{M}^I$, which is equivalent to having one unobserved type $m$ with the $M^I$ different possibilities. I do not allow for firm-specific private information types.} 
For firm $i$, $x_t$ also includes the instrument of the other firms, i.e., $w_t^{(-i)} = d_{t-1}^{(-i)}$. Thus, for firm $i$, the covariates are $x_t^{(i)} = (\tilde{x}_t, d_{t-1}^{(-i)}, m)$. %
Assume the environment is \textit{stationary}, with an infinite horizon, as is often assumed in the games literature. 
The current utility of firm $i$ in period $t$ when it picks $d_t^{(i)}$ is
\begin{align*}
	&\mathcal{U}_{d^{(i)}_t}^{(i)}\Big(c_{d^{(i)}t}^{(i)}, \ x_t^{(i)}, w_t^{(i)}, d_t^{(-i)}, c_{d^{(-i)}t}^{(-i)}, \eta^{(i)}_t,  \epsilon^{(i)}_t\Big)  \\ 
	& = \ u_{d^{(i)}_t}^{(i)}\Big(c_{d^{(i)}t}^{(i)}, x_t^{(i)}, d_t^{(-i)}, c_{d^{(-i)}t}^{(-i)}, \eta^{(i)}_t\Big) + m_{d^{(i)}_t}^{(i)}\Big(x_t^{(i)}, w_t^{(i)}, d_t^{(-i)}, c_{d^{(-i)}t}^{(-i)}, \eta^{(i)}_t\Big) + \epsilon^{(i)}_{d^{(i)}t}.
\end{align*}
$\eta_t^{(i)}$ and $\epsilon_t^{(i)}$ are identically and independently distributed shocks which are private information to the firm. As before, $\eta_t^{(i)}$ is nonseparable shock impacting the continuous and discrete choices of firm $i$, and $\epsilon_t^{(i)}$ are additively separable shocks which only impact the discrete choice $d^{(i)}_t$ of firm $i$. 
The setup satisfies the assumptions described before (monotonicity, independence from the instrument, independence between the shocks, ...), but applied to firm $i$ specific variables ($d_t^{(i)}, c_{d^{(i)}_t}^{(i)}, \eta_t^{(i)}, \epsilon_t^{(i)}, ...)$. 
Notice the utility functions can be firm specific, meaning that the same variable might have a different impact on different firms. For example, a characteristic of firm $i$ may affect the payoff of firm $i$ differently than a characteristic of firm $i'$. \\
\indent Firms make their decisions simultaneously in each periods. 
The main difficulty with games is that firm $i$'s payoff depends on the other firms choices, which are not observed when firm $i$ makes its own choices. So firm $i$ need to form beliefs about its competitors behaviour. 
Denote $L(d_t^{(-i)}, c_{d^{(-i)}t}^{(-i)} | X_t^{(i)}=x_t^{(i)})$, the likelihood that firm $i$'s competitors select $(d_t^{(-i)}, c_{d^{(-i)}t}^{(-i)})$ given $X_t^{(i)}=x_t^{(i)}$. This likelihood is time-independent since we consider stationary environment with infinite horizon.  
Since $(\eta_t^{(i)}, \epsilon_t^{(i)})$ are independently distributed across firms, we have 
\begin{align*}
	L(d_t^{(-i)}, c_{d^{(-i)}t}^{(-i)} | X_t^{(i)}=x_t^{(i)}) = \prod^I_{\substack{i'=1\\ i' \neq i}} L^{(i')}\big(d_t^{(i')}, c_{d^{(i')}t}^{(i')} | X_t^{(i')}=x_t^{(i')}\big),
\end{align*}
where $L^{(i')}(d_t^{(i')}, c_{d^{(i')}t}^{(i')} | X_t^{(i)}=x_t^{(i)})$ is the likelihood that firm $i'$ selects $D_t^{(i')} = d_t^{(i')}$ and $C_{t}^{(i')} = c_{d^{(i')}t}^{(i')}$ given $X_t^{(i')}=x_t^{(i')}$. 
\noindent I impose \textit{rational expectations} on firms' beliefs about their competitors' choices and assume firms are playing stationary Markov-perfect equilibrium strategies. Hence, the true densities match the beliefs of the firm. 
Firm $i$ uses its rational beliefs about the others choices in order to form expectations about the payoff it will obtain. In other words, firm $i$ makes its decision with respect to
\begin{align*}
\tilde{u}_{d^{(i)}_t}^{(i)}\Big(c_{d^{(i)}t}^{(i)}, x_t^{(i)}, \eta^{(i)}_t\Big) + \tilde{m}_{d^{(i)}_t}^{(i)}\Big(x_t^{(i)}, w_t^{(i)}, \eta^{(i)}_t\Big) + \epsilon^{(i)}_{d^{(i)}t}, 
\end{align*}
where we define 
\begin{flalign*}
	\tilde{u}_{d^{(i)}_t}^{(i)}\Big(c_{d^{(i)}t}^{(i)}, x_t^{(i)}, \eta^{(i)}_t\Big) &= \mathbb{E}_{D_t^{(-i)}, C_{D^{(-i)}t}^{(-i)}} \Big[ u_{d^{(i)}_t}^{(i)}\Big(c_{d^{(i)}t}^{(i)}, X_t^{(i)}, D_t^{(-i)}, C_{D^{(-i)}t}^{(-i)}, \eta^{(i)}_t\Big) \Big| X_t^{(i)}=x_t^{(i)} \Big] &\\
	\text{and } \quad \quad \tilde{m}_{d^{(i)}_t}^{(i)}\Big(x_t^{(i)}, w_t^{(i)}, \eta^{(i)}_t \Big) &=\mathbb{E}_{D_t^{(-i)}, C_{D^{(-i)}t}^{(-i)}} \Big[ m_{d^{(i)}_t}^{(i)}\Big(X_t^{(i)}, w_t^{(i)}, D_t^{(-i)}, C_{D^{(-i)}t}^{(-i)}, \eta^{(i)}_t\Big) \Big| X_t^{(i)}=x_t^{(i)} \Big], &
\end{flalign*} 
and where the expectations are computed as 
\begin{align*}
	&\mathbb{E}_{D_t^{(-i)}, C_{D^{(-i)}t}^{(-i)}} \Big[ u(\cdot, D_t^{(-i)}, C_{D^{(-i)}t}^{(-i)}) \Big| X_t^{(i)}=x_t^{(i)} \Big] \\
	&= \sum_{d_t^{(-i)} \in \mathcal{D}^{I-1}} \int_{c_{d^{(-i)}t}^{(-i)} \in \  \mathcal{C}_{d^{(-i)}t}^{(-i)}} L(d_t^{(-i)}, c_{d^{(-i)}t}^{(-i)} | X_t^{(i)}=x_t^{(i)}) \  u(\cdot, d_t^{(-i)}, c_{d^{(-i)}t}^{(-i)}), 
\end{align*}
with $\mathcal{C}_{d^{(-i)}t}^{(-i)} = \mathcal{C}_{d^{(1)}t}^{(1)} \times ...\times \mathcal{C}_{d^{(i-1)}t}^{(i-1)} \times \mathcal{C}_{d^{(i+1)}t}^{(i+1)} \times ... \times \mathcal{C}_{d^{(I)}t}^{(I)}$, where $\mathcal{C}_{d^{(k)}t}^{(k)}$ is the support of $C_{d^{(k)}}^{(k)}$ and can be firm-specific. 
%

\indent Similarly, the transition of the state variables will depend on the firm $i$'s choices but also on its competitors choices, i.e., 
\begin{align*}
f_{Z_{t+1}^{(i)} | x_t^{(i)}, d_t^{(i)}, c_{d^{(i)}t}^{(i)}, 	d_t^{(-i)}, c_{d^{(-i)}t}^{(-i)}}\big(z_{t+1}^{(i)}\big).
\end{align*}
The market types are time-invariant, $\mu_t = m$ for all $t$. 
The transitions are independent of the current instrument $w_{t}^{(i)}$ (conditional on $d_t^{(i)}$) by exclusion of the instrument from the transition(Assumption \ref{instru_transi}).
The instrument $w_{t+1}^{(i)} = d_t^{(i)}$, so its transition is known given the current discrete choice. 
The other firms instruments, $d_t^{(-i)}$, are also known directly given the current discrete choices of all firms. 
So really, the transitions that matters are the transitions of $\tilde{X}_t$ to $\tilde{X}_{t+1}$. 
When making its discrete-continuous choice, firm $i$ takes into account an expectation of the transition of $\tilde{X}_t$ with respect to its competitors choices, i.e., 
\begin{align*}
&f^{(i)}_{\tilde{X}_{t+1} | x_t^{(i)}, d_t^{(i)}, c_{d^{(i)}t}^{(i)}}\big(\tilde{x}_{t+1}\big) = \mathbb{E}_{D_t^{(-i)}, C_{D^{(-i)}t}^{(-i)}} \Big[ f_{\tilde{X}_{t+1} | x_t^{(i)}, d_t^{(i)}, c_{d^{(i)}t}^{(i)}, 	D_t^{(-i)}, C_{D^{(-i)}t}^{(-i)}}\big(\tilde{x}_{t+1}\big) \Big| X_t^{(i)}=x_t^{(i)}   \Big]. 
\end{align*}
%
%
Thus, firm $i$ expected covariate transition with respect to its competitors choices is
\begin{align*}
	f^{(i)}_{Z_{t+1}^{(i)} | x_t^{(i)}, d_t^{(i)}, c_{d^{(i)}t}^{(i)}}(z_{t+1}^{(i)}) = f^{(i)}_{\tilde{X}_{t+1} | x_t^{(i)}, d_t^{(i)}, c_{d^{(i)}t}^{(i)}}(\tilde{x}_{t+1}) \times \mathds{1}\{\mu_{t+1} = m_t\} \times \mathds{1}\{W_{t+1}^{(i)} = d_{t}^{(i)}\}.
\end{align*}

\vspace{3mm}
\indent Once we have firm $i$ rational expected payoffs and transitions with respect to the others behaviour, we can proceed exactly as we did in the case of \textit{dynamic single-agent models} (Section \ref{section_dynamic}). 
Knowing the transition densities, firm $i$ chooses $(d_t^{(i)}, c_{d^{(i)}t}^{(i)})$ to sequentially maximize its expected discounted sum of payoffs. 
Let $V^{(i)}(z_t^{(i)})$ be the (ex ante) value function of this discounted sum of future payoffs at the beginning of the period, just before the private information shocks $(\epsilon_t^{(i)}, \eta_t^{(i)})$ are revealed to firm $i$ and conditional on behaving according to the optimal decision rule. We have
\begin{align*}
V^{(i)}(z_t^{(i)}) \equiv \mathbb{E}\Bigg[ \sum^T_{\tau = t} \beta^{\tau-t} \underset{D^{(i)}_\tau, C_{D^{(i)}\tau}^{(i)}}{\textrm{max}} \Big[\tilde{u}_{D^{(i)}_\tau}^{(i)}\Big(C_{D^{(i)}\tau}^{(i)}, X_\tau^{(i)}, \Eta^{(i)}_\tau\Big) + \tilde{m}_{D^{(i)}_\tau}^{(i)}\Big(X_\tau^{(i)}, W_\tau^{(i)}, \Eta^{(i)}_\tau\Big) + \Epsilon^{(i)}_{D^{(i)}\tau}  \Big] \Bigg].\footnotemark 
\end{align*}
\footnotetext{In the stationary environment with infinite horizon, the value function $V$ is the same for all $t$ and we could write the sum starting from $0$ to $\infty$ instead.}  
This ex ante value function can be written recursively:
\begin{align*}
V^{(i)}(z_t^{(i)})  = \mathbb{E}_{\Epsilon^{(i)}, \Eta^{(i)}} \bigg[ \  \underset{d_t^{(i)}, c_{d^{(i)}t}^{(i)}}{\textrm{max}} \Big[&\tilde{u}_{d^{(i)}_t}^{(i)}\Big(c_{d^{(i)}t}^{(i)}, x_t^{(i)}, \Eta^{(i)}_t\Big) + \tilde{m}_{d^{(i)}_t}^{(i)}\Big(x_t^{(i)}, w_t^{(i)}, \Eta^{(i)}_t\Big) + \Epsilon^{(i)}_{d^{(i)}t}   \\
 &\ + \beta \mathbb{E}_{Z_{t+1}^{(i)}}[V^{(i)}(Z_{t+1}^{(i)}) \Big| X_t^{(i)}=x_t^{(i)}, C_t^{(i)}= c_{d^{(i)}t}^{(i)}, D_t^{(i)}=d_t^{(i)}] \Big] \  \bigg], 
\end{align*}
where the expectation about the next period value function is computed using firm $i$ expected transitions given its rational beliefs $f^{(i)}_{Z_{t+1}^{(i)} | x_t^{(i)}, d_t^{(i)}, c_{d^{(i)}t}^{(i)}}(z_{t+1}^{(i)})$. 
Thus, each period, after observing $(\Epsilon_t^{(i)}, \Eta_t^{(i)}) = (\epsilon_t^{(i)}, \eta_t^{(i)})$, firm $i$ chooses $d_t^{(i)}$ and $c_{d^{(i)}t}^{(i)}$ to maximize
\begin{align*}
	\underset{d_t^{(i)}, c_{d^{(i)}t}^{(i)}}{\textrm{max}} \ \tilde{u}_{d^{(i)}_t}^{(i)}\Big(c_{d^{(i)}t}^{(i)}, x_t^{(i)}, \eta^{(i)}_t\Big) &+ \beta \mathbb{E}_{Z_{t+1}^{(i)}}\Big[ V^{(i)}(Z_{t+1}^{(i)}) \Big| X_t^{(i)}=x_t^{(i)}, C_t^{(i)}= c_{d^{(i)}t}^{(i)}, D_t^{(i)}=d_t^{(i)} \Big]  \\
	 &+ \tilde{m}_{d^{(i)}_t}^{(i)}\Big(x_t^{(i)}, w_t^{(i)}, \eta^{(i)}_t\Big) + \epsilon^{(i)}_{d^{(i)}t}.
\end{align*}
Denote the conditional value functions of firm $i$ as
\begin{align*}
	v_{d^{(i)}_t}^{(i)}\Big(c_{d^{(i)}t}^{(i)}, x_t^{(i)}, \eta^{(i)}_t\Big) =& \  \tilde{u}_{d^{(i)}_t}^{(i)}\Big(c_{d^{(i)}t}^{(i)}, x_t^{(i)}, \eta^{(i)}_t\Big) \\
	&+ \beta \mathbb{E}_{Z_{t+1}^{(i)}}\Big[ V^{(i)}(Z_{t+1}^{(i)}) \Big| X_t^{(i)}=x_t^{(i)}, C_t^{(i)}= c_{d^{(i)}t}^{(i)}, D_t^{(i)}=d_t^{(i)} \Big]. 
\end{align*}
Now, the dynamic games yields the same maximization problem as in the general framework of Section \ref{section_framework}. Every period, firm $i$ selects $d_t^{(i)}$ and $c_{d^{(i)}t}^{(i)}$ to solve: 
\begin{align*}
	\underset{d_t^{(i)}, c_{d^{(i)}t}^{(i)}}{\textrm{max}} \ v_{d^{(i)}_t}^{(i)}\Big(c_{d^{(i)}t}^{(i)}, x_t^{(i)}, \eta^{(i)}_t\Big) + \tilde{m}_{d^{(i)}_t}^{(i)}\Big(x_t^{(i)}, w_t^{(i)}, \eta^{(i)}_t\Big) + \epsilon^{(i)}_{d^{(i)}t}.
\end{align*}

\indent Now, for the identification of the dynamic game, the market unobserved type $m$ are identified provided $T \geq 6$ following Section \ref{section_types} and pooling all the market observations. 
Once the types are identified, firm $i$ type-dependent stationary CCCs and CCPs are identified following Section \ref{section_identification}. \\
%
%
%

\noindent \textbf{Example:} Price and (discrete) quality choice. \\
There are $I$ firms that we observe over many periods ($T \geq 6$) in each of many markets. These firms compete by choosing a high ($d^{(i)}_t = 1$) or low quality ($d^{(i)}_t = 0$) for their products, and by choosing the corresponding conditional prices $c_{d^{(i)}t}^{(i)}$. They do so taking into account observed (by the econometrician) market characteristics $\tilde{x}_t$ (e.g., observed price of the inputs to produce the good), as well as time-invariant market specific type $m$, unobserved by the econometrician. $m$ could for example represent whether the market has an intrinsically high or low demand. Switching the quality of their product is costly, so $w_t^{(i)} = d_{t-1}^{(i)}$ is a relevant instrument. Conditional on the current quality choice, the past quality is not relevant when picking the price so $d_{t-1}^{(i)}$ is also excluded from the CCCs. Thus, the previous quality choice is a good instrument and each firm $i$ also takes into account observations about its competitors previous quality choice $d_{t-1}^{(-i)}$ when making its own choices. Not because it directly affects its own current payoff: conditional on $d_{t}^{(-i)}$, $d_{t-1}^{(-i)}$ has no impact on firm $i$ payoff at time $t$. But because $d_{t-1}^{(-i)}$ is important for firm $i$ to build rational expectations about the other firms choices probabilities today. Finally they also take into consideration $\eta^{(i)}_t$ and $\epsilon^{(i)}_t$ which are firm-specific temporary shocks. $\epsilon^{(i)}_t$ only impacts the quality choice while $\eta^{(i)}_t$ impacts quality and price decisions.   \\

\end{document}